\documentclass[article]{article}
\usepackage{color}
\usepackage{graphicx}
\usepackage{amsmath}
\usepackage{amssymb}
\usepackage{mathrsfs}
\usepackage[numbers,sort&compress]{natbib}
\usepackage{amsthm}
\usepackage{extarrows}
\usepackage{tikz}
\usepackage{float}
\usepackage{appendix}
\usepackage{hyperref}
\hypersetup{hypertex=true,
            colorlinks=true,
            linkcolor=blue,
            anchorcolor=blue,
            citecolor=red}
\usepackage{indentfirst}
\newtheorem{theorem}{Theorem}[section]
\newtheorem{lemma}{Lemma}[section]

\newtheorem{proposition}{Proposition}[section]
\newtheorem{Assumption}{Assumption}[section]
\newtheorem{remark}{Remark}[section]
\newtheorem{RHP}{RH problem}
\DeclareMathOperator*{\im}{Im}
\DeclareMathOperator*{\re}{Re}
\usepackage{caption}
\usepackage{float}
\usepackage{subfigure}
\usepackage{geometry}
\makeatletter
\@addtoreset{equation}{section}
\makeatother
\renewcommand{\theequation}{\arabic{section}.\arabic{equation}}
\begin{document}
\title{Large-time asymptotics to the focusing nonlocal modified Kortweg-de Vries equation with step-like boundary conditions}
\author{Taiyang Xu$^1$\thanks{\ Taiyang's email: tyxu19@fudan.edu.cn} \ \ \ \ Engui Fan$^1$\thanks{\ Corresponding author and email address: faneg@fudan.edu.cn } }
\footnotetext[1]{\ School of Mathematical Sciences and Key Laboratory for Nonlinear Science, Fudan University,
Shanghai 200433, P.R. China.}

\date{\today}
\baselineskip=16pt
\maketitle
\begin{abstract}
\baselineskip=16pt
We investigate the large-time asymptotics of solution for the Cauchy problem of the   nonlocal focusing
modified Kortweg-de Vries (MKdV) equation with step-like initial data, i.e.,
$u_0(x)\rightarrow 0$ as $x\rightarrow-\infty$, $u_0(x)\rightarrow A$ as $x\rightarrow+\infty$,
where $A$ is an arbitrary positive real number. We firstly develop the direct scattering theory
to establish the basic Riemann-Hilbert (RH) problem associated with step-like initial data. Thanks to the symmetries $x\rightarrow-x$, $t\rightarrow-t$
of nonlocal MKdV equation, we investigate the asymptotics for $t\rightarrow-\infty$ and $t\rightarrow+\infty$ respectively.
Our main technique is to use the steepest descent analysis to deform the original matrix-valued RH problem to
corresponded regular RH problem, which could be explicitly solved. Finally we obtain the different large-time asymptotic behaviors
of the solution of the Cauchy problem for focusing nonlocal MKdV equation in different
space-time sectors $\mathcal{R}_{I}$,  $\mathcal{R}_{II}$,  $\mathcal{R}_{III}$ and $\mathcal{R}_{IV}$ on the whole $(x,t)$-plane.
\\
\\
{\bf Keywords:} Nonlocal focusing MKdV equation, Cauchy problem with step-like initial data, Riemann-Hilbert problem, Nonlinear steepest method, Large-time asymptotics.
\\
\\
{\bf Mathematics Subject Classification:} 35Q51; 35Q15; 35C20; 37K15; 37K40.
\end{abstract}
\baselineskip=16pt

\tableofcontents

\section{Introduction and state of results}
The pioneering work for the nonlocal integrable systems was introduced by M. Ablowitz and Z. Musslimani to study the nonlinear nonlocal Schr\"odinger equation (NNLS) with PT symmetry \cite{MZPRL2013}.
As the first nonlocal integrable system, nonlocal NLS equation is a reduction of a member of the AKNS hierarchy \cite{AKNS}, namely, of the coupled Schr\"odinger equations
\begin{align*}
    iq_t+q_{xx}+2q^2r=0, \quad -ir_t+r_{xx}+2r^2q=0,
\end{align*}
corresponding to $r(x,t)=\bar{q}(-x,t)$.

Besides the nonlocal NLS equation, researchers apply the PT symmetric reduction to AKNS and the other hierarchies to derive the other types of integrable nonlocal PDEs, which include
the space-time Sine-Gordon/Sinh-Gordon equation \cite{ZFLMnonlocalSine-Gordon}, the complex/real MKdV equation \cite{MZNon2016,MZSAPM2017}, the nonlocal derivative NLS equation \cite{ZXZhou},
as well as the multidimensional nonlocal Davey-Stewartson equation \cite{nonlocalDS}.

The real nonlocal (also called reverse-space-time) MKdV (NMKdV) equation mentioned in reference \cite{MZNon2016,MZSAPM2017} is as follows:
\begin{equation}\label{FDNMKdV}
    u_{t}(x,t)+6\sigma u(x,t)u(-x,-t)u_{x}(x,t)+u_{xxx}(x,t)=0,
\end{equation}
where $\sigma=\pm 1$. For $\sigma=1$ and $\sigma=-1$, we call the corresponding equation \eqref{FDNMKdV} the focusing NMKdV equation, defocusing NMKdV equation respectively.

In the present paper, we study the following Cauchy problem for the focusing NMKdV equation (which corresponds to $\sigma=1$) with a step-like initial data
\begin{subequations}\label{Cauchy Problem}
    \begin{align}
        &u_{t}(x,t)+6u(x,t)u(-x,-t)u_{x}(x,t)+u_{xxx}(x,t)=0, \ \ x\in\mathbb{R}, \ t\in\mathbb{R},  \label{FNMKdV} \\
        &u(x,t=0)=u_{0}(x) \rightarrow\left\{
            \begin{aligned}
            &0, \ \ x\rightarrow-\infty,\\
            &A, \ \ x\rightarrow +\infty,
            \end{aligned}
            \right.,  \label{initialcon}
    \end{align}
where sufficiently fast with some real constant $A>0$.
\end{subequations}

Like the local MKdV equation, the nonlinear nonlocal MKdV equation also attracted much attention because of its
rich physical and mathematical properties. Some mathematical theories for nonlocal MKdV equation have been constructed by a lot of
investigators. For instance, some exact solutions of nonlocal MKdV equation including soliton solution, kink solution, rogue-wave and breathers,
which display some new different properties from those of local MKdV equation, are obtained through Darboux transformation \cite{ZhuZN1} as well as IST technique \cite{ZhuZN2}.
In \cite{NMKdVRHNZBCs}, Zhang and Yan use the Riemann-Hilbert approach to construct the soliton solutions under the nonzero boundary conditions with symmetric limiting functions
like $|q_{\pm}|=q_0$. In physical application, the nonlocal MKdV equation possesses the shifted parity and delayed time reveal symmetry, and thus it can be related to the Alice-Bob system \cite{SYLou}.
For the long-time asymptotic analysis of NMKdV equation \eqref{FDNMKdV}, He, Fan and Xu establish the long-time asymptotics for nonlocal defocusing MKdV equation (corresponding to $\sigma=-1$)
with decaying boundary conditions. Under the nonzero boundary conditions ($|q_{\pm}|=q_0$), Zhou and Fan use the uniformization technique to
investigate the long-time asymptotics for defocusing NMKdV equation \cite{NMKDVNBZCslongtimeI,NMKDVNBZCslongtimeII} via Dbar steepest method which was developed by
McLaughlin and his collaborators \cite{Dieng,M&M2006,M&M2008}.

Cauchy problems for nonlinear integrable systems with step-like initial data have a long history and origin from the work of Gurevich and Pitaevsky for KdV equation
\cite{GuPitae}. Since Deift-Zhou nonlinear steepest method \cite{DZAnn} became a powerful tool to investigate the long-time asymptotics of Cauchy problem for integrable PDEs,
researchers apply it to investigate a large number of integrable equations, in particular, local integrable equations with step-like initial data. The typical representative works are as follows.
Monvel, Lenells, Shepelsky and their co-authors investigate the long-time asymptotics for focusing NLS equations (corresponding to non-self-adjoint Lax operator)
\cite{MonKotSheIMRN,MonLenSheCMP,MonLenSheCMP3}. For local focusing MKdV equation, refer Minakov and his collaborators' work
\cite{GraMinakovSIAM,KotMinakovJMP,KotMinakovJMPAG1,KotMinakovJMPAG2,MinakovJPAMaTheor}. For defocusing equations (corresponding to self-adjoint Lax operator), please refer
Fromm, Lenells, Quirchmayr' work on local defocusing NLS equation \cite{LenellsDNLS}, Jenkins' work on local defocusing NLS \cite{JenkinsNon} with pure step initial functions as well as
Xu, Fan's work on local defocusing MKdV \cite{Xudmkdvsteplik}. For the step-like initial valued problem of nonlocal integrable equations, there is little related work except to some work
on nonlocal NLS equations due to Rybalko and Shepelsky \cite{RDJDE2021,RDCMP,RDJMPAG,RDStudies}. Inspired by Rybalko and Shepelsky's work, we generate the motivation to investigate the
long-time asymptotics for the Cauchy problem of the nonlocal MKdV equation \eqref{Cauchy Problem}, of which long-time asymptotics have not been studied to the best of my knowledge.

In the present work, we assume that the solution $u(x,t)$ of Cauchy problem \eqref{FNMKdV}-\eqref{initialcon} satisfies the boundary conditions (consistent with the equation) for all $t$:
\begin{subequations}\label{bdrycondition}
    \begin{align}
        &u(x,t)=A+o(1), &x\rightarrow+\infty,\\
        &u(x,t)=o(1), &x\rightarrow-\infty.
    \end{align}
\end{subequations}
In fact, we will make the sense of $o(1)$ more precise in the following contents. This choice of initial data and boundary conditions is inspired by some works for the classical (local) MKdV equation
(eg. \cite{KotMinakovJMP}),
\begin{equation}\label{localMkdV}
    u(x,t)+6u^2(x,t)u_{x}(x,t)+u_{xxx}(x,t)=0
\end{equation}
with
\begin{subequations}
    \begin{align}
        &u_0(x)=A+o(1), \quad x\rightarrow+\infty,\\
        &u_0(x)=o(1), \quad x\rightarrow-\infty.
    \end{align}
\end{subequations}
where $A$ indeed is an rarefaction wave solution of local MKdV equation \eqref{localMkdV}.

If we make the initial value $u_0(x)$ be the pure step functions (or say ``shifted step'' initial data) as follows
\begin{align}\label{purestepinitial}
    &u_{0A}(x):=\left\{
        \begin{aligned}
        &0, \ \ x<0,\\
        &A, \ \ x>0,
        \end{aligned}
            \right.
\end{align}
then the asymptotic analysis can be simplified. In the present paper, we assume that the initial function $u_0(x)$ is a compact perturbation
of the pure step initial function defined by \eqref{purestepinitial}, i.e., $u_0(x)-u_{0A}(x)=0$ for $|x|>N$ with some $N>0$, which can make our spectral functions
take properties similar to those in the case of the pure step initial functions.

\subsection{Main results}
We  divide the whole $(x,t)$-plane into different space-time cones depicted by Fig \ref{cone}, in which we obtain the different long-time asymptotic behavior of the solution
for nonlocal focusing MKdV equation \eqref{Cauchy Problem} with step-like initial data.
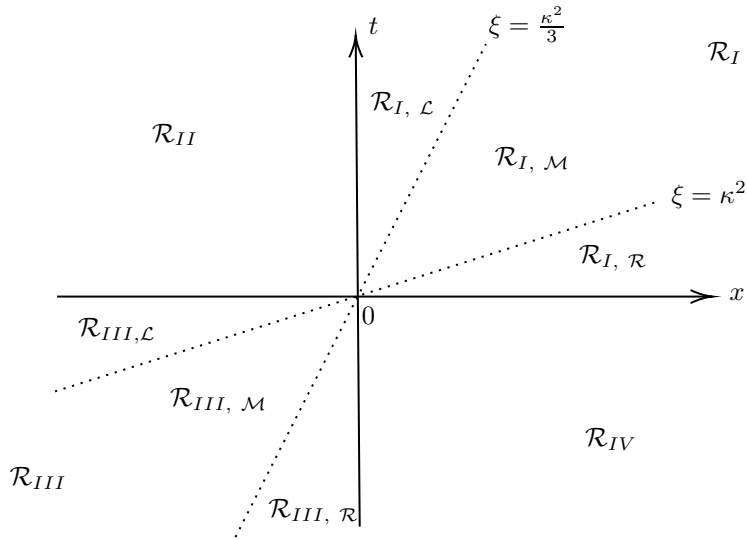
\begin{figure}[htbp]
    \begin{center}
        \tikzset{every picture/.style={line width=0.75pt}} 
        \begin{tikzpicture}[x=0.75pt,y=0.75pt,yscale=-1,xscale=1]
        \draw    (142,162) -- (467,162) ;
        \draw [shift={(469,162)}, rotate = 180] [color={rgb, 255:red, 0; green, 0; blue, 0 }  ][line width=0.75]    (10.93,-3.29) .. controls (6.95,-1.4) and (3.31,-0.3) .. (0,0) .. controls (3.31,0.3) and (6.95,1.4) .. (10.93,3.29)   ;
        \draw    (293,277.71) -- (291.02,32.71) ;
        \draw [shift={(291,30.71)}, rotate = 89.54] [color={rgb, 255:red, 0; green, 0; blue, 0 }  ][line width=0.75]    (10.93,-3.29) .. controls (6.95,-1.4) and (3.31,-0.3) .. (0,0) .. controls (3.31,0.3) and (6.95,1.4) .. (10.93,3.29)   ;
        \draw  [dash pattern={on 0.84pt off 2.51pt}]  (440,114.71) -- (141,209.71) ;
        \draw  [dash pattern={on 0.84pt off 2.51pt}]  (231,283) -- (356,35) ;
        \draw (188,73.4) node [anchor=north west][inner sep=0.75pt]    {$\mathcal{R}_{II}$};
        \draw (404,226.4) node [anchor=north west][inner sep=0.75pt]    {$\mathcal{R}_{IV}$};
        \draw (476,156.4) node [anchor=north west][inner sep=0.75pt]    {$x$};
        \draw (296,20.4) node [anchor=north west][inner sep=0.75pt]    {$t$};
        \draw (465,32.4) node [anchor=north west][inner sep=0.75pt]    {$\mathcal{R}_{I}$};
        \draw (447,102.4) node [anchor=north west][inner sep=0.75pt]    {$\xi =\kappa ^{2}$};
        \draw (356,14.4) node [anchor=north west][inner sep=0.75pt]    {$\xi =\frac{\kappa ^{2}}{3}$};
        \draw (297,57.4) node [anchor=north west][inner sep=0.75pt]    {$\mathcal{R}_{I,\ \mathcal{L}}$};
        \draw (359,85.4) node [anchor=north west][inner sep=0.75pt]    {$\mathcal{R}_{I,\ \mathcal{M}}$};
        \draw (401,134.4) node [anchor=north west][inner sep=0.75pt]    {$\mathcal{R}_{I,\ \mathcal{R}}$};
        \draw (246,262.4) node [anchor=north west][inner sep=0.75pt]    {$\mathcal{R}_{III,\ \mathcal{R}}$};
        \draw (197,206.4) node [anchor=north west][inner sep=0.75pt]    {$\mathcal{R}_{III,\ \mathcal{M}}$};
        \draw (151,171.4) node [anchor=north west][inner sep=0.75pt]    {$\mathcal{R}_{III,\mathcal{L}}$};
        \draw (117,246.4) node [anchor=north west][inner sep=0.75pt]    {$\mathcal{R}_{III}$};
        \draw (292.5,165.61) node [anchor=north west][inner sep=0.75pt]    {$0$};
        \end{tikzpicture}
    \caption{\small The space-time cones with $\xi:=\frac{x}{12t}$}\label{cone}
    \end{center}
\end{figure}

As $t\rightarrow+\infty$ and $t\rightarrow-\infty$ respectively, different asymptotics are presented in different sectors.

The $\mathcal{R}_{II}$, $\mathcal{R}_{IV}$ are typical Zakhrov-Manakov regions, and the asymptotic solution $u(x,t)$ for the Cauchy problem of \eqref{Cauchy Problem} in these regions
could be called the self-similar solution. In $\mathcal{R}_{IV}$, the leading term in the asymptotics is presented by the constant $A$ multiplied by
a slowly varying factor which tends to $1$ as $\xi \rightarrow-\infty$, in compatible with the boundary condition \eqref{bdrycondition} as $x>0$.
In $\mathcal{R}_{II}$, the leading term in the asymptotics is $0$, which is also compatible with the boundary condition \eqref{bdrycondition} as $x<0$.
And the asymptotic solution in both  $\mathcal{R}_{II}$, $\mathcal{R}_{IV}$ admit the sub-leading term as the form $\tau^{-1/2}$, which is obtained by parabolic cylinder functions.
Main difference between $\mathcal{R}_{II}$ and $\mathcal{R}_{IV}$ is that the sub-leading term of the later depends on the parameter $\im\nu(-k_0)$, However the former does not.
Refer the Theorem \ref{mainthm1} as follows.
\begin{theorem}\label{mainthm1}
    Consider the initial-valued problem \eqref{Cauchy Problem} and \eqref{bdrycondition}, where the initial data $u_0(x)$ is a compact perturbation of the pure step initial function \eqref{purestepinitial}:
    $u_0(x)-u_{0A}(x)=0$ for $|x|>N$ with some $N>0$. Assume that the spectral functions $a_j(k)$, $j=1,2$ and $b(k)$ associated to $u_0(x)$ satisfy
    \begin{itemize}
        \item[Con1.]$a_1(k)$ has a single, simple zero in $\overline{\mathbb{C}_{+}}$ located $k=i\kappa$, and $a_2(k)$ either has no zeros in $\overline{\mathbb{C}_-}$ (generic case) or has a
        single, simple zero at $k=0$ (non-generic case);
        \item[Con2.] $\im\nu(-k_0)=-\frac{1}{2\pi}d\int_{-\infty}^{-k_0}d\arg(1+r_1(s)r_2(s))\in(-\frac{1}{2},\frac{1}{2})$, where $r_1(k)=\frac{b(k)}{a_1(k)}$, $r_2(k)=\frac{b(k)}{a_2(k)}$.
    \end{itemize}
    Assuming that the solution of Cauchy problem $u(x,t)$ exists, then the long-time asymptotics of $u(x,t)$ along any line $\xi:=\frac{x}{12t}<0$, $|\xi|=O(1)$ can be described as follows:
    \begin{itemize}
    \item[I.] For $x<0$, $t>0$ (corresponding to $\mathcal{R}_{II}$), as $t\rightarrow+\infty$
    \begin{equation}\label{asyRII}
        u(x,t)=-4\eta(-\tau)^{-\frac{1}{2}-\im\nu(-k_0)}\re\left(\gamma(\xi)e^{t\varphi(\xi,0)}(-\tau)^{i\re\nu(-k_0)}\right)+R_1(\xi,-t).
    \end{equation}
    \item[II.] For $x>0$, $t<0$ (corresponding to $\mathcal{R}_{IV}$), as $t\rightarrow-\infty$, three types of asymptotic forms are possible, depending on the $\im\nu(-k_0)$, in detail,
    \begin{subequations}
    \item[(II.a)]if $\im\nu(-k_0)\in(-\frac{1}{2},-\frac{\alpha}{6}]$,
    \begin{align}\label{asyRIVa}
        u(x,t)=A\delta^2(\xi,0)-\frac{4c_0^2}{k_0^2}\eta\tau^{-\frac{1}{2}-\im\nu(-k_0)}\re\left(i\gamma(\xi)e^{t\varphi(\xi,0)}\tau^{i\re\nu(-k_0)}\right)+R_1(\xi,t),
    \end{align}
    \item[(II.b)]if $\im\nu(-k_0)\in(-\frac{\alpha}{6},\frac{\alpha}{6})$,
    \begin{align}\label{asyRIVb}
        u(x,t)&=A\delta^2(\xi,0)-\frac{4c_0^2}{k_0^2}\eta\tau^{-\frac{1}{2}-\im\nu(-k_0)}\re\left(i\gamma(\xi)e^{t\varphi(\xi,0)}\tau^{i\re\nu(-k_0)}\right)\nonumber\\
        &\hspace{4.5em}+4\eta\tau^{-\frac{1}{2}+\im\nu(-k_0)}\re\left(\beta(\xi)e^{-t\varphi(\xi,0)}\tau^{-i\re\nu(-k_0)}\right)+R_3(\xi,t),
    \end{align}
    \item[(II.c)]if $\im\nu(-k_0)\in[\frac{\alpha}{6},\frac{1}{2})$,
    \begin{align}\label{asyRIVc}
        u(x,t)&=A\delta^2(\xi,0)+4\eta\tau^{-\frac{1}{2}+\im\nu(-k_0)}\re\left(\beta(\xi)e^{-t\varphi(\xi,0)}\tau^{-i\re\nu(-k_0)}\right)+R_2(\xi,t),
    \end{align}
    \end{subequations}
    \end{itemize}
Here
\begin{align*}
    &\delta(0,\xi)=\exp\left\{\frac{1}{2\pi i}\int_{(-\infty,-k_0)\cup(k_0,+\infty)}\frac{{\log \left(1+r_1(s)r_2(s)\right)}}{s}ds\right\}, \\
    &\pm k_0=\pm\sqrt{-\frac{x}{12t}}, \quad \eta:=\frac{k_0}{2}, \quad \rho=\eta\sqrt{48k_0}, \quad \tau:=-t\rho^2=-12tk_0^3, \quad \nu:=\nu(-k_0), \\
    &\varphi(\xi;\zeta):=2i\theta\left(\xi,-k_0+\frac{\eta}{\rho}\right)=16ik_0^3-\frac{i}{2}\zeta^2+\frac{i\zeta^3}{12\rho},\\
    &\alpha\in(\lambda,1),\quad \lambda:=\max\left(1/2, \ \underset{\xi<0, |\xi|=O(1)}{\rm sup}\ 2|\im\nu(k(\xi))|\right),\\
    &\beta(\xi)=\frac{\sqrt{2\pi}e^{\frac{i\pi}{4}}e^{-\frac{\pi \nu(-k_0)}{2}}}{{q}_{1}(-k_0)\Gamma(-i\nu(-k_0))}, \quad q_1(-k_0)=e^{-2\chi(\xi,-k_0)}r_1(-k_0)e^{2i\nu(-k_0)\log 4}, \\
    &\gamma(\xi)=\frac{\sqrt{2\pi}e^{-\frac{i\pi}{4}}e^{-\frac{\pi \nu(-k_0)}{2}}}{{q}_{2}(-k_0)\Gamma(i\nu(-k_0))},\quad q_2(-k_0)=e^{2\chi(\xi,-k_0)}r_2(-k_0)e^{-2i\nu(-k_0)\log 4}.
\end{align*}
And the error estimates are as follows
\begin{align*}
    &R_1(\xi,t)=\left\{
        \begin{aligned}
        &O(\epsilon\tau^{-\frac{1+\alpha}{2}}), &\im\nu(-k_0)\geqslant 0 \\
        &O(\epsilon\tau^{-\frac{1+\alpha}{2}+2\vert\im\nu\vert}), &\im\nu(-k_0)<0
        \end{aligned}
        \right. \\
    &R_2(\xi,t)=\left\{
        \begin{aligned}
        &O(\epsilon\tau^{-\frac{1+\alpha}{2}+2\vert\im\nu\vert}), &\im\nu(-k_0)\geqslant 0 \\
        &O(\epsilon\tau^{-\frac{1+\alpha}{2}}), &\im\nu(-k_0)<0
        \end{aligned}
        \right.
\end{align*}
and
\begin{equation*}
    R_3(\xi,t):=R_1(\xi,t)+R_2(\xi,t)=\left\{
            \begin{aligned}
            &O(\epsilon\tau^{-\frac{1+\alpha}{2}}), &\im\nu(-k_0)=0,\\
            &O(\epsilon\tau^{-\frac{1+\alpha}{2}+2\vert\im\nu\vert}), &\im\nu(-k_0)\neq 0.
            \end{aligned}
            \right.
\end{equation*}
\end{theorem}

\begin{remark}\rm
    For $x>0$, $t<0$, as $x\rightarrow+\infty$, $k_0=\sqrt{-\frac{x}{12t}}\rightarrow+\infty$, $-k_0=-\sqrt{-\frac{x}{12t}}\rightarrow-\infty$, then $A\delta^2(0,\xi)\rightarrow A$, which is compatible
    to the boundary condition as $x>0$. The compatibilities to the boundary condition of the other sectors are obvious.
\end{remark}

$\mathcal{R}_{I}$ and $\mathcal{R}_{III}$ are similar. For $\mathcal{R}_{I}$, we divide it into three sectors $\mathcal{R}_{I,\mathcal{L}}$, $\mathcal{R}_{I,\mathcal{M}}$
and $\mathcal{R}_{I,\mathcal{R}}$. We name the $\mathcal{R}_{I,\mathcal{L}}$ the solitonic region because it's asymptotics take the leading term with the one-soliton form.
The leading term of $\mathcal{R}_{I,\mathcal{M}}$ is the constant $A$ and admit the same error order with the asymptotics of $\mathcal{R}_{I,\mathcal{L}}$.
As for $\mathcal{R}_{I,\mathcal{R}}$, the asymptotics admit the same leading term as asymptotics in $\mathcal{R}_{I,\mathcal{M}}$, but with the different radiation term which depends on
a small parameter $\kappa_\delta\in(0,\kappa)$.
The asymptotics of $\mathcal{R}_{III}$ are obtained by the symmetry of nonlocal focusing MKdV equation ($x\rightarrow-x$, $t\rightarrow-t$). $\mathcal{R}_{III,\mathcal{L}}$,
$\mathcal{R}_{III,\mathcal{M}}$, $\mathcal{R}_{III,\mathcal{R}}$ are corresponded to $\mathcal{R}_{I,\mathcal{R}}$, $\mathcal{R}_{I,\mathcal{M}}$
and $\mathcal{R}_{I,\mathcal{L}}$ respectively. Some similar claims for $\mathcal{R}_{I}$ are also fitted to $\mathcal{R}_{III}$. For the asymptotics of
$\mathcal{R}_{I}$ and $\mathcal{R}_{III}$, please refer the Theorem \ref{mainthm2} as follows .

\begin{theorem}\label{mainthm2}
    Under the same conditions of Theorem \ref{mainthm1}, then the long-time asymptotics of $u(x,t)$ along any line $\xi:=\frac{x}{12t}>0$, can be described as follows:
\begin{itemize}
    \item[I.] For $x>0$, $t>0$ (corresponding to $\mathcal{R}_{I}$), as $t\rightarrow+\infty$, three asymptotic forms are presented for different $\xi$ as follows
    \begin{subequations}
    \item[(I.a)]if $\xi\in(0,\frac{\kappa^2}{3})$ (corresponding to solitonic region $\mathcal{R}_{I,\mathcal{L}}$)
    \begin{equation}\label{asyRIL}
        u(x,t)=\frac{A}{1-C_1(\kappa)e^{-2\kappa x+8\kappa^3t}}+O\left(t^{-\frac{1}{2}}e^{-16t\xi^{3/2}}\right),
    \end{equation}
    \item[(I.b)]if $\xi\in(\frac{\kappa^2}{3},\kappa^2)$ (corresponding to region $\mathcal{R}_{I,\mathcal{M}}$)
    \begin{equation}\label{asyRIM}
        u(x,t)=A+O\left(t^{-\frac{1}{2}}e^{-16t\xi^{3/2}}\right),
    \end{equation}
    \item[(I.c)]if $\xi\in(\kappa^2,+\infty)$ (corresponding to region $\mathcal{R}_{I,\mathcal{L}}$)
    \begin{equation}\label{asyRIR}
        u(x,t)=A+O\left(t^{-\frac{1}{2}}e^{-8t\kappa_{\delta}(3\xi-\kappa_{\delta}^2)}\right).
    \end{equation}
    \end{subequations}
    \item[II.] For $x<0$, $t<0$ (corresponding to $\mathcal{R}_{III}$), as $t\rightarrow-\infty$, three asymptotic forms are presented for different $\xi$ as follows
    \begin{subequations}
    \item[(II.a)]if $\xi\in(0,\frac{\kappa^2}{3})$ (corresponding to solitonic region $\mathcal{R}_{III,\mathcal{R}}$)
    \begin{equation}\label{asyRIIIR}
        u(x,t)=\frac{4}{C_2(\kappa)e^{-2\kappa x+8\kappa^3t}-A\kappa^{-2}}+O\left((-t)^{-\frac{1}{2}}e^{16t\xi^{3/2}}\right),
    \end{equation}
    \item[(II.b)]if $\xi\in(\frac{\kappa^2}{3},\kappa^2)$ (corresponding to region $\mathcal{R}_{III,\mathcal{M}}$)
    \begin{equation}\label{asyRIIIM}
        u(x,t)=O\left((-t)^{-\frac{1}{2}}e^{16t\xi^{3/2}}\right),
    \end{equation}
    \item[(II.c)]if $\xi\in(\kappa^2,+\infty)$ (corresponding to region $\mathcal{R}_{III,\mathcal{L}}$)
    \begin{equation}\label{asyRIIIL}
        u(x,t)=O\left((-t)^{-\frac{1}{2}}e^{8t\kappa_{\delta}(3\xi-\kappa_{\delta}^2)}\right).
    \end{equation}
    \end{subequations}
\end{itemize}
Here
\begin{align*}
    C_1(\kappa)=\frac{A\gamma_0}{2ia_1'(i\kappa)\kappa^2}, \quad  C_2(\kappa)=\frac{2ia_1'(i\kappa)}{\gamma_0}  , \quad \kappa_{\delta}\in(0,\kappa), \quad \gamma_0^2=1.
\end{align*}
\end{theorem}

\subsection{Comparison of local/nonlocal focusing MKdV equation}
We make a short comparison of step-like Cauchy problem \eqref{Cauchy Problem} between local and nonlocal focusing MKdV equation as follows.
\begin{itemize}
    \item[(i)]Under the view of Riemann-Hilbert problem formalism, there exists a cut $(-iA,iA)$ belongs to imaginary axis for the local focusing MKdV equation, refer \cite{KotMinakovJMP}.
    For the nonlocal focusing MKdV equation, we have a singular point at $k=0$ as well as a discrete spectrum located at the pure imaginary axis instead of the cut.
    \item[(ii)]Under the view of analysis technique, we shall introduce the related $g$ function to deal with the cut for local MKdV equation with step-like initial data \eqref{Cauchy Problem}.
    However, in nonlocal case, our technique is to convert the singularity conditions of $k=0$ into residue condition so that we can construct a regular Riemann-Hilbert problem without
    residue conditions.
    \item[(iii)] Under the view of the asymptotic results, for local focusing MKdV equation, Minakov and Kotlyrov present a modulated sector between two straight line boundaries
    $x/t=Const.1$ and $x/t=Const.2$, where the leading asymptotic term is described in terms of modulated elliptic functions \cite{KotMinakovJMP}.
    However, with the Cauchy problem \eqref{Cauchy Problem}, we lack the corresponded modulated sectors and  all asymptotics of different sectors
    are described by genus-$0$ solution for nonlocal focusing MKdV equation.
\end{itemize}

\subsection{Outline of this paper}
The structure of the rest part is as follows.

In Section \ref{sectionISTandBasicRHP}, we firstly perform spectral analysis and set up the direct scattering theory for both
nonlocal focusing and defocusing MKdV equation associated to the Cauchy problem \eqref{Cauchy Problem} in subsection \ref{DirectScattering}.
Then by inverse scattering theory, we construct the basic RH problem which is suitable for the asymptotic analysis.

In Section \ref{section1-soliton}, we construct the one-soliton solution under the Assumption \ref{assforsoliton} for the nonlocal focusing MKdV equation
with Cauchy problem \eqref{Cauchy Problem}.

In Section \ref{sectionAANO1}, we mainly construct the asymptotics for $\xi<0$ as $t\rightarrow-\infty$. In Subsection \ref{subsectionFirstdeform}, we
introduce the $\delta$ function for the first RH transformation. In Subsection \ref{subsectionSeconddeform}, we execute the so-called ``opening lens" to
make the second RH deformation. In Subsection \ref{subsectionregularRHP}, we import Blaschke-Potapov factors to reduce the RH problem in \ref{subsectionSeconddeform}
to a regular RH problem without residue conditions. In Subsection \ref{subsectionlocal}, we construct the local parametrix for the regular RH problem. In Subsection \ref{subsectionAANO1erroranalysis},
we do the error analysis via Beals-Cofiman theory for regular RH problem. Finally, we review the deformations path and the symmetries of nonlocal focusing MKdV equation to form the asymptotics for
$\mathcal{R}_{II}$ and $\mathcal{R}_{IV}$ as in
the Theorem \ref{mainthm1}.

In Section \ref{sectionAANO2}, we take the similar technique mentioned in Section \ref{sectionAANO1} to construct the asymptotics for $\xi>0$ as $t\rightarrow+\infty$.
According to different asymptotic forms, we divide both  $\mathcal{R}_{I}$ and $\mathcal{R}_{III}$ into three sectors and form the corresponded asymptotics as presented
in Theorem \ref{mainthm2}.

In the last Section \ref{sectionFinalRemark}, we give a brief conclusion for the present work and provide some further discussions for the subsequent work.

\section{Inverse scattering transform and the basic RH problem}\label{sectionISTandBasicRHP}
The concrete ways to deal with RH approach to the step-like problem for classical MKdV equations substantially differ in the focusing case and the defocusing case. For the focusing case of classical MKdV equation
, the structure of the spectrum is associated to the non-self-adjoint Lax operator and a part of spectrum is outside the real axis. For the defocusing case of classical MKdV equation, the structure of the spectrum
is associated to the self-adjoint Lax operator and the whole spectrum is located on the real axis. However, we notice that the focusing case and defocusing case for the step-like Cauchy problem of NMKdV are close
to each other, in particular, there is a point singularity on the real axis. Owe to this observation, we will present some results of the direct scattering theory for the both cases (focusing and defocusing),
see the following Subsection \ref{DirectScattering}.

\subsection{Eigenfunctions and direct scattering}\label{DirectScattering}
The nonlinear nonlocal MKdV equation admits compatibility condition of the following two linear equations (Lax pairs)
\begin{subequations}\label{Laxpair}
    \begin{align}
        &\phi_x+ik\sigma_3\phi=U(x,t)\phi,\label{lpspace}\\
        &\phi_t+4ik^3\sigma_3\phi=V(x,t,k)\phi, \label{lptime}
    \end{align}
\end{subequations}
where $\phi(x,t,k)$ is a $2\times 2$ matrix-valued function, $\sigma_3=diag(1,-1)$, $k$ is a spectral parameter and
\begin{equation}
    U(x,t)=\begin{pmatrix}0 & u(x,t) \\  -\sigma u(-x,-t)& 0\end{pmatrix},
    \quad V(x,t)=\begin{pmatrix}  V_{11} & V_{12} \\ V_{21} & -V_{11}\end{pmatrix}
\end{equation}
with
\begin{align*}
    &V_{11}=2ik\sigma u(x,t)u(-x,-t)-\sigma u(-x,-t)u_{x}(x,t)-\sigma u(x,t)u_{x}(-x,-t),\\
    &V_{12}=4k^2u(x,t)+2iku_{x}(x,t)-2\sigma u^2(x,t)u(-x,-t)-u_{xx}(x,t),\\
    &V_{21}=-4k^2\sigma u(-x,-t)-2ik\sigma u_{x}(-x,-t)+2u^2(-x,-t)u(x,t)+\sigma u_{xx}(-x,-t).
\end{align*}

Consider the boundary conditions \eqref{bdrycondition}, we obtain that the matrices $U(x,t)$ and $V(x,t,k)$ converge to the following matrices
\begin{equation}
    U(x,t)\rightarrow U_{\pm}, \quad V(x,t,k)\rightarrow V_{\pm}(k), \quad x\rightarrow\pm\infty,
\end{equation}
with
\begin{subequations}
    \begin{align}
        &U_{+}=\begin{pmatrix} 0 & A \\ 0 & 0 \end{pmatrix}, \quad U_{-}=\begin{pmatrix} 0 & 0 \\ -\sigma A & 0 \end{pmatrix}, \\
        &V_{+}(k)=\begin{pmatrix} 0 & 4k^2A \\ 0 & 0 \end{pmatrix}, \quad V_{-}(k)=\begin{pmatrix} 0 & 0 \\ -4\sigma k^2A & 0 \end{pmatrix}.
    \end{align}
\end{subequations}
It's easy to notice that the systems
\begin{equation}\label{Laxpairreplace1}
    \phi_x+ik\sigma_3\phi=U_{+}\phi, \quad  \phi_t+4ik^3\sigma_3\phi=V_{+}(k)\phi,
\end{equation}
and
\begin{equation}\label{Laxpairreplace2}
    \phi_x+ik\sigma_3\phi=U_{-}\phi, \quad  \phi_t+4ik^3\sigma_3\phi=V_{-}(k)\phi,
\end{equation}
are still compatible. Then we could solve the ``background solutions" $\phi_{\pm}$ of \eqref{Laxpairreplace1} and \eqref{Laxpairreplace2} respectively
\begin{equation}
    \phi_{\pm}(x,t,k)=N_{\pm}(k)e^{-(ikx+4ik^3t)\sigma_3},
\end{equation}
where
\begin{equation}
    N_{+}(k)=\begin{pmatrix} 1 & \frac{A}{2ik} \\ 0 & 1 \end{pmatrix}, \quad N_{-}(k)=\begin{pmatrix} 1 & 0 \\ \frac{\sigma A}{2ik} & 1 \end{pmatrix}.
\end{equation}
We can see that $N_{\pm}(k)$ have singularities at $k=0$, which play a significant role in our analysis.

Define the Jost solutions as
\begin{equation}\label{moJost}
    \psi_{j}(x,t,k):=\phi_{j}(x,t,k)e^{(ikx+4ik^3t)\sigma_3}, \quad j=1,2,
\end{equation}
where
\begin{equation}
\phi_1\rightarrow\phi_{-}, \ x\rightarrow-\infty \quad \textnormal{and} \quad  \phi_{2}\rightarrow\phi_{+}, \ x\rightarrow+\infty.
\end{equation}

And $\psi_{j}$, $j=1,2$ admit the following lax pairs respectively:
\begin{subequations}\label{Laxpairreplace3}
    \begin{align}
        &(N_{-}^{-1}\psi_1)_{x}-ik[N_{-}^{-1}\psi_1, \sigma_3]=N_{-}^{-1}(U-U_{-})\psi_1,\\
        &(N_{-}^{-1}\psi_1)_{x}-4ik^3[N_{-}^{-1}\psi_1, \sigma_3]=N_{-}^{-1}(V-V_{-})\psi_1,
    \end{align}
\end{subequations}
as well as
\begin{subequations}\label{Laxpairreplace4}
    \begin{align}
        &(N_{+}^{-1}\psi_2)_{x}-ik[N_{+}^{-1}\psi_2, \sigma_3]=N_{+}^{-1}(U-U_{+})\psi_2,\\
        &(N_{+}^{-1}\psi_2)_{x}-4ik^3[N_{+}^{-1}\psi_2, \sigma_3]=N_{+}^{-1}(V-V_{+})\psi_2.
    \end{align}
\end{subequations}
Then we can solve that the Jost solutions $\psi_j$, $j=1,2$ which admit the following Volterra integral equations
\begin{subequations}\label{Volterrapsi}
    \begin{align}
        &\psi_1(x,t,k)=N_{-}(k)+\int_{-\infty}^{x}G_{-}(x,y,t,k)\left(U(y,t)-U_{-}\right)\psi_{1}(y,t,k)e^{ik(x-y)\sigma_3}dy, \label{Volterrapsi1}\\
        &\psi_2(x,t,k)=N_{+}(k)-\int_{x}^{+\infty}G_{+}(x,y,t,k)\left(U(y,t)-U_{+}\right)\psi_{2}(y,t,k)e^{ik(x-y)\sigma_3}dy  \label{Volterrapsi2}
    \end{align}
\end{subequations}
where the $G_{\pm}(x,y,t,k)$ takes the form
\begin{equation}
    G_{\pm}(x,y,t,k)=\phi_{\pm}(x,t,k)\phi^{-1}_{\pm}(y,t,k).
\end{equation}

We summarize some basic properties of $\psi_{j}$, $j=1,2$ as the following proposition
\begin{proposition}\label{propofpsi}
The matrix-valued functions $\psi_1(x,t,k)$ and $\psi_2(x,t,k)$ have the following properties:
\begin{itemize}
    \item [(i)] The columns $\psi_{1}^{(1)}(x,t,k)$ and  $\psi_{2}^{(2)}(x,t,k)$ are analytic for $k\in\mathbb{C}_+$ and continuous in $\overline{\mathbb{C}_+}\backslash\{0\}$.
    The columns $\psi_{1}^{(2)}(x,t,k)$ and  $\psi_{2}^{(1)}(x,t,k)$ are analytic for $k\in\mathbb{C}_-$ and continuous in $\overline{\mathbb{C}_-}$.
    \item [(ii)] As $k\rightarrow\infty$, $\psi_{j}(x,t,k)=I+O(k^{-1})$. More over
    \begin{equation} \label{uintermsofpsi}
    u(x,t)=2i\lim_{k\rightarrow\infty}{k\psi_j(x,t,k)}_{12}, \quad -\sigma u(-x,-t)=2i\lim_{k\rightarrow\infty}{k\psi_j(x,t,k)}_{21}.
    \end{equation}
    \item [(iii)] \textnormal{det}$\psi_j(x,t,k)=1$ for $x,k\in\mathbb{R}$.
    \item [(iv)] The following two symmetric relations hold\\
    \underline{Symmetry Reduction I}:
    \begin{subequations}\label{symall}
    \begin{align}\label{syma}
        \Lambda\overline{\psi_{1}(-x,-t,-\bar{k})}\Lambda^{-1}=\psi_2(x,t,k), \quad k\in\mathbb{C}\backslash\{0\}
    \end{align}
    where $\Lambda=\begin{pmatrix} 0 & \sigma \\ 1 & 0 \end{pmatrix}$. In particular,
    \begin{align}\label{symb}
        \Lambda\overline{\psi_{1}(-x,-t,-k)}\Lambda^{-1}=\psi_2(x,t,k), \quad k\in\mathbb{R}\backslash\{0\}
    \end{align}
    \end{subequations}
    as well as \\
    \underline{Symmetry Reduction II:}
    \begin{align}\label{Supplementary symmetry}
        \Lambda\psi_{1}(-x,-t,k)\Lambda^{-1}=\psi_2(x,t,k), \quad k\in\mathbb{C}\backslash\{0\}.
    \end{align}
    \item[(v)] As $k\rightarrow 0$, $\psi_j(x,t,k)$, $j=1,2$ admit the following singularities
    \begin{subequations}\label{psi1psi2atk=0}
        \begin{align}
        & \psi_1(x,t,k)=\left(
            \begin{array}{cc}
            \frac{1}{k}v_1(x,t)+O(1) & \frac{2i\sigma}{A}v_{1}(x,t)+O(k) \\
            \frac{1}{k}v_2(x,t)+O(1) & \frac{2i\sigma}{A}v_{2}(x,t)+O(k)
            \end{array}
        \right),\\
        & \psi_2(x,t,k)=\left(
            \begin{array}{cc}
            -\frac{2i}{A}\sigma\overline{v_2(-x,-t)}+O(k) & -\frac{1}{k}\sigma\overline{v_{2}(-x,-t)}+O(1) \\
            -\frac{2i}{A}\overline{v_1(-x,-t)}+O(k)  & -\frac{1}{k}\overline{v_{1}(-x,-t)}+O(1)
            \end{array}
        \right),
        \end{align}
    \end{subequations}
    where $v_j(x,t)$, $j=1,2$ are some functions and  satisfy the system of integral equations
    \begin{align}\label{v1v2int}
        \left\{
        \begin{aligned}
         &v_1(x,t)=\int_{-\infty}^{x}u(y,t)v_2(y,t)dy, \\
         &v_2(x,t)=-i\sigma A-\sigma \int_{-\infty}^{x}\left(u(-y,-t)-A\right)v_1\left(y,t\right)dy, \\
        \end{aligned}
    \right.
    \end{align}
\end{itemize}
\end{proposition}
\begin{proof}
The item (i) directly follows from the representation of $\psi_j$, $j=1,2$ in terms of the Neumann series associated with equations \eqref{Volterrapsi}.

For the item (ii), we expand $\psi_{j}$ as
\begin{equation}\label{expandpsikatinfty}
    \psi_j=\psi_{j,E_0}+\frac{\psi_{j,E_1}}{k}+O(k^{-2}), \quad k\rightarrow\infty.
\end{equation}
Notice the $\psi_j$ satisfies the linear equations
\begin{subequations}\label{Laxpairpsi}
\begin{align}
    & \psi_{j,x}+ik[\sigma_3,\psi_j]=U\psi_{j} \\
    & \psi_{j,t}+4ik^3[\sigma_3,\psi_j]=V\psi_{j},
\end{align}
\end{subequations}
we substitute \eqref{expandpsikatinfty} into \eqref{Laxpairpsi} and compare the order to obtain that\\
\begin{subequations}
    $x$-part:
    \begin{align}
        &O(1): \psi_{j,E_0,x}+i[\sigma_3,\psi_{j,E_1}]=U\psi_{j,E_0},\label{xpart1}\\
        &O(k): i[\sigma_3,\psi_{j,E_0}]=0, \label{xpart2}
    \end{align}
    $t$-part:
    \begin{align}
        &O(k^2): 4i[\sigma_3,\psi_{j,E_1}]=4U\psi_{j,E_0},\label{tpart1}\\
        &O(k^3): 4i[\sigma_3,\psi_{j,E_0}]=0. \label{tpart2}
    \end{align}
\end{subequations}
From \eqref{xpart2} and \eqref{tpart2}, we know $\psi_{j,E_0}$ is a diagonal matrix. From \eqref{xpart1} and \eqref{tpart1}, we can get $\psi_{j,E_0}$ is independent of parameter $x$.
Then
\begin{equation}
    \psi_{j,E_0}=\lim_{x\rightarrow\infty}\lim_{k\rightarrow\infty}\psi_j=\lim_{k\rightarrow\infty}\lim_{x\rightarrow\infty}\psi_j=\lim_{k\rightarrow\infty}N_{\pm}(k)=I,
\end{equation}
which implies $\psi_{j}(x,t,k)=I+O(k^{-1})$ as $k\rightarrow\infty$. Furthermore, Combine \eqref{xpart1} and \eqref{tpart1}, we have
\begin{equation}
    u(x,t)=2i\psi_{j,E_1,12}, \ {\rm and} \ -\sigma u(-x,-t)=-2i\psi_{j,E_1,21},
\end{equation}
which imply \eqref{uintermsofpsi}.

Item (iii) follows from the facts that $U(x,t)$ and $V(x,t)$ are traceless matrices and det$\psi_j=$det$\phi_j$ for $x,k\in\mathbb{R}$.

Item (iv) follows from the symmetries
\begin{equation*}
    \Lambda U(-x,-t)\Lambda^{-1}=-U(x,t), \quad \Lambda\overline{N_{-}(-\bar{k})}\Lambda^{-1}=N_{+}(k), \quad \Lambda N_{-}(k)\Lambda^{-1}=N_{+}(k).
\end{equation*}

Let us consider item (v). From \eqref{Volterrapsi}. Observing the accurate structure of the singularities of $N_{\pm}(k)$ as $k\rightarrow 0$ and the integral expressions of $\psi_j(x,t,k)$, $j=1,2$ in
the \eqref{Volterrapsi}, we suppose that, as $k\rightarrow 0$,
\begin{subequations}\label{(v)con1}
    \begin{align}
        \psi_{1}^{(1)}(x,t,k)=\frac{1}{k}\left(
            \begin{array}{cc}
            v_{1}(x,t)\\
            v_{2}(x,t)
            \end{array}
        \right)+O(1), \quad
        \psi_{1}^{(2)}(x,t,k)=\left(
            \begin{array}{cc}
            \tilde{v}_{1}(x,t)\\
            \tilde{v}_{2}(x,t)
            \end{array}
        \right)+O(k), \label{psi1k=0}\\
        \psi_{2}^{(1)}(x,t,k)=\left(
            \begin{array}{cc}
            \tilde{w}_{1}(x,t)\\
            \tilde{w}_{2}(x,t)
            \end{array}
        \right)+O(k), \quad
        \psi_{2}^{(2)}(x,t,k)=\frac{1}{k}\left(
            \begin{array}{cc}
            w_{1}(x,t)\\
            w_{2}(x,t)
            \end{array}
        \right)+O(1) \label{psi2k=0}
    \end{align}
\end{subequations}
with some functions $v_j$, $\tilde{v}_j$, $w_j$ and $\tilde{w}_j$, $j=1,2$. Then the symmetry relation \eqref{symb} implies that
\begin{equation}\label{(v)con2}
        \left(\begin{array}{cc}
            w_{1}(x,t)\\
            w_{2}(x,t)
            \end{array}
        \right)=
        \left(\begin{array}{cc}
            -\sigma\overline{v_2(-x,-t)}\\
            -\overline{v_1(-x,-t)}
            \end{array}
        \right), \quad
        \left(\begin{array}{cc}
            \tilde{w}_{1}(x,t)\\
            \tilde{w}_{2}(x,t)
            \end{array}
        \right)=
        \left(\begin{array}{cc}
            \overline{\tilde{v}_2(-x,-t)}\\
            \sigma\overline{\tilde{v}_1(-x,-t)}
            \end{array}
        \right).
\end{equation}
Furthermore, substituting \eqref{psi1k=0} into \eqref{Volterrapsi1}, we obtain that \eqref{v1v2int}. Similarly, we substitute \eqref{psi2k=0} into \eqref{Volterrapsi2}, we have
\begin{align}\label{tildev1v2int}
    \left\{
    \begin{aligned}
     &\tilde{v}_1(x,t)=\int_{-\infty}^{x}u(y,t)\tilde{v}_2(y,t)dy, \\
     &\tilde{v}_2(x,t)=1-\sigma \int_{-\infty}^{x}\left(u(-y,-t)-A\right)\tilde{v}_1\left(y,t\right)dy, \\
    \end{aligned}
\right.
\end{align}
Comparing \eqref{v1v2int} with \eqref{tildev1v2int}, we easily get
\begin{equation}\label{(v)con3}
    \left(\begin{array}{cc}
        \tilde{v}_{1}(x,t)\\
        \tilde{v}_{2}(x,t)
        \end{array}
    \right)=\frac{2i\sigma}{A}
    \left(\begin{array}{cc}
        v_1(x,t)\\
        v_2(x,t)
        \end{array}
    \right).
\end{equation}
Summarizing \eqref{(v)con1}, \eqref{(v)con2} and \eqref{(v)con3}, we completely finish the proof of item (v).
\end{proof}
\begin{remark} \rm
    Use the second symmetry reduction of item (iv) in Proposition \ref{propofpsi}, i.e., \eqref{Supplementary symmetry}, we know that the unprescribed functions $v_{j}(x,t)$, $j=1,2$ defined
    by \eqref{psi1psi2atk=0} admit that $v_{j}(x,t)+\overline{v_j(x,t)}=0$, which is equivalent to $v_{j}(x,t)$ is pure imaginary.
\end{remark}

Since $\phi_1(x,t,k)$ and $\phi_2(x,t,k)$ are two fundamental matrix solutions of Lax integrable system \eqref{Laxpair}, thus there exists a so-called scattering matrix $S(k)$ such that
\begin{equation}
    \phi_1(x,t,k)=\phi_2(x,t,k)S(k), \quad k\in\mathbb{R}\backslash\{0\}
\end{equation}
equivalently, in terms of $\psi_{j}(x,t,k)$
\begin{equation}\label{linearrelation}
    \psi_1(x,t,k)=\psi_2(x,t,k)e^{-(ikx+4ik^3t)\sigma_3}S(k)e^{(ikx+4ik^3t)\sigma_3}, \quad k\in\mathbb{R}\backslash\{0\}.
\end{equation}
The following proposition presents the form of scattering matrix $S(k)$
\begin{proposition}\label{propspecfuncs}
    The scattering matrix $S(k)$ defined by \eqref{linearrelation} admits the form
    \begin{subequations}\label{Sform}
    \begin{align}
        &S(k)=\begin{pmatrix}\label{Sform1}
        a_1(k) & -\sigma\overline{b(-k)}\\
        b(k)  &  a_{2}(k)
        \end{pmatrix}, \quad  k\in\mathbb{R}\backslash\{0\},
    \end{align}
    which is equivalent to
    \begin{align}\label{Sform2}
        &S(k)=\begin{pmatrix}
            a_1(k) & -\sigma b(k)\\
            b(k)  &  a_{2}(k)
            \end{pmatrix}, \quad  k\in\mathbb{R}\backslash\{0\}.
    \end{align}
    \end{subequations}
    Moreover
    \begin{itemize}
        \item[(i)] $a_{1}(k)$ is analytic for $k\in\mathbb{C}_{+}$ and continuous in $\overline{\mathbb{C}_+}\backslash\{0\}$, $a_{2}(k)$ is analytic for $k\in\mathbb{C}_{-}$
        and continuous in $\overline{\mathbb{C}_{-}}$.
        \item[(ii)] $a_1(k)=\overline{a_1(-\bar{k})}$ for $k\in\overline{\mathbb{C}_+}\backslash\{0\}$, $a_2(k)=\overline{a_2(-\bar{k})}$ for $k\in\overline{\mathbb{C}_{-}}$, $b(k)=\overline{b(-k)}$ for $k\in\mathbb{R}$.
        In particular $a_1(k)=\overline{a_1(-k)}$ for $k\in\mathbb{R}\backslash\{0\}$, $a_2(k)=\overline{a_2(-k)}$ for $k\in\mathbb{R}$.
        \item[(iii)] $a_{j}(k)=1+O(k^{-1})$ as $k\rightarrow\infty$, $b(k)=O(k^{-1})$ as $k\rightarrow\infty$ for $k\in\mathbb{R}$.
        \item[(iv)] $a_{1}a_{2}(k)+\sigma b(k)\overline{b(-k)}=a_{1}a_{2}(k)+\sigma b^2(k)=1$, for $k\in\mathbb{R}\backslash\{0\}$.
        \item[(v)]As $k\rightarrow 0$, $a_1(k)=\sigma\frac{A^2a_{2}(0)}{4k^2}+O(k^{-1})$ for $k\in\overline{\mathbb{C}_{+}}$, $b(k)=\sigma\frac{Aa_{2}(0)}{2ik}+O(1)$ for $k\in\mathbb{R}$.
    \end{itemize}
\end{proposition}
\begin{proof}
Replace $x$, $t$, $k$ by $-x$, $-t$, $-\bar{k}$ in \eqref{linearrelation}, then take conjugation and use the symmetry relation \eqref{symb}, we obtain that
\begin{equation}\label{linearrelation2}
    \psi_1(x,t,k)=\psi_2(x,t,k)e^{-(ikx+4ik^3t)\sigma_3}\sigma\Lambda \overline{S^{-1}(-k)}\Lambda e^{(ikx+4ik^3t)\sigma_3}, \quad k\in\mathbb{R}\backslash\{0\}.
\end{equation}
Comparing \eqref{linearrelation2} with \eqref{linearrelation}, we derive the form \eqref{Sform1}.

Similarly, replace $x$, $t$ by $-x$, $-t$ in \eqref{linearrelation} and use the symmetric relation \eqref{Supplementary symmetry}, we obtain that
\begin{equation}\label{linearrelation3}
    \psi_1(x,t,k)=\psi_2(x,t,k)e^{-(ikx+4ik^3t)\sigma_3}\sigma\Lambda S^{-1}(k)\Lambda e^{(ikx+4ik^3t)\sigma_3}, \quad k\in\mathbb{R}\backslash\{0\}.
\end{equation}
Comparing \eqref{linearrelation3} with \eqref{linearrelation}, we obtain the form \eqref{Sform2}.

Item (ii) directly follows from the detailed derivation for $S(k)$.

The spectral functions can be obtained in terms of initial data only
\begin{equation}
    S(k)=\psi_2^{-1}(0,0,k)\psi_1(0,0,k),
\end{equation}
alternatively could be written in terms of determinant representations
\begin{subequations}\label{detrepforspecfunc}
    \begin{align}
        &a_1(k)=\textnormal{Wr}\left(\psi_1^{(1)}(0,0,k),\psi_2^{(2)}(0,0,k)\right), &k\in\overline{\mathbb{C}_+}\backslash\{0\}, \\
        &a_2(k)=\textnormal{Wr}\left(\psi_2^{(1)}(0,0,k),\psi_1^{(2)}(0,0,k)\right), &k\in\overline{\mathbb{C}_{-}}, \\
        &b(k)=\textnormal{Wr}\left(\psi_2^{(1)}(0,0,k),\psi_1^{(1)}(0,0,k)\right), &k\in\mathbb{R}.
    \end{align}
\end{subequations}
Use the properties of $\psi_j$ (Proposition \ref{propofpsi}), we easily get item (i) and item (iii). Item (iv) follows from det$S(k)=1$ for $k\in\mathbb{R}$.

substitute \eqref{psi1psi2atk=0} into \eqref{detrepforspecfunc}, we derive that, as $k\rightarrow 0$
\begin{subequations}\label{repforspecfuncatk=0}
    \begin{align}
        &a_1(k)=\frac{1}{k^2}\left(\sigma\vert v_2(0,0) \vert^2-\vert v_1(0,0) \vert^2\right)+O\left(\frac{1}{k}\right), \\
        &a_2(k)=\frac{4\sigma}{A^2}\left(\sigma\vert v_2(0,0) \vert^2-\vert v_1(0,0) \vert^2\right)+O(k), \\
        &b(k)=-\frac{2i}{Ak}\left(\sigma\vert v_2(0,0) \vert^2-\vert v_1(0,0) \vert^2\right)+O(1),
    \end{align}
\end{subequations}
from which item (v) follows.
\end{proof}

\begin{remark}
    \rm{In \eqref{detrepforspecfunc}, one can replace $(0,0)$ by $(x,t)$, which implies that the quantity $\sigma\vert v_2(0,0) \vert^2-\vert v_1(0,0) \vert^2$ in the r.h.s of
    \eqref{repforspecfuncatk=0} should be replaced by $\sigma v_2(x,t)\overline{v_2(-x,-t)} -v_1(x,t)\overline{v_1(-x,-t)}$, the latter being a conserved quantity (independent of
    variables $x$ and $t$).}
\end{remark}

\begin{remark}\label{pure step-like}
    \rm In the case of pure-step like initial data \eqref{purestepinitial}, the scattering matrix $S(k)$ admits the following specific form
\begin{equation*}
    S(k)=\phi_2^{-1}(0,0,k)\phi_1(0,0,k)=N_{+}^{-1}(k)N_{-}(k)=\left(
        \begin{array}{cc}
        1+\frac{\sigma A^2}{4k^2} & -\frac{A}{2ik} \\
        \frac{\sigma A}{2ik} & 1
        \end{array}
    \right).
\end{equation*}
For the focusing case ($\sigma=1$), $a_1(k)$ has a single, simple zero at $k=i\frac{A}{2}$ in the upper half-plane whereas $a_2(k)$ has no zeros in the lower half-plane. The discussion for the zeros of spectral
functions under the pure-step case could guide us make proper assumptions under the non-pure-step case.
\end{remark}

We define the following reflection coefficients as
\begin{equation}\label{defofreflectioncoe}
    r_{1}(k):=\frac{b(k)}{a_1(k)}, \quad r_{2}(k):=\frac{\overline{b(-k)}}{a_2(k)}=\frac{b(k)}{a_2(k)} \ \left(\textnormal{follows \ from }\overline{b(-k)}=b(k)\right).
\end{equation}
We notice that the item (ii) of Proposition \ref{propspecfuncs} implies that
\begin{equation}\label{reflectioncoeCondition1}
    \overline{r_1(-k)}=r_1(k), \quad  \overline{r_2(-k)}=r_2(k), \quad k\in\mathbb{R}\backslash\{0\}.
\end{equation}
For $\sigma=1$, via the item (iv) of Proposition \ref{propspecfuncs}, we obtain
\begin{equation}\label{reflectioncoeCondition2}
1+r_1(k)r_2(k)=\frac{1}{a_1(k)a_2(k)}, \quad k\in\mathbb{R}\backslash\{0\}.
\end{equation}

\subsection{The basic Riemann-Hilbert problem}\label{subsectionbasiRHP}
The aim of the present work is to present the asymptotic analysis of the focusing ($\sigma=1$) nonlocal MKdV equation with step-like initial data, so we fix $\sigma=1$ in the above results as well as the
subsequent contents. The Riemann-Hilbert problem formalism via the IST technique is based on setting up the following $2\times 2$ matrix-valued piecewise meromorphic functions $M(x,t,k)$ with jump condition
on the real line:
\begin{equation}\label{piecewisematrixM}
    M(x,t,k):=\left\{
        \begin{aligned}
        &\left(\frac{\psi_1^{(1)}(x,t,k)}{a_1(k)}, \psi_2^{(2)}(x,t,k)\right), \quad k\in \mathbb{C}_{+} \\
        &\left(\psi_2^{(1)}(x,t,k), \frac{\psi_1^{(2)}(x,t,k)}{a_2(k)}\right), \quad k\in \mathbb{C}_{-}.
        \end{aligned}
        \right.
\end{equation}
Then the boundary values $M_{\pm}(x,t,k):=\lim_{k'\rightarrow k,k'\in\mathbb{C}_{\pm}}M(x,t,k)$, $k\in\mathbb{R}$ satisfy the jump condition
\begin{equation*}
    M_{+}(x,t,k)=M_{-}(x,t,k)J(x,t,k)
\end{equation*}
where
\begin{equation*}
    J(x,t,k)=\left(
        \begin{array}{cc}
        1+r_1(k)r_2(k) &  r_2(k)e^{-2ikx-8ik^3t}\\
        r_1(k)e^{2ikx+8ik^3t} &  1
        \end{array}
        \right), \quad k\in\mathbb{R}\backslash\{0\}.
\end{equation*}

Under the view of \eqref{repforspecfuncatk=0}, the behavior of $M$ near $k=0$ is different in the case $a_2(0)=0$ as well as $a_2(0)\neq 0$. The former case contains the case of pure step-like initial-valued
problem, refer Remark \ref{pure step-like}, where $a_1(k)$ has a single simple zero located on the $i\mathbb{R}_{+}$, whereas $a_2(k)$ has not any zeros in the closure of lower half-plane. Since small perturbation
of the pure step initial data preserves these properties, we shall discuss the following two cases in the present work:

{\bf Case I, Generic Case}: The spectral function $a_1(k)$ has a simple pure imaginary zero at $k=i\kappa$, $\kappa>0$ and $a_{2}(k)$ has no zeros in $\overline{\mathbb{C}_-}$.

{\bf Case II, Non-Generic Case}: The spectral function $a_1(k)$ has a simple pure imaginary zero at $k=i\kappa$, $\kappa>0$, and $a_{2}(k)$ has one simple zero in $\overline{\mathbb{C}_-}$ located on $k=0$.
Thus we assume that $a_2'(0)\neq 0$, where $(')=\frac{d}{dk}$. Additionally, we suppose that $a_1(k)=\frac{a_{11}}{k}+O(1)$, $a_{11}\neq 0$ as $k\rightarrow 0$.

\begin{remark}\label{remarkv200-v100}
    \rm Under the view of \eqref{repforspecfuncatk=0}, Case I corresponds to the $\vert v_2(0,0) \vert^2-\vert v_1(0,0) \vert^2\neq 0$ whereas Case II corresponds to the equality
    $\vert v_2(0,0) \vert^2-\vert v_1(0,0) \vert^2=0$.
\end{remark}

Now we introduce the following proposition, which shows that the value $\kappa$ could be described in terms of some principle valued integrals with respect to the spectral functions $b(k)$.
\begin{proposition}\label{kappaexpression}
    the simple zero $k=i\kappa$ of $a_1(k)$, $\kappa>0$ is determined as the following equalities
    \begin{itemize}
    \item[(i)] In Case I (generic case),
        \begin{equation}\label{kappageneric}
            \kappa=\frac{A}{2}{\rm exp}\left\{-\frac{1}{2\pi i}{\rm p.v.}\int_{-\infty}^{\infty}\frac{{\rm log}\frac{s^2}{1+s^2}\left(1-b^2(s)\right)}{s}ds\right\}.
        \end{equation}
    \item[(ii)] In case II (non-generic case),
    \begin{equation}\label{kappanongeneric}
        \kappa=A\frac{\sqrt{b^2(0)+I_2^2}-b(0)}{2I_1I_2}.
    \end{equation}
    where
    \begin{equation}\label{I1I2}
        I_1={\rm exp}\left\{\frac{1}{2\pi i}{\rm p.v.}\int_{-\infty}^{\infty}\frac{{\rm log}\left(1-b^2(s)\right)}{s}ds\right\},
        \quad I_2=\exp\left\{\frac{1}{2}{\rm log}\left(1-b^2\left(0\right)\right)\right\}.
    \end{equation}
    Additionally, $a_{11}a_2'(0)=1-b^2(0)\neq 0$.
    \end{itemize}
\end{proposition}
\begin{proof}
Follow the homologous procedure which be used in the reference \cite[Proposition 3]{RDJDE2021}, we can obtain this proposition. More details for the proof are set in the Appendix \ref{kappaCalculation}.
\end{proof}

Notice the singularities of eigenfunctions $\psi_j(x,t,k)$, $j=1,2$ and spectral function $a_1(k)$ at $k=0$ (see the item (v) in Proposition \ref{propofpsi} and item (v) in Proposition \ref{propspecfuncs}),
we exhibit the behavior of piecewise matrix-valued functions $M(x,t,k)$ defined by \eqref{piecewisematrixM} at $k=0$ in the following proposition.
\begin{proposition}\label{propMbehaviorat0}
    The behavior of $M(x,t,k)$ defined by \eqref{piecewisematrixM} at $k=0$ can be presented as follows:

    {\bf Case I}
    \begin{subequations}
        \begin{align*}
            M_{+}(x,t,k)&=\left(
                    \begin{array}{cc}
                    \frac{4v_1(x,t)}{A^2a_2(0)}k+O(k^2) & -\frac{1}{k}\overline{v_2(-x,-t)}+O(1) \\
                    \frac{4v_2(x,t)}{A^2a_2(0)}k+O(k^2) & -\frac{1}{k}\overline{v_1(-x,-t)}+O(1)
                    \end{array}
                \right)\nonumber\\
                &=\left(
                    \begin{array}{cc}
                    \frac{4v_1(x,t)}{A^2a_2(0)} & -\overline{v_2(-x,-t)} \\
                    \frac{4v_2(x,t)}{A^2a_2(0)} & -\overline{v_1(-x,-t)}
                    \end{array}
                \right)\left(I+O(k)\right)\left(
                    \begin{array}{cc}
                    k & 0 \\
                    0 & \frac{1}{k}
                    \end{array}
                \right), &\mathbb{C}_{+}\ni k\rightarrow 0,\\
            M_{-}(x,t,k)&=\frac{2i}{A}\left(
                \begin{array}{cc}
                -\overline{v_2(-x,-t)} &  \frac{v_1(x,t)}{a_2(0)}\\
                -\overline{v_1(-x,-t)} &  \frac{v_2(x,t)}{a_2(0)}
                \end{array}
                \right)+O(k), &\mathbb{C}_{-}\ni k\rightarrow 0.
        \end{align*}
    \end{subequations}

    {\bf Case II}
    \begin{subequations}
        \begin{align*}
            M_{+}(x,t,k)&=\left(
                    \begin{array}{cc}
                    \frac{v_1(x,t)}{a_{11}}+O(k) & -\frac{1}{k}\overline{v_2(-x,-t)}+O(1) \\
                    \frac{v_2(x,t)}{a_{11}}+O(k) & -\frac{1}{k}\overline{v_1(-x,-t)}+O(1)
                    \end{array}
                \right)\nonumber\\
                &=\left(
                    \begin{array}{cc}
                    \frac{v_1(x,t)}{a_{11}} & -\overline{v_2(-x,-t)} \\
                    \frac{v_2(x,t)}{a_{11}} & -\overline{v_1(-x,-t)}
                    \end{array}
                \right)\left(I+O(k)\right)\left(
                    \begin{array}{cc}
                    1 & 0 \\
                    0 & \frac{1}{k}
                    \end{array}
                \right), &\mathbb{C}_{+}\ni k\rightarrow 0,\\
            M_{-}(x,t,k)&=\frac{2i}{A}\left(
                \begin{array}{cc}
                -\overline{v_2(-x,-t)}+O(k) & \frac{v_1(x,t)}{ka_{2}'(0)}+O(1) \\
                -\overline{v_1(-x,-t)}+O(k) & \frac{v_2(x,t)}{ka_{2}'(0)}+O(1)
                \end{array}
            \right)\nonumber\\
            &=\frac{2i}{A}\left(
                \begin{array}{cc}
                    -\overline{v_2(-x,-t)} & \frac{v_1(x,t)}{a_{2}'(0)} \\
                    -\overline{v_1(-x,-t)} & \frac{v_2(x,t)}{a_{2}'(0)}
                \end{array}
            \right)\left(I+O(k)\right)\left(
                \begin{array}{cc}
                1 & 0 \\
                0 & \frac{1}{k}
                \end{array}
            \right), &\mathbb{C}_{-}\ni k\rightarrow 0.
        \end{align*}
    \end{subequations}
\end{proposition}
\begin{proof}
    Use the item (v) in Proposition \ref{propofpsi}, the item (v) in Proposition \ref{propspecfuncs} as well as $a_{11}$ is determined by $a_1(k)=a_{11}k^{-1}+O(1)$ as $k\rightarrow 0$.
\end{proof}

Next, we consider the residue conditions for $M(x,t,k)$ defined by \eqref{piecewisematrixM}. The results is exhibited as follows
\begin{proposition}\label{propMrescondition}
    Assume that $i\kappa$, $\kappa>0$ is a simple zero of $a_{1}(k)$, then $M(x,t,k)$ satisfies the following residue conditions
    \begin{equation*}
        \underset{k=i\kappa}{\rm Res}M^{(1)}(x,t,k)=\frac{\gamma_0}{a_{1}'(i\kappa)}e^{-2\kappa x+8\kappa^3t}M^{(2)}(x,t,i\kappa), \quad \gamma_0^2=1,
    \end{equation*}
    where $\psi_1^{(1)}(0,0,i\kappa)=\gamma_0\psi_2^{(2)}(0,0,i\kappa)$.
\end{proposition}
\begin{proof}
    The constructing of residue conditions is standard. Now we give an explanation for $\gamma_0^2=1$. Take advantage of the symmetry \eqref{syma} and \eqref{Supplementary symmetry}, we can obtain that
    \begin{subequations}
    \begin{align}
    \left(\begin{array}{cc}\label{tofindeta}
        \overline{\psi_{1,21}(0,0,i\kappa)}\\
        \overline{\psi_{1,11}(0,0,i\kappa)}
        \end{array}
    \right)=
    \left(\begin{array}{cc}
        \psi_{2,12}(0,0,i\kappa)\\
        \psi_{2,22}(0,0,i\kappa)
        \end{array}
    \right)\\
    \left(\begin{array}{cc}\label{tofindeta2}
        \psi_{1,21}(0,0,i\kappa)\\
        \psi_{1,11}(0,0,i\kappa)
        \end{array}
    \right)=
    \left(\begin{array}{cc}
        \psi_{2,12}(0,0,i\kappa)\\
        \psi_{2,22}(0,0,i\kappa)
        \end{array}
    \right)
    \end{align}
\end{subequations}
where $\psi_{j,lm}$ represents the $(l,m)$-entry of $\psi_j$. Combine \eqref{tofindeta} and \eqref{tofindeta2}, we can see that $\psi_{1,11}(0,0,i\kappa)$ and $\psi_{1,21}(0,0,i\kappa)$ are real-valued.
Then take into account $\psi_1^{(1)}(0,0,i\kappa)=\gamma_0\psi_2^{(2)}(0,0,i\kappa)$
\begin{equation}
    \left\{
        \begin{aligned}
        & \gamma_0\psi_{2,22}(0,0,i\kappa)=\psi_{2,12}(0,0,i\kappa) \\
        & \gamma_0\psi_{2,12}(0,0,i\kappa)=\psi_{2,22}(0,0,i\kappa),
        \end{aligned}
        \right.
\end{equation}
from which, we derive $\gamma_0^2=1$.
\end{proof}

Additionally, we discuss the symmetries of $M(x,t,k)$ defined by \eqref{piecewisematrixM} in the following proposition.
\begin{proposition}
    The piecewise matrix-valued function $M(x,t,k)$ satisfies the following symmetry reduction
        \begin{equation*}
            M(x,t,k):=\left\{
                \begin{aligned}
                &\Lambda\overline{M(-x,-t,-\bar{k})}\Lambda^{-1}\begin{pmatrix}\frac{1}{a_1(k)} & 0 \\ 0 & a_1(k) \end{pmatrix}, \quad k\in \mathbb{C}_{+}\backslash\{0\} \\
                &\Lambda\overline{M(-x,-t,-\bar{k})}\Lambda^{-1}\begin{pmatrix}a_2(k) & 0 \\ 0 & \frac{1}{a_2(k)} \end{pmatrix}, \quad k\in \mathbb{C}_{-}\backslash\{0\}.
                \end{aligned}
                \right.
        \end{equation*}
\end{proposition}
\begin{proof}
    Due to the symmetries of reflection coefficients $r_1(k)$ and $r_2(k)$ (see \eqref{reflectioncoeCondition1}-\eqref{reflectioncoeCondition2}), we have the symmetry for $J(x,t,k)$ with
    \begin{equation*}
        \begin{pmatrix}a_{2}(k) & 0 \\ 0 & \frac{1}{a_2(k)}\end{pmatrix}J(x,t,k)\begin{pmatrix}a_{1}(k) & 0 \\ 0 & \frac{1}{a_1(k)}\end{pmatrix}=\Lambda\overline{J(-x,-t,-k)}\Lambda^{-1}, \quad k\in\mathbb{R}\backslash\{0\}.
    \end{equation*}
    To match consistency with the the behavior of $M(x,t,k)$ at $k=0$ (see Proposition \ref{propMbehaviorat0}) and residue condition (see Proposition \ref{propMrescondition}), we reach the finale of this proposition.
\end{proof}

Summarize all the above contents in this subsection, we construct the following basic Riemann-Hilbert problem, which is the basis to construct the soliton solution and long-time asymptotics for the
focusing nonlocal MKdV equation.

\begin{RHP}[{\bf Basic Riemann-Hilbert Problem}]\label{basicRHP}
    Find a $2\times2$ matrix-valued function $M(x,t,k)$ such that
    \begin{itemize}
        \item[(i)] $M(x,t,k)$ is meromorphic for $k\in\mathbb{C}\backslash\mathbb{R}$ and has a simple pole located at $k=i\kappa$, $\kappa>0$.
        \item[(ii)] Jump conditions. The  The non-tangential limits $M_{\pm}(x,t,k)=\underset{k'\rightarrow k, k'\in\mathbb{C}_{\pm}}{\lim}M(x,t,k')$ exist for $k\in\mathbb{R}$ and $M_{\pm}(x,t,k)$ satisfy
        the jump condition $M_{+}(x,t,k)=M_{-}(x,t,k)J(x,t,k)$ for $k\in\mathbb{R}\backslash\{0\}$, where
        \begin{equation}\label{jumpforM}
            J(x,t,k)=\left(
        \begin{array}{cc}
        1+r_1(k)r_2(k) &  r_2(k)e^{-2ikx-8ik^3t}\\
        r_1(k)e^{2ikx+8ik^3t} &  1
        \end{array}
        \right), \quad k\in\mathbb{R}\backslash\{0\}.
        \end{equation}
        with $r_1$ and $r_2$ given in terms of $b$ by \eqref{defofreflectioncoe} with \eqref{a1a2CaseI} (Case I) or \eqref{a1a2CaseII} (Case II).
        \item[(iii)] Normalization condition at $k=\infty$. $M(x,t,k)=I+O(k^{-1})$ uniformly as $k\rightarrow\infty$.
        \item[(iv)] Residue condition at $k=i\kappa$.
        \begin{equation}\label{rescondition of M}
            \underset{k=i\kappa}{\rm Res}M^{(1)}(x,t,k)=\frac{\gamma_0}{a_{1}'(i\kappa)}e^{-2\kappa x+8\kappa^3t}M^{(2)}(x,t,i\kappa), \quad \gamma_0^2=1,
        \end{equation}
        where $\kappa$ could be exhibited in terms of $b$ using \eqref{kappageneric} (Case I) or \eqref{kappanongeneric} (Case II).
        \item[(v)] Singularities at the origin. As $k\rightarrow 0$, $M(x,t,k)$ satisfies

        For Case I
    \begin{subequations}\label{singularity of M at k=0 CaseI}
        \begin{align}
            &M_{+}(x,t,k)=\left(
                    \begin{array}{cc}
                    \frac{4v_1(x,t)}{A^2a_2(0)} & -\overline{v_2(-x,-t)} \\
                    \frac{4v_2(x,t)}{A^2a_2(0)} & -\overline{v_1(-x,-t)}
                    \end{array}
                \right)\left(I+O(k)\right)\left(
                    \begin{array}{cc}
                    k & 0 \\
                    0 & \frac{1}{k}
                    \end{array}
                \right), &\mathbb{C}_{+}\ni k\rightarrow 0,\label{singularity of M at k=+i0 CaseI}\\
            &M_{-}(x,t,k)=\frac{2i}{A}\left(
                \begin{array}{cc}
                -\overline{v_2(-x,-t)} &  \frac{v_1(x,t)}{a_2(0)}\\
                -\overline{v_1(-x,-t)} &  \frac{v_2(x,t)}{a_2(0)}
                \end{array}
                \right)+O(k), &\mathbb{C}_{-}\ni k\rightarrow 0.\label{singularity of M at k=-i0 CaseI}
        \end{align}
    \end{subequations}

        For Case II
        \begin{subequations}\label{singularity of M at k=0 CaseII}
            \begin{align}
                &M_{+}(x,t,k)=\left(
                        \begin{array}{cc}
                        \frac{v_1(x,t)}{a_{11}} & -\overline{v_2(-x,-t)} \\
                        \frac{v_2(x,t)}{a_{11}} & -\overline{v_1(-x,-t)}
                        \end{array}
                    \right)\left(I+O(k)\right)\left(
                        \begin{array}{cc}
                        1 & 0 \\
                        0 & \frac{1}{k}
                        \end{array}
                    \right), &\mathbb{C}_{+}\ni k\rightarrow 0,\label{singularity of M at k=+i0 CaseII}\\
                &M_{-}(x,t,k)=\frac{2i}{A}\left(
                    \begin{array}{cc}
                        -\overline{v_2(-x,-t)} & \frac{v_1(x,t)}{a_{2}'(0)} \\
                        -\overline{v_1(-x,-t)} & \frac{v_2(x,t)}{a_{2}'(0)}
                    \end{array}
                \right)\left(I+O(k)\right)\left(
                    \begin{array}{cc}
                    1 & 0 \\
                    0 & \frac{1}{k}
                    \end{array}
                \right), &\mathbb{C}_{-}\ni k\rightarrow 0.\label{singularity of M at k=-i0 CaseII}
            \end{align}
        \end{subequations}
        where $v_j(x,t)$, $j=1,2$ are some functions.
    \item[(vi)] Symmetry. $M(x,t,k)$ satisfies the following symmetry reduction
    \begin{equation}
        M(x,t,k):=\left\{
            \begin{aligned}
            &\Lambda\overline{M(-x,-t,-\bar{k})}\Lambda^{-1}\begin{pmatrix}\frac{1}{a_1(k)} & 0 \\ 0 & a_1(k) \end{pmatrix}, \quad k\in \mathbb{C}_{+}\backslash\{0\} \\
            &\Lambda\overline{M(-x,-t,-\bar{k})}\Lambda^{-1}\begin{pmatrix}a_2(k) & 0 \\ 0 & \frac{1}{a_2(k)} \end{pmatrix}, \quad k\in \mathbb{C}_{-}\backslash\{0\}.
            \end{aligned}
            \right.
    \end{equation}
    \end{itemize}
\end{RHP}
Assume that the RH problem \ref{basicRHP} exists a solution $M(x,t,k)$, then the Cauchy problem for nonlinear nonlocal MKdV equation \eqref{Cauchy Problem} with initial data $u_0$ can be
expressed in terms of the $(1,2)$-entry and $(2,1)$-entry of $M$ as the following equalities
\begin{equation}\label{uintermsofM}
    u(x,t)=2i\lim_{k\rightarrow\infty}kM_{12}(x,t,k), \quad u(-x,-t)=2i\lim_{k\rightarrow\infty}kM_{21}(x,t,k),
\end{equation}
which are direct results following from \eqref{uintermsofpsi}.

\begin{remark}
    \rm The solution of the RH problem \ref{basicRHP} is unique, if exists. In fact, if we assume there exist two solutions $M$ and $\tilde{M}$, the behaviors of $M$ at $k=0$ show that $M\tilde{M}^{-1}$
    is bounded at the origin. To reach $M\tilde{M}^{-1}\equiv I$, we just need to obey the standard  procedure  based on the Liouville theorem.
\end{remark}

\begin{remark}
    \rm As long as we derive the long-time ($0<t\rightarrow\infty$) asymptotics for the nonlocal MKdV equation with $x\in\mathbb{R}$, we can obtain the negative long-time ($0>t\rightarrow-\infty$) asymptotics
    and vice versa. This a typical characteristic of nonlocal MKdV equation.
\end{remark}

\subsection{Jump factorizations}\label{subsectionjumpfac}
To investigate the long-time asymptotics via nonlinear descent method, factorizations of jump matrix \eqref{jumpforM} play an important role in our analysis.

Introduction the phase function
\begin{equation}
    \theta(k,\xi):=4k^3+12k\xi,
\end{equation}
where $\xi:=\frac{x}{12t}$. For $\xi<0$, we define that
\begin{equation}
    \pm k_{0}=\pm\sqrt{-\frac{x}{12t}},
\end{equation}
which are the saddle points of phase function $\theta(k,\xi)$ and locate on the real axis. For $\xi>0$, we solve the saddle points
of $\theta(k,\xi)$ as follows
\begin{equation}
    \pm k_{0}=\pm i\sqrt{\frac{x}{12t}}.
\end{equation}

\begin{figure}[htbp]
\begin{center}
\tikzset{every picture/.style={line width=0.75pt}} 
\begin{tikzpicture}[x=0.75pt,y=0.75pt,yscale=-1,xscale=1]
\draw    (125,173.5) -- (502,173.5) ;
\draw   (407,58.5) .. controls (358.85,135.86) and (359.83,212.6) .. (409.93,288.71) ;
\draw   (197,290.49) .. controls (246.95,212.31) and (245.86,134.82) .. (193.71,58.03) ;
\draw (381.62,179.48) node [anchor=north west][inner sep=0.75pt]   {$k_{0}$};
\draw (166,113.4) node [anchor=north west][inner sep=0.75pt]    {$+$};
\draw (294,97.4) node [anchor=north west][inner sep=0.75pt]    {$-$};
\draw (420,117.4) node [anchor=north west][inner sep=0.75pt]    {$+$};
\draw (423,230.4) node [anchor=north west][inner sep=0.75pt]    {$-$};
\draw (169,232.4) node [anchor=north west][inner sep=0.75pt]    {$-$};
\draw (297,222.4) node [anchor=north west][inner sep=0.75pt]    {$+$};
\draw (206.62,179.48) node [anchor=north west][inner sep=0.75pt]    {$-k_{0}$};
\end{tikzpicture}
\caption{\small Signature table of the function $\im \theta(k,\xi)$ ($\xi<0$).}\label{FigCasexi<0}
\end{center}
\end{figure}

\begin{figure}[htbp]
\begin{center}
\tikzset{every picture/.style={line width=0.75pt}} 
\begin{tikzpicture}[x=0.75pt,y=0.75pt,yscale=-1,xscale=1]
\draw    (179,160) -- (440.01,159.3) ;
\draw   (235,49) .. controls (283,102.33) and (331,102.33) .. (379.01,49) ;
\draw   (378.98,281.09) .. controls (331.38,227.7) and (283.71,227.65) .. (235.98,280.92) ;
\draw (281,62.4) node [anchor=north west][inner sep=0.75pt]  [font=\small]  {$i\sqrt{3\xi }$};
\draw (278,242.4) node [anchor=north west][inner sep=0.75pt]  [font=\small]  {$-i\sqrt{3\xi }$};
\draw (405,131.4) node [anchor=north west][inner sep=0.75pt]    {$+$};
\draw (407,172.4) node [anchor=north west][inner sep=0.75pt]    {$-$};
\draw (333,271.4) node [anchor=north west][inner sep=0.75pt]    {$+$};
\draw (343,46.4) node [anchor=north west][inner sep=0.75pt]    {$-$};
\draw (317.62,110.48) node [anchor=north west][inner sep=0.75pt]  [font=\small]  {$k_{0} =i\sqrt{\xi }$};
\draw (318.62,199.48) node [anchor=north west][inner sep=0.75pt]  [font=\small]  {$-k_{0} =-i\sqrt{\xi }$};
\draw (306,106.4) node [anchor=north west][inner sep=0.75pt]    {$.$};
\draw (307,203.4) node [anchor=north west][inner sep=0.75pt]    {$.$};
\end{tikzpicture}
\caption{\small Signature table of the function $\im \theta(k,\xi)$ ($\xi>0$).}\label{FigCasexi>0}
\end{center}
\end{figure}

Since the phase function $\theta(\xi,k)$ is the same as in the case of local MKdV equation, its signature tables (see Fig \ref{FigCasexi<0} and Fig \ref{FigCasexi>0}) suggest us to follow
the standard procedure (see \cite{DZAnn}) to make factorizations of the jump matrix \eqref{jumpforM} along the real axis to deform the contours onto those
on which the oscillatory jump on the real axis is traded for exponential decay as $t\rightarrow+\infty$ or $t\rightarrow-\infty$. This step is aided by two well known factorizations
of the jump matrix $J(x,t,k)$:
\begin{subequations}
    \begin{align}
        &J(x,t,k)=\begin{pmatrix} 1 & r_2(k)e^{-2it\theta} \\ 0 & 1 \end{pmatrix}\begin{pmatrix} 1 & 0 \\r_1(k)e^{2it\theta} & 1\end{pmatrix},
    \end{align}
    as well as
    \begin{align}
        &J(x,t,k)=\begin{pmatrix}1 & 0 \\ \frac{r_1(k)}{1+r_1(k)r_2(k)}e^{2it\theta} & 1\end{pmatrix}
        \begin{pmatrix}1+r_1(k)r_2(k) & 0 \\ 0 & \frac{1}{1+r_1(k)r_2(k)}\end{pmatrix}
        \begin{pmatrix}1 &  \frac{r_2(k)}{1+r_1(k)r_2(k)}e^{-2it\theta} \\ 0 & 1\end{pmatrix}.
    \end{align}
\end{subequations}
Additionally, we should present more details for the behavior at the  singularity point $k=0$ when we do the RH deformations
in the long-time asymptotic analysis.

\section{One-soliton solution}\label{section1-soliton}
Before studying the long-time asymptotics for the solution of NMKdV equation, we firstly construct the one-soliton solution corresponding to the discrete spectrum $k=i\kappa$ under the reflectionless environment via
the RH problem \ref{basicRHP}. Some assumptions are given as
follows
\begin{Assumption}\label{assforsoliton}
    Suppose the spectral functions $a_1(k)$, $a_2(k)$ and $b(k)$ satisfy
    \begin{itemize}
        \item $a_1(k)$ has a single, simple zero $k=i\kappa$ with $\kappa>0$ in $\overline{\mathbb{C}_{+}}$.
        \item $a_2(k)$ has a single, simple zero located at $k=0$ in $\overline{\mathbb{C}_{-}}$.
        \item $b(k)=0$ for $k\in\mathbb{R}$.
    \end{itemize}
\end{Assumption}
Then, the following proposition gives the one-soliton solution associated with the Cauchy problem \eqref{Cauchy Problem}.
\begin{proposition}\label{one-soliton}
    Let $a_1(k)$, $a_2(k)$ and $b(k)$ be the spectral functions associated with $u_0(x)$ and admit the assumptions in Assumption \ref{assforsoliton}.  Then we have
    \begin{itemize}
        \item[(i)] $\kappa$ is uniquely determined as $\kappa=\frac{A}{2}$.
        \item[(ii)]The RH problem \ref{basicRHP} has a unique solution for all $(x,t)$ with $x\in\mathbb{R}$, $t\in\mathbb{R}$ except $t=x/A^2$ and $\gamma_0=1$.
        \item[(iii)] The associated exact one-soliton solution $u(x,t)$ of Cauchy problem \eqref{Cauchy Problem} is exhibited as
        \begin{equation}\label{onesolitonexpression}
            u(x,t)=\frac{A}{1-\gamma_0 e^{-Ax+A^3t}}.
        \end{equation}
    \end{itemize}
\end{proposition}
\begin{proof}
    Since $b(k)=0$ for $k\in\mathbb{R}$, thus $b(0)=0$. Due to $b(0)=0$, we are in Case II (non-generic case) to start our discussion for one-soliton solution. Use \eqref{kappanongeneric}, we arrive
    at item (i) immediately. Furthermore, from \eqref{a1a2CaseII} (see Case II in Appendix \ref{kappaCalculation}), we obtain that
    \begin{equation}
        a_1(k)=\frac{k-i\frac{A}{2}}{k}, \quad a_2(k)=\frac{k}{k-\frac{iA}{2}},
    \end{equation}
    and
    \begin{equation}
        a_{11}=\lim_{k\rightarrow\infty}ka_1(k)=\frac{A}{2i}, \quad a_2'(0)=\frac{2i}{A}.
    \end{equation}
\end{proof}
Observing that $b(k)=0$ for $k\in\mathbb{R}$ (i.e., $J(x,t,k)=1$ for $k\in\mathbb{R}$), we find that $M(x,t,k)$ is a meromorphic function in $k$ complex plane with the only pole $k=i\frac{A}{2}$.
Due to this observation, then we compare that \eqref{singularity of M at k=+i0 CaseII} and \eqref{singularity of M at k=-i0 CaseII}, we conclude that $v_1(x,t)=-\overline{v_2(-x,-t)}$ and convert
the singularity condition \eqref{singularity of M at k=0 CaseII} into a formal residue condition
\begin{equation}\label{resk=0forsoliton}
    \underset{k=0}{\rm Res}M^{(2)}(x,t,k)=\frac{A}{2i}M^{(1)}(x,t,0).
\end{equation}
Moreover, take into account the normalization condition of $M(x,t,k)$ at $k=\infty$ (see item (iii) of RH problem \ref{basicRHP}), we express
the $M(x,t,k)$ as follows:
\begin{equation}
    M(x,t,k)=\left(
        \begin{array}{cc}
            \frac{k+v_1(x,t)}{k-\frac{iA}{2}} & \frac{v_1(x,t)}{k} \\
            \frac{-\overline{v_1(-x,-t)}}{k-\frac{iA}{2}} & \frac{k-\overline{v_1(-x,-t)}}{k}.
        \end{array}
    \right)
\end{equation}
Use the residue condition \eqref{rescondition of M} at $k=i\frac{A}{2}$, we obtain that
\begin{equation}
    v_1(x,t)=\frac{A}{2i}\frac{1}{1-\gamma_0 e^{-Ax+A^3t}}
\end{equation}
At the end, use $u(x,t)=2i\underset{k\rightarrow\infty}{\ lim}kM_{12}$, the one-soliton solution \eqref{onesolitonexpression} follows.
\begin{remark}
\rm Use another potential recovering formulae $u(-x,-t)=2i\underset{k\rightarrow\infty}{\rm lim}kM_{21}$, we can solve that
\begin{equation*}
    u(-x,-t)=-2i\overline{v_1(-x,-t)}=\frac{A}{1-\gamma_0 e^{Ax-A^3t}},
\end{equation*}
which is also the direct result to replace $x$, $t$ by $-x$, $-t$ in \eqref{onesolitonexpression}.
\end{remark}

\begin{remark}
    \rm From the one-soliton solution formulae \eqref{onesolitonexpression}, we can find that $u(x,t)\rightarrow A$ as $x\rightarrow+\infty$, $u(x,t)\rightarrow 0$ as $x\rightarrow-\infty$,
    which are consistency to the boundary conditions \eqref{bdrycondition}.
\end{remark}

\section{Asymptotic behavior for $\xi:=x/(12t)<0$, $\vert \xi \vert=O(1)$}\label{sectionAANO1}
In this section, we investigate the long-time asymptotics under the condition $\xi<0$, $\vert \xi\vert=O(1)$, and the signature
table of this case corresponds to Fig \ref{FigCasexi<0}. Before our analysis, we point out that the analysis for $t\rightarrow-\infty$
instead of $t\rightarrow+\infty$ is more convenient to deal with the singularities of RH problem at $k=0$. See Remark \ref{whywedot-infty} below.

\subsection{First RH problem transformation: $\delta$ function}\label{subsectionFirstdeform}
At the beginning, we introduce the $\delta$ function as the solution of the scalar RH problem as follows
\begin{itemize}
    \item $\delta(k,\xi)$ is holomorphic for $k\in\mathbb{C}\backslash((-\infty,-k_0)\cup(k_0,+\infty))$,
    \item $\delta_{+}(k,\xi)=\delta_{-}(k,\xi)(1+r_1(k)r_2(k))$, $k\in\mathbb{R}\backslash[-k_0,k_0]$,
    \item normalization condition: $\delta(k,\xi)\rightarrow 1$ as $k\rightarrow\infty$.
\end{itemize}
And the solution can be given by the Cauchy-type integral
\begin{equation}\label{deltafunc}
    \delta(k,\xi)=\exp\left\{\frac{1}{2\pi i}\int_{(-\infty,-k_0)\cup(k_0,+\infty)}\frac{{\log \left(1+r_1(s)r_2(s)\right)}}{s-k}ds\right\}.
\end{equation}
By \eqref{reflectioncoeCondition1}, $\delta(k,\xi)$ admits the symmetry
\begin{equation}\label{deltasymmetry}
    \delta(k,\xi)=\overline{\delta(-\bar{k},\xi)},  \quad k\in\mathbb{C}\backslash\left\{0\right\}.
\end{equation}

\begin{remark}
    \rm
    In the general case (not the case of pure step initial data), $1+r_1(k)r_2(k)$ is complex-valued, which could make
    $\delta$ function be singular at $k=\pm k_0$.
\end{remark}

\begin{remark}\label{whywedot-infty}
    \rm If we consider the asymptotics as $t\rightarrow+\infty$, we have to use the following $\delta$ function
    \begin{equation*}
        \delta(k,\xi)=\exp\left\{\frac{1}{2\pi i}\int_{-k_0}^{k_0}\frac{{\log \left(1+r_1(s)r_2(s)\right)}}{s-k}ds\right\}.
    \end{equation*}
    In this function, we notice that $0\in(-k_0,k_0)$ thus we have to pay more attention to the behavior of
    $\delta$ at $k=0$. Luckily, owe to the symmetries of integrable nonlocal MKdV equation, we can study the large-negative-$t$ asymptotics firstly.
\end{remark}

Moreover, integrations by parts in formulae \eqref{deltafunc} yields
\begin{equation}\label{deltachi}
    \delta(k,\xi)=\frac{\left(k+k_0\right)^{i\nu(-k_0)}}{\left(k-k_0\right)^{i\overline{\nu(-k_0)}}}e^{\hat{\chi}(\xi,k)},
\end{equation}
where $\hat{\chi}(\xi,k)$ is a uniformly bounded function with respect to $|\xi|=O(1)$ and $k\in\mathbb{C}\backslash\mathbb{R}$ defined by
\begin{equation}
    \hat{\chi}(\xi,k)=-\frac{1}{2\pi i}\int_{(-\infty,-k_0)\cup(k_0,+\infty)}\log\left(k-s\right)d_{s}\log\left(1+r_1\left(s\right)r_2\left(s\right)\right).
\end{equation}
And $\nu(-k_0)$ can be expressed in terms of
\begin{equation}
    \nu(-k_0):=-\frac{1}{2\pi}\log(1+r_1(-k_0)r_2(-k_0))=-\frac{1}{2\pi}\log\vert1+r_1(-k_0)r_2(-k_0)\vert-\frac{i}{2\pi}\Delta(-k_0).
\end{equation}
with
\begin{align*}
    \Delta(-k_0):=\int_{-\infty}^{-k_0}d\arg\left(1+r_1(s)r_2(s)\right).
\end{align*}
Furthermore, we assume that
\begin{equation}
    -\pi<\Delta(k)<\pi, \quad \xi<0, \ |\xi|=O(1)
\end{equation}
and thus $-\frac{1}{2}<\im \nu(k)<\frac{1}{2}$. With this assumption, $\log(1+r_1(k)r_2(k))$ is single-valued,
consequently the singularities $k=\pm k_0$ of $\delta(k,\xi)$ is square integrable.

On the other hand, we can rewrite $\delta(k,\xi)$ as
\begin{equation}\label{deltahatchi}
    \delta(k,\xi)=\left(\frac{k+k_0}{k-k_0}\right)^{^{i\nu(-k_0)}}e^{\chi(\xi,k)},
\end{equation}
where
\begin{equation}
    \chi(\xi,k)=\hat{\chi}(\xi,k)+\frac{1}{2\pi i}\log\left(\frac{1+r_1(-k_0)r_2(-k_0)}{1+\overline{r_1(-k_0)r_2(-k_0)}}\right)\log(k-k_0).
\end{equation}
\begin{remark}\rm
    In the local MKdV equation, (see \cite{DZAnn,lenellsmkdv}), we can obtain that $\chi(k,\xi)$ is equivalent to $\hat{\chi}(k,\xi)$ by $r_1(k)=r(k)$, $r_2(k)=\overline{r(\bar{k})}$.
\end{remark}

With the help of $\delta(k,\xi)$, we define a new RH problem as follows
\begin{equation}
    \tilde{M}(x,t,k)=M(x,t,k)\delta^{-\sigma_3}(\xi,k).
\end{equation}
Then we have properties for $\tilde{M}(x,t,k)$ as follows:
\begin{RHP}\label{RHPtildeM}
    Find a $2\times2$ matrix-valued function $\tilde{M}(x,t,k)$ such that
    \begin{itemize}
        \item[(i)] $\tilde{M}(x,t,k)$ is meromorphic for $k\in\mathbb{C}\backslash\mathbb{R}$ and has a simple pole located at $k=i\kappa$, $\kappa>0$.
        \item[(ii)] Jump conditions. The  The non-tangential limits $\tilde{M}_{\pm}(x,t,k)=\underset{k'\rightarrow k, k'\in\mathbb{C}_{\pm}}{\lim}\tilde{M}(x,t,k')$ exist for $k\in\mathbb{R}$ and $\tilde{M}_{\pm}(x,t,k)$ satisfy
        the jump condition $\tilde{M}_{+}(x,t,k)=\tilde{M}_{-}(x,t,k)\tilde{J}(x,t,k)$ for $k\in\mathbb{R}\backslash\{0\}$, where
        \begin{equation}\label{jumpfortildeM}
            \tilde{J}(x,t,k)=\left\{
                \begin{aligned}
                &\begin{pmatrix}
                    1 & r_2(k)\delta^2(\xi,k)e^{-2it\theta}\\
                    0 & 1
                \end{pmatrix}
                \begin{pmatrix}
                    1 & 0\\
                    r_1(k)\delta^{-2}(\xi,k)e^{2it\theta} & 1
                \end{pmatrix}  \quad k\in(-k_0,k_0)\backslash\{0\} \\
                &\begin{pmatrix}
                    1 & 0\\
                    \frac{r_1(k)\delta_{-}^{-2}(k,\xi)}{1+r_1(k)r_2(k)}e^{2it\theta} & 1
                \end{pmatrix}
                \begin{pmatrix}
                    1 & \frac{r_2(k)\delta_{+}^{2}(k,\xi)}{1+r_1(k)r_2(k)}e^{-2it\theta}\\
                    0 & 1
                \end{pmatrix}, \quad k\in(-\infty,-k_0)\cup(k_0,+\infty)
                \end{aligned}
                \right.
        \end{equation}
        \item[(iii)] Normalization condition at $k=\infty$. $\tilde{M}(x,t,k)=I+O(k^{-1})$ uniformly as $k\rightarrow\infty$.
        \item[(iv)] Residue condition at $k=i\kappa$.
        \begin{equation}\label{rescondition of tildeM}
            \underset{k=i\kappa}{\rm Res}\tilde{M}^{(1)}(x,t,k)=\frac{\gamma_0}{a_{1}'(i\kappa)\delta^2(i\kappa,\xi)}e^{-2\kappa x+8\kappa^3t}\tilde{M}^{(2)}(x,t,i\kappa), \quad \gamma_0^2=1,
        \end{equation}
        \item[(v)] Singularities at the origin. As $k\rightarrow 0$, $M(x,t,k)$ satisfies

        For Case I
    \begin{subequations}\label{singularity of tildeM at k=0 CaseI}
        \begin{align}
            &\tilde{M}_{+}(x,t,k)=\left(
                    \begin{array}{cc}
                    \frac{4v_1(x,t)}{A^2a_2(0)\delta(0,\xi)} & -\delta(0,\xi)\bar{v}_2(-x,-t) \\
                    \frac{4v_2(x,t)}{A^2a_2(0)\delta(0,\xi)} & -\delta(0,\xi)\bar{v}_1(-x,-t)
                    \end{array}
                \right)\left(I+O(k)\right)\left(
                    \begin{array}{cc}
                    k & 0 \\
                    0 & \frac{1}{k}
                    \end{array}
                \right), &\mathbb{C}_{+}\ni k\rightarrow 0,\label{singularity of tildeM at k=+i0 CaseI}\\
            &\tilde{M}_{-}(x,t,k)=\frac{2i}{A}\left(
                \begin{array}{cc}
                    \frac{\bar{v}_2(-x,-t)}{\delta(0,\xi)} &  \frac{v_1(x,t)}{a_2(0)}\delta(0,\xi)\\
                    \frac{\bar{v}_1(-x,-t)}{\delta(0,\xi)} &  \frac{v_2(x,t)}{a_2(0)}\delta(0,\xi)
                \end{array}
                \right)+O(k), &\mathbb{C}_{-}\ni k\rightarrow 0.\label{singularity of tildeM at k=-i0 CaseI}
        \end{align}
    \end{subequations}

        For Case II
        \begin{subequations}\label{singularity of tildeM at k=0 CaseII}
            \begin{align}
                &\tilde{M}_{+}(x,t,k)=\left(
                        \begin{array}{cc}
                        \frac{v_1(x,t)}{a_{11}\delta(0,\xi)} & -\bar{v}_2(-x,-t)\delta(0,\xi) \\
                        \frac{v_2(x,t)}{a_{11}\delta(0,\xi)} & -\bar{v}_1(-x,-t)\delta(0,\xi)
                        \end{array}
                    \right)\left(I+O(k)\right)\left(
                        \begin{array}{cc}
                        1 & 0 \\
                        0 & \frac{1}{k}
                        \end{array}
                    \right), &\mathbb{C}_{+}\ni k\rightarrow 0,\label{singularity of tildeM at k=+i0 CaseII}\\
                &\tilde{M}_{-}(x,t,k)=\frac{2i}{A}\left(
                    \begin{array}{cc}
                        -\frac{\bar{v}_2(-x,-t)}{\delta(0,\xi)} & \frac{v_1(x,t)}{a_{2}'(0)}\delta(0,\xi) \\
                        -\frac{\bar{v}_1(-x,-t)}{\delta(0,\xi)} & \frac{v_2(x,t)}{a_{2}'(0)}\delta(0,\xi)
                    \end{array}
                \right)\left(I+O(k)\right)\left(
                    \begin{array}{cc}
                    1 & 0 \\
                    0 & \frac{1}{k}
                    \end{array}
                \right), &\mathbb{C}_{-}\ni k\rightarrow 0.\label{singularity of tildeM at k=-i0 CaseII}
            \end{align}
        \end{subequations}
        where $v_j(x,t)$, $j=1,2$ are some unprescribed functions.
    \end{itemize}
\end{RHP}

\subsection{Second RH problem transformation: opening lens}\label{subsectionSeconddeform}
In this subsection, we execute the standard procedure so called ``opening lens" to make some jump matrices decay to identity $I$ as the parameter $t\rightarrow-\infty$.
In general, the RH deformations for $\tilde{M}(x,t,k)$ depend on the properties of reflection coefficients $r_{j}(k)$. In classical Deift-Zhou method, $r_{j}(k)$ and
$\frac{r_j(k)}{1+r_1(k)r_2(k)}$ have to be analytically approximated by some rational functions with well-behaved errors (see \cite{DZAnn}). Besides, $r_{j}(k)$ and
$\frac{r_j(k)}{1+r_1(k)r_2(k)}$ could be continuously extended via the $\bar{\partial}$ technique, which is developed by McLaughlin and his collaborators (see \cite{M&M2006,M&M2008,Dieng}).

For the sake of clarity, we assume that the initial data $u_0(x)$ are a compact perturbation of pure step initial function $u_{0A}(x)$ defined by \eqref{purestepinitial}, which
ensures that eigenfunctions $\psi_j^{l}(x,0,k)$, $j,l=1,2$ and thus $r_j(k)$ are meromorphic in $\mathbb{C}$ (see \cite{RDJDE2021}). Then we open lens (see Fig \ref{Figxi<0openlens})
via defining $\breve{M}(x,t,k)$ by
\begin{equation}\label{breveM}
    \breve{M}(x,t,k)=\left\{
        \begin{aligned}
        &\tilde{M}(x,t,k)
        \begin{pmatrix}
            1 & -\frac{r_2(k)\delta^{2}(\xi,k)}{1+r_1(k)r_2(k)}e^{-2it\theta}\\
            0 & 1
        \end{pmatrix}, & k\in U_1\cup U_3, \\
        &\tilde{M}(x,t,k)
        \begin{pmatrix}
            1 & 0\\
            \frac{r_1(k)\delta^{-2}(\xi,k)}{1+r_1(k)r_2(k)}e^{2it\theta} & 1
        \end{pmatrix}, & k\in U^*_1\cup U_3^*, \\
        &\tilde{M}(x,t,k)
        \begin{pmatrix}
            1 & 0\\
            -r_1(k)\delta^{-2}(\xi,k)e^{2it\theta} & 1
        \end{pmatrix}, & k\in U_2, \\
        &\tilde{M}(x,t,k)
        \begin{pmatrix}
            1 & r_2(k)\delta^{2}(\xi,k)e^{-2it\theta}\\
            0 & 1
        \end{pmatrix}, & k\in U^{*}_2, \\
        &\tilde{M}(x,t,k), & k\in U_0\cup U_0^*.
        \end{aligned}
        \right.
\end{equation}
Here tha angles between the real axis and rays $\Gamma_j=\Gamma_j(\xi)$ are given such that the discrete spectrum $i\kappa$ is located on the sector $U_0$.
The RH problem $\breve{M}(x,t,k)$ have properties as follows

\begin{figure}[htbp]
\begin{center}
    \tikzset{every picture/.style={line width=0.75pt}} 
    \begin{tikzpicture}[x=0.75pt,y=0.75pt,yscale=-1,xscale=1]
    \draw  [dash pattern={on 0.84pt off 2.51pt}]  (51.99,161.5) -- (168.99,161.21) ;
    \draw  [dash pattern={on 0.84pt off 2.51pt}]  (168.99,161.21) -- (450.99,162.93) ;
    \draw  [dash pattern={on 0.84pt off 2.51pt}]  (450.99,162.93) -- (567.99,162.63) ;
    \draw    (450.99,162.93) -- (599.29,94.32) ;
    \draw [shift={(530.59,126.1)}, rotate = 155.17] [color={rgb, 255:red, 0; green, 0; blue, 0 }  ][line width=0.75]    (10.93,-3.29) .. controls (6.95,-1.4) and (3.31,-0.3) .. (0,0) .. controls (3.31,0.3) and (6.95,1.4) .. (10.93,3.29)   ;
    \draw    (450.99,162.93) -- (309.29,90.32) ;
    \draw [shift={(386.37,129.82)}, rotate = 207.13] [color={rgb, 255:red, 0; green, 0; blue, 0 }  ][line width=0.75]    (10.93,-3.29) .. controls (6.95,-1.4) and (3.31,-0.3) .. (0,0) .. controls (3.31,0.3) and (6.95,1.4) .. (10.93,3.29)   ;
    \draw    (168.99,161.21) -- (309.29,90.32) ;
    \draw [shift={(244.5,123.06)}, rotate = 153.19] [color={rgb, 255:red, 0; green, 0; blue, 0 }  ][line width=0.75]    (10.93,-3.29) .. controls (6.95,-1.4) and (3.31,-0.3) .. (0,0) .. controls (3.31,0.3) and (6.95,1.4) .. (10.93,3.29)   ;
    \draw    (450.99,162.93) -- (311.01,229.02) ;
    \draw [shift={(375.57,198.54)}, rotate = 334.73] [color={rgb, 255:red, 0; green, 0; blue, 0 }  ][line width=0.75]    (10.93,-3.29) .. controls (6.95,-1.4) and (3.31,-0.3) .. (0,0) .. controls (3.31,0.3) and (6.95,1.4) .. (10.93,3.29)   ;
    \draw    (597.29,239.32) -- (450.99,162.93) ;
    \draw [shift={(518.82,198.35)}, rotate = 27.57] [color={rgb, 255:red, 0; green, 0; blue, 0 }  ][line width=0.75]    (10.93,-3.29) .. controls (6.95,-1.4) and (3.31,-0.3) .. (0,0) .. controls (3.31,0.3) and (6.95,1.4) .. (10.93,3.29)   ;
    \draw    (311.01,229.02) -- (168.99,161.21) ;
    \draw [shift={(234.58,192.53)}, rotate = 25.52] [color={rgb, 255:red, 0; green, 0; blue, 0 }  ][line width=0.75]    (10.93,-3.29) .. controls (6.95,-1.4) and (3.31,-0.3) .. (0,0) .. controls (3.31,0.3) and (6.95,1.4) .. (10.93,3.29)   ;
    \draw    (28.29,85.32) -- (168.99,161.21) ;
    \draw [shift={(103.92,126.11)}, rotate = 208.34] [color={rgb, 255:red, 0; green, 0; blue, 0 }  ][line width=0.75]    (10.93,-3.29) .. controls (6.95,-1.4) and (3.31,-0.3) .. (0,0) .. controls (3.31,0.3) and (6.95,1.4) .. (10.93,3.29)   ;
    \draw    (168.99,161.21) -- (44.29,223.32) ;
    \draw [shift={(101.27,194.94)}, rotate = 333.52] [color={rgb, 255:red, 0; green, 0; blue, 0 }  ][line width=0.75]    (10.93,-3.29) .. controls (6.95,-1.4) and (3.31,-0.3) .. (0,0) .. controls (3.31,0.3) and (6.95,1.4) .. (10.93,3.29)   ;
    \draw  [dash pattern={on 0.84pt off 2.51pt}]  (309.29,50.32) -- (309.99,162.07) ;
    \draw (159.99,169.92) node [anchor=north west][inner sep=0.75pt]  [rotate=-0.03]  {$-k_{0}$};
    \draw (445.99,173.07) node [anchor=north west][inner sep=0.75pt]  [rotate=-0.03]  {$k_{0}$};
    \draw (517,97.4) node [anchor=north west][inner sep=0.75pt]   {$\Gamma _{1}$};
    \draw (518,211.4) node [anchor=north west][inner sep=0.75pt]    {$\Gamma _{1}^{*}$};
    \draw (536,135.4) node [anchor=north west][inner sep=0.75pt]    {$U_{1}$};
    \draw (537,177.4) node [anchor=north west][inner sep=0.75pt]    {$U_{1}^{*}$};
    \draw (345,126.4) node [anchor=north west][inner sep=0.75pt]    {$U_{2}$};
    \draw (343,177.4) node [anchor=north west][inner sep=0.75pt]    {$U_{2}^{*}$};
    \draw (67,134.4) node [anchor=north west][inner sep=0.75pt]    {$U_{3}$};
    \draw (68,173.4) node [anchor=north west][inner sep=0.75pt]    {$U_{3}^{*}$};
    \draw (399,115.4) node [anchor=north west][inner sep=0.75pt]    {$\Gamma _{2}$};
    \draw (206,116.4) node [anchor=north west][inner sep=0.75pt]    {$\Gamma _{3}$};
    \draw (102,100.4) node [anchor=north west][inner sep=0.75pt]   {$\Gamma _{4}$};
    \draw (389,194.4) node [anchor=north west][inner sep=0.75pt]    {$\Gamma _{2}^{*}$};
    \draw (209,196.4) node [anchor=north west][inner sep=0.75pt]  {$\Gamma _{3}^{*}$};
    \draw (105.64,196.66) node [anchor=north west][inner sep=0.75pt]   {$\Gamma _{4}^{*}$};
    \draw (257,44.4) node [anchor=north west][inner sep=0.75pt]   {$U_{0}$};
    \draw (258,252.4) node [anchor=north west][inner sep=0.75pt]   {$U_{0}^{*}$};
    \draw (315,41.4) node [anchor=north west][inner sep=0.75pt]   {$i\kappa $};
    \draw (303,168.4) node [anchor=north west][inner sep=0.75pt]    {$0$};
    \end{tikzpicture}
    \caption{\small The contours $\Gamma:=\Gamma_j\cup\Gamma^{*}_j$, $j=1,2,3,4$ and the regions $U_j$, $U_j^{*}$, $j=0,1,2,3$}\label{Figxi<0openlens}.
\end{center}
\end{figure}
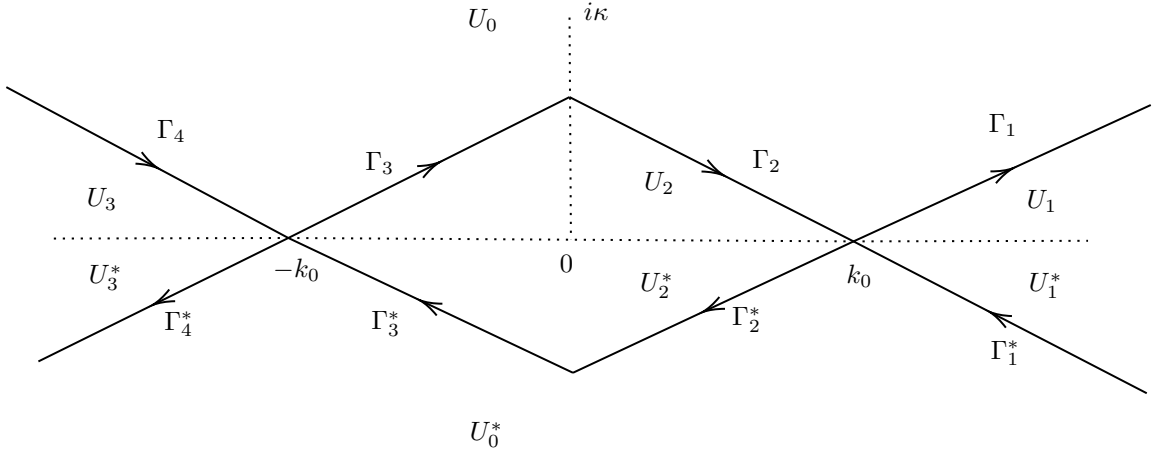

\begin{RHP}\label{RHPbreveM}
    Find a $2\times2$ matrix-valued function $\breve{M}(x,t,k)$ such that
    \begin{itemize}
        \item[(i)] $\breve{M}(x,t,k)$ is meromorphic for $k\in\mathbb{C}\backslash\Gamma$ and has a simple pole located at $k=i\kappa$, $\kappa>0$.
        \item[(ii)] Jump conditions. The  The non-tangential limits $\breve{M}_{\pm}(x,t,k)=\underset{k'\rightarrow k, k'\in\mathbb{C}_{\pm}}{\lim}\breve{M}(x,t,k')$ exist for
        $k\in\Gamma:=(\cup_{j=1}^{4}\Gamma_j)\cup(\cup_{j=1}^{4}\Gamma^*_j)$ and $\breve{M}_{\pm}(x,t,k)$ satisfy
        the jump condition $\breve{M}_{+}(x,t,k)=\breve{M}_{-}(x,t,k)\breve{J}(x,t,k)$ for $k\in\Gamma$, where
        \begin{equation}\label{jumpforbreveM}
            \breve{J}(x,t,k)=\left\{
                \begin{aligned}
                &\begin{pmatrix}
                    1 & \frac{r_2(k)\delta^{2}(k,\xi)}{1+r_1(k)r_2(k)}e^{-2it\theta}\\
                    0 & 1
                \end{pmatrix}  & k\in\Gamma_1\cup\Gamma_4,\\
                &\begin{pmatrix}
                    1 & 0\\
                    -\frac{r_1(k)\delta^{-2}(k,\xi)}{1+r_1(k)r_2(k)}e^{2it\theta} & 1
                \end{pmatrix}, & k\in\Gamma_1^*\cup\Gamma_4^*, \\
                &\begin{pmatrix}
                    1 & 0\\
                    r_1(k)\delta^{-2}(\xi,k)e^{2it\theta} & 1
                \end{pmatrix}, & k\in\Gamma_2\cup\Gamma_3, \\
                &\begin{pmatrix}
                    1 & -r_2(k)\delta^2(\xi,k)e^{-2it\theta}\\
                    0 & 1
                \end{pmatrix} , & k\in\Gamma^*_2\cup\Gamma^*_3,\\
                &I, & elsewhere.
                \end{aligned}
                \right.
        \end{equation}
        \item[(iii)] Normalization condition at $k=\infty$. $\breve{M}(x,t,k)=I+O(k^{-1})$ as $k\rightarrow\infty$.
        \item[(iv)] Residue condition at $k=i\kappa$.
        \begin{equation}\label{rescondition of breveM}
            \underset{k=i\kappa}{\rm Res}\breve{M}^{(1)}(x,t,k)=c_1(x,t)\breve{M}^{(2)}(x,t,i\kappa),
        \end{equation}
        with $c_1(x,t)=\frac{\gamma_0}{a_{1}'(i\kappa)\delta^2(i\kappa,\xi)}e^{-2\kappa x+8\kappa^3t}=\frac{\gamma_0}{a_{1}'(i\kappa)\delta^2(i\kappa,\xi)}e^{t(-24\kappa\xi+8\kappa^3)}$, $\gamma_0^2=1$.
        \item[(v)] Singularities at the origin. In both cases (Case I and II), as $k\rightarrow 0$, $M(x,t,k)$ satisfies

    \begin{subequations}\label{singularity of breveM at k=0 BothCases}
        \begin{align}
            &\breve{M}_{+}(x,t,k)=\left(
                    \begin{array}{cc}
                    -\frac{2i\bar{v}_2(-x,-t)}{A\delta(0,\xi)}+O(k) & -\frac{1}{k}\bar{v}_2(-x,-t)\delta(0,\xi)+O(1) \\
                    -\frac{2i\bar{v}_1(-x,-t)}{A\delta(0,\xi)}+O(k) & -\frac{1}{k}\bar{v}_1(-x,-t)\delta(0,\xi)+O(1)
                    \end{array}
                \right), &\mathbb{C}_{+}\ni k\rightarrow 0,\label{singularity of breveM at k=+i0 BothCases}\\
            &\breve{M}_{-}(x,t,k)=\frac{2i}{A}\left(
                \begin{array}{cc}
                    -\frac{\bar{v}_2(-x,-t)}{\delta(0,\xi)}+O(k) &  -\frac{A}{2ik}\delta(0,\xi)\bar{v}_2(-x,-t)+O(k)\\
                    -\frac{\bar{v}_1(-x,-t)}{\delta(0,\xi)}+O(k) &  -\frac{A}{2ik}\delta(0,\xi)\bar{v}_2(-x,-t)+O(k)
                \end{array}
                \right), &\mathbb{C}_{-}\ni k\rightarrow 0.\label{singularity of breveM at k=-i0 BothCases}
        \end{align}
    Furthermore, we can see that the singularity conditions at origin can be formally reduced to
    \begin{align}\label{singularity of breveM as ResCondition BothCases}
        \underset{k=0}{\rm Res}\breve{M}^{(2)}(x,t,k)=c_0(\xi)\breve{M}^{(1)}(x,t,0),
    \end{align}
    with $c_0(\xi)=\frac{A\delta^2(0,\xi)}{2i}$.
\end{subequations}
\end{itemize}
\end{RHP}

\begin{proof}
    The items (i)-(iv) are not difficult to check. For the singular behavior of $\breve{M}(x,t,k)$ at $k=0$, we exhibit more details below.

    Notice that $0\in \overline{U_2}\cup\overline{U^*_2}$, thus we should investigate the singular behavior of $r_j(k)$ at $k=0$, $j=1,2$.

    For Case I, item (v) of Proposition \ref{propspecfuncs}, i.e., $a_1(k)=\frac{A^2a_2(0)}{4k^2}+O(k^{-1})$, $b(k)=\frac{Aa_2(0)}{2ik}+O(1)$ as $k\rightarrow 0$  imply that
    \begin{subequations}
        \begin{align}
        &-r_1(k)\delta^{-2}(k,\xi)e^{2it\theta}=\frac{2i}{A\delta^2(0,\xi)}k+O(k^2), & U_2\ni k\rightarrow 0, \\
        &r_2(k)\delta^{2}(k,\xi)e^{-2it\theta}=\frac{A\delta^2(0,\xi)}{2ik}+O(1), & U^*_2\ni k\rightarrow 0.
        \end{align}
    \end{subequations}
    Then by \eqref{breveM} for $k\in U_2$ and $k\in U^*_2$ respectively, we derive \eqref{singularity of breveM at k=0 BothCases} in Case I.

    For Case II, review that $a_1(k)=\frac{a_{11}}{k}+O(1)$ and $a_2(k)=ka_2'(0)+O(k^2)$ as $k\rightarrow 0$, we obtain
    \begin{subequations}
        \begin{align}
            &-r_1(k)\delta^{-2}(k,\xi)e^{2it\theta}=-\frac{b(0)}{a_{11}\delta^{2}(0,\xi)}k+O(k^2), & U_2\ni k\rightarrow 0, \\
            &r_2(k)\delta^{2}(k,\xi)e^{-2it\theta}=\frac{b(0)}{ka_2'(0)}\delta(0,\xi)+O(1), & U^*_2\ni k\rightarrow 0.
        \end{align}
    \end{subequations}
    Then by \eqref{breveM} for $k\in U_2$ and $k\in U^*_2$ respectively, we derive
    \begin{subequations}\label{singularity of breveM at k=0 CaseII}
        \begin{align}
            &\breve{M}_{+}(x,t,k)=\left(
                    \begin{array}{cc}
                    \frac{v_1(x,t)+b(0)\bar{v}_2(-x,-t)}{a_{11}\delta(0,\xi)}+O(k) & -\frac{\bar{v}_2(-x,-t)}{k}\delta(0,\xi)+O(1) \\
                    \frac{v_2(x,t)+b(0)\bar{v}_1(-x,-t)}{a_{11}\delta(0,\xi)}+O(k) & -\frac{\bar{v}_1(-x,-t)}{k}\delta(0,\xi)+O(1)
                    \end{array}
                \right), &\mathbb{C}_{+}\ni k\rightarrow 0,\label{singularity of breveM at k=+i0 CaseII}\\
            &\breve{M}_{-}(x,t,k)=\frac{2i}{A}\left(
                \begin{array}{cc}
                    -\frac{\bar{v}_2(-x,-t)}{\delta(0,\xi)}+O(k) & \frac{v_1(x,t)-b(0)\bar{v}_2(-x,-t)}{ka_2'(0)}\delta(0,\xi)+O(1) \\
                    -\frac{\bar{v}_1(-x,-t)}{\delta(0,\xi)}+O(k) & \frac{v_2(x,t)-b(0)\bar{v}_1(-x,-t)}{ka_2'(0)}\delta(0,\xi)+O(1)
                \end{array}
            \right), &\mathbb{C}_{-}\ni k\rightarrow 0.\label{singularity of breveM at k=-i0 CaseII}
        \end{align}
    \end{subequations}
Comparing \eqref{singularity of breveM at k=+i0 CaseII} with \eqref{singularity of breveM at k=-i0 CaseII} (cf.\eqref{resk=0forsoliton}), we have two equations as follows
\begin{subequations}\label{v1v2forresconlinearequs}
    \begin{align}
        &\bar{v}_1(x,t)=-\left(\frac{2i}{A}+b(0)\right)\bar{v}_2(-x,-t)\\
        &\bar{v}_1(x,t)=\left(b(0)-\frac{Aa_2'(0)}{2i}\right)\bar{v}_2(-x,-t),
    \end{align}
\end{subequations}
    Indeed, we can find that the two equations in \eqref{v1v2forresconlinearequs} are the same one by taking into account \eqref{AppendixCala_{11}CaseII},
    from which, we derive that \eqref{singularity of breveM at k=0 BothCases} in Case II.
\end{proof}

\subsection{Regular RH problem: leading term of the asymptotics}\label{subsectionregularRHP}
In this subsection, we convert the RH problem \ref{RHPbreveM} with two formal residue conditions \eqref{rescondition of breveM} and \eqref{singularity of breveM as ResCondition BothCases}
into a regular RH problem without residue conditions via using the Blaschke-Potapov factors (reference, e.g., \cite[Proposition. 6]{RDJDE2021}, \cite{HamilMintheTheoryofSolitons}).

Make the transformation
\begin{equation}\label{transforregularrhp}
    \breve{M}=B(x,t,k)\breve{M}^{r}(x,t,k)\begin{pmatrix}
        1 & 0 \\
        0 & \frac{k-i\kappa}{k}
    \end{pmatrix}, \quad k\in\mathbb{C},
\end{equation}
where $B(x,t,k)=I+\frac{i\kappa}{k-i\kappa}P(x,t)$. And $B(x,t,k)$, $P(x,t)$ are the so called Blaschke-Potapov factors. Next RH problem
reveals the relation of regular RH problem $\breve{M}^{r}$ and the Blaschke-Potapov factors.
\begin{RHP}\label{RHPregular}
    Find a $2\times2$ matrix-valued function $\breve{M}^{r}(x,t,k)$ such that
\begin{itemize}
    \item[(i)] $\breve{M}^{r}(x,t,k)$ is analytic for $k\in\mathbb{C}\backslash\Gamma$.
    \item[(ii)]  $\breve{M}^{r}_{+}(x,t,k)=\breve{M}^{r}_{-}(x,t,k)\breve{J}^{r}(x,t,k)$ with
    \begin{equation}
        \breve{J}^{r}(x,t,k)
        \begin{pmatrix}
            1 & 0\\
            0 & \frac{k-i\kappa}{k}
        \end{pmatrix}
        \breve{J}(x,t,k)
        \begin{pmatrix}
            1 & 0\\
            0 & \frac{k}{k-i\kappa}
        \end{pmatrix}.
    \end{equation}
    More specifically,
    \begin{equation}\label{jumpforbreveMr}
        \breve{J}^r(x,t,k)=\left\{
            \begin{aligned}
            &\begin{pmatrix}
                1 & \frac{{r}^{r}_2(k)\delta^{2}(k,\xi)}{1+{r}^{r}_1(k){r}^{r}_2(k)}e^{-2it\theta}\\
                0 & 1
            \end{pmatrix}  & k\in\Gamma_1\cup\Gamma_4,\\
            &\begin{pmatrix}
                1 & 0\\
                -\frac{{r}^{r}_1(k)\delta^{-2}(k,\xi)}{1+{r}^{r}_1(k){r}^{r}_2(k)}e^{2it\theta} & 1
            \end{pmatrix}, & k\in\Gamma_1^*\cup\Gamma_4^*, \\
            &\begin{pmatrix}
                1 & 0\\
                {r}^{r}_1(k)\delta^{-2}(\xi,k)e^{2it\theta} & 1
            \end{pmatrix}, & k\in\Gamma_2\cup\Gamma_3, \\
            &\begin{pmatrix}
                1 & -{r}^{r}_2(k)\delta^2(\xi,k)e^{-2it\theta}\\
                0 & 1
            \end{pmatrix} , & k\in\Gamma^*_2\cup\Gamma^*_3.
            \end{aligned}
            \right.
    \end{equation}
    with
    \begin{equation}\label{rtor^r}
        {r}^{r}_1(k)=\frac{k-i\kappa}{k}r_1(k), \quad {r}^{r}_2(k)=\frac{k}{k-i\kappa}r_2(k).
    \end{equation}
    Moreover
    \item[(iii)] Normalization at $k=\infty$. $\breve{M}^{r}(x,t,k)\rightarrow I$ as $k\rightarrow\infty$.
    \item[(iv)] Matrix-valued factor $P(x,t)$ are determined in terms of $\breve{M}^{r}(x,t,k)$:
    \begin{subequations}\label{P12andP21}
        \begin{align}
            &P_{12}(x,t)=\frac{g_1(x,t)h_1(x,t)}{g_1(x,t)h_2(x,t)-g_2(x,t)h_1(x,t)},\\
            &P_{21}(x,t)=-\frac{g_2(x,t)h_2(x,t)}{g_1(x,t)h_2(x,t)-g_2(x,t)h_1(x,t)},\\
            &P_{11}(x,t)=-\frac{P_{12}(x,t)g_2(x,t)}{g_1(x,t)}, \quad P_{22}(x,t)=-\frac{P_{21}(x,t)g_1(x,t)}{g_2(x,t)}
        \end{align}
    \end{subequations}
   where $g(x,t)=(g_1(x,t), g_2(x,t))^{\rm T}$ and $h(x,t)=(h_1(x,t), h_2(x,t))^{\rm T}$ are given by
   \begin{subequations}\label{gandhdefine}
    \begin{align}
        &g(x,t)=i\kappa\breve{M}^{r(1)}(x,t,i\kappa)-c_1(x,t)\breve{M}^{r(2)}(x,t,i\kappa), \label{gdefine}\\
        &h(x,t)=i\kappa\breve{M}^{r(2)}(x,t,0)+c_0(\xi)\breve{M}^{r(1)}(x,t,0), \label{hdefine}
    \end{align}
    where $\breve{M}^{r(j)}$, $j=1,2$ represent the $j$-th column of $\breve{M}^{r}$.
   \end{subequations}
\end{itemize}
\end{RHP}
\begin{proof}
    The proof of this RH problem is similar to \cite[Proposition. 6]{RDJDE2021}.
\end{proof}

Moreover, by \eqref{reflectioncoeCondition1}, we can see that
\begin{equation*}
    {r}^{r}_1(k)=\overline{{r}^{r}_{1}(-\bar{k})}, \quad {r}^{r}_2(k)=\overline{{r}^{r}_{2}(-\bar{k})}.
\end{equation*}
Combine the symmetry of $\delta(\xi,k)$, we obtain that
\begin{equation}\label{breveJrsym}
\breve{J}^r(x,t,k)=\overline{\breve{J}^r(x,t,-\bar{k})}
\end{equation}
With these symmetries, we conclude that the
regular RH problem $\breve{M}^{r}(x,t,k)$ admits the symmetry
\begin{equation}\label{breveMsym}
    \breve{M}^{r}(x,t,k)=\overline{\breve{M}^{r}(x,t,-\bar{k})}.
\end{equation}

Next, we introduce next proposition, which establishes the relation between the solution $u(x,t)$ of \eqref{Cauchy Problem}, \eqref{bdrycondition} and the regular RH problem $\breve{M}^{r}$.
From the relation, we can obtain a rough result of large-time asymptotics for $u(x,t)$.
\begin{proposition}
    The solution $u(x,t)$ of the Cauchy problem \eqref{Cauchy Problem}, \eqref{bdrycondition} can be expressed in terms of
    \begin{subequations}\label{uintermofregularrhp}
    \begin{align}
        &u(x,t)=-2\kappa P_{12}(x,t)+2i\lim_{k\rightarrow\infty}k\breve{M}^{r}_{12}(x,t,k), \quad x>0, \ t<0, \label{uintermofregularrhp1}\\
        &u(x,t)=-2\kappa P_{21}(-x,-t)+2i\lim_{k\rightarrow\infty}k\breve{M}^{r}_{21}(-x,-t,k), \quad x<0, \ t>0, \label{uintermofregularrhp2}
    \end{align}
\end{subequations}
Moreover, as $t\rightarrow-\infty$, we have a rough estimation
\begin{equation}\label{roughresultforuxi<01}
    u(x,t)=A\delta^2(0,\xi)+o(1), \quad x>0, \ t<0
\end{equation}
along any ray $\xi=\frac{x}{12t}=const<0$. And as $t\rightarrow+\infty$,
\begin{equation}\label{roughresultforuxi<02}
    u(x,t)=o(1), \quad x<0, \ t>0
\end{equation}
along any ray $\xi=\frac{x}{12t}=const<0$.
\end{proposition}
\begin{proof}
    Firstly, we use \eqref{transforregularrhp} to obtain the large-$k$ behavior as
    \begin{equation}
        \breve{M}(x,t,k)=\begin{pmatrix}
            1 & 0 \\
            0 & \frac{k-i\kappa}{k}
        \end{pmatrix}+\frac{\breve{M}^r_{1}(x,t)}{k}+\frac{i\kappa}{k-i\kappa}P(x,t)+O(k^{-2}), \quad {\rm as} \ k\rightarrow\infty,
    \end{equation}
    where $\breve{M}^r(x,t,k)=I+\frac{\breve{M}^r_{1}(x,t)}{k}+O(\frac{1}{k^2})$, $k\rightarrow\infty$. Then we combine \eqref{uintermsofM} and
    a series of transformation for basic RH problem $M(x,t,k)$ to obtain \eqref{uintermofregularrhp}.

    Because the regular RH problem $\breve{M}^r$ has the same form as in the case of zero background, refer \cite{NMKdVzeroback}, $\breve{M}^{r}\approx I$
    as $t\rightarrow-\infty$. Owe to this, we can give the leading term of $u(x,t)$ as $t\rightarrow-\infty$ before formulating the precise more precise asymptotics.
    Indeed, $\breve{M}^{r}\approx I$ as $t\rightarrow-\infty$ reveals
    $(g_1(x,t), g_2(x,t))^{\rm T}\approx(i\kappa, -c_1(x,t))^{\rm T}\underset{x>0, \ t\rightarrow-\infty}{\approx}(i\kappa, 0)^{\rm T}$ as well as
    $(h_1(x,t), h_2(x,t))^{\rm T}\approx(c_0(\xi), i\kappa)^{\rm T}$. Then for $x>0$, we have
    \begin{equation}
        u(x,t)\approx-2\kappa P_{12}(x,t)\approx\frac{2i\kappa^2c_0(\xi)}{\kappa^2-c_1(x,t)c_0(\xi)}\underset{t\rightarrow-\infty}{\approx}2ic_0(\xi)\approx A\delta^2(0,\xi),
    \end{equation}
    which is the leading term of $u(x,t)$ as $t\rightarrow-\infty$, $x>0$.

    Similarly, for $x<0$, $t\rightarrow+\infty$ we have $(g_1(-x,-t), g_2(-x,-t))^{\rm T}\approx(i\kappa, -c_1(-x,-t))^{\rm T}\underset{t\rightarrow+\infty}{\approx}(i\kappa, 0)^{\rm T}$
    as well as $(h_1(-x,-t), h_2(-x,-t))^{\rm T}\approx(c_0(\xi), i\kappa)^{\rm T}$
    \begin{equation}
        u(x,t)\approx-2\kappa P_{21}(-x,-t)\approx 2\kappa\frac{-c_1(-x,-t)(-i\kappa)}{-\kappa^2+c_0(\xi)c_1(-x,-t)}\approx 0.
    \end{equation}
    which is the leading term of $u(x,t)$ as $t\rightarrow+\infty$, $x<0$.
\end{proof}

\subsection{Local parametrix near saddle points}\label{subsectionlocal}
In this subsection, our goal is to use the parabolic cylinder functions to present a good approximation
of $\breve{M}^{r}(x,t,k)$ locally around the saddle points $-k_0$ and $k_0$. We use $\breve{M}^r_{-k_0}$ and $\breve{M}^r_{k_0}$ to express
the local parametrix of $-k_0$ and $k_0$ respectively. To be brief, we only exhibit the details for
$\breve{M}^r_{-k_0}$, then use symmetry relation to obtain the information of the local neighborhood of $\breve{M}^r_{k_0}$.
The methods we use in this subsection mainly refer to our previous work \cite{NMKdVzeroback} and Lenells' work \cite{lenellsmkdv}.

Let $U_{\epsilon}(-k_0)$ and $U_{\epsilon}(k_0)$ denote the open disc of radius $\epsilon$ around $-k_0$ and $k_0$ respectively. Define $\epsilon$
\begin{equation}
    \epsilon:=\min\left\{\eta:=\frac{k_0}{2},\quad  \frac{1}{2}\vert i\kappa+k_0 \vert\right\},
\end{equation}
which make the $0$ and $i\kappa$ do not belong to the neighborhood of  $U_{\epsilon}(-k_0)$ and $U_{\epsilon}(k_0)$.

And define the contours as follows:
\begin{subequations}
\begin{align}
    &\Gamma_{k_0,\epsilon}:=\Gamma\cap U_{\epsilon}(k_0)=\Gamma_{j,\epsilon}\cup\Gamma_{j,\epsilon}^{*}, & j=1, 2, \\
    &\Gamma_{-k_0,\epsilon}:=\Gamma\cap U_{\epsilon}(-k_0)=\Gamma_{j,\epsilon}\cup\Gamma_{j,\epsilon}^{*}, & j=3, 4, \\
    &\Gamma_\epsilon:=\Gamma_{k_0,\epsilon}\cup\Gamma_{-k_0,\epsilon}.
\end{align}
\end{subequations}
See Fig \ref{FigLocalModel}.

\begin{figure}[htbp]
\begin{center}
\tikzset{every picture/.style={line width=0.75pt}} 
\begin{tikzpicture}[x=0.75pt,y=0.75pt,yscale=-1,xscale=1]
\draw    (480.01,143.16) -- (558.29,74.02) ;
\draw [shift={(523.65,104.62)}, rotate = 138.55] [color={rgb, 255:red, 0; green, 0; blue, 0 }  ][line width=0.75]    (10.93,-3.29) .. controls (6.95,-1.4) and (3.31,-0.3) .. (0,0) .. controls (3.31,0.3) and (6.95,1.4) .. (10.93,3.29)   ;
\draw    (480.01,143.16) -- (403.29,72.02) ;
\draw [shift={(446.78,112.35)}, rotate = 222.84] [color={rgb, 255:red, 0; green, 0; blue, 0 }  ][line width=0.75]    (10.93,-3.29) .. controls (6.95,-1.4) and (3.31,-0.3) .. (0,0) .. controls (3.31,0.3) and (6.95,1.4) .. (10.93,3.29)   ;
\draw    (204.01,148.16) -- (278.29,71.02) ;
\draw [shift={(245.31,105.27)}, rotate = 133.92] [color={rgb, 255:red, 0; green, 0; blue, 0 }  ][line width=0.75]    (10.93,-3.29) .. controls (6.95,-1.4) and (3.31,-0.3) .. (0,0) .. controls (3.31,0.3) and (6.95,1.4) .. (10.93,3.29)   ;
\draw    (480.01,143.16) -- (398.29,219.02) ;
\draw [shift={(434.75,185.17)}, rotate = 317.13] [color={rgb, 255:red, 0; green, 0; blue, 0 }  ][line width=0.75]    (10.93,-3.29) .. controls (6.95,-1.4) and (3.31,-0.3) .. (0,0) .. controls (3.31,0.3) and (6.95,1.4) .. (10.93,3.29)   ;
\draw    (559.29,217.02) -- (479.01,142.16) ;
\draw [shift={(514.76,175.5)}, rotate = 43] [color={rgb, 255:red, 0; green, 0; blue, 0 }  ][line width=0.75]    (10.93,-3.29) .. controls (6.95,-1.4) and (3.31,-0.3) .. (0,0) .. controls (3.31,0.3) and (6.95,1.4) .. (10.93,3.29)   ;
\draw    (277.29,223.02) -- (204.01,148.16) ;
\draw [shift={(236.45,181.3)}, rotate = 45.61] [color={rgb, 255:red, 0; green, 0; blue, 0 }  ][line width=0.75]    (10.93,-3.29) .. controls (6.95,-1.4) and (3.31,-0.3) .. (0,0) .. controls (3.31,0.3) and (6.95,1.4) .. (10.93,3.29)   ;
\draw    (126.29,72.02) -- (204.01,148.16) ;
\draw [shift={(169.43,114.29)}, rotate = 224.41] [color={rgb, 255:red, 0; green, 0; blue, 0 }  ][line width=0.75]    (10.93,-3.29) .. controls (6.95,-1.4) and (3.31,-0.3) .. (0,0) .. controls (3.31,0.3) and (6.95,1.4) .. (10.93,3.29)   ;
\draw    (204.01,148.16) -- (130.29,226.02) ;
\draw [shift={(163.02,191.45)}, rotate = 313.43] [color={rgb, 255:red, 0; green, 0; blue, 0 }  ][line width=0.75]    (10.93,-3.29) .. controls (6.95,-1.4) and (3.31,-0.3) .. (0,0) .. controls (3.31,0.3) and (6.95,1.4) .. (10.93,3.29)   ;
\draw  [dash pattern={on 0.84pt off 2.51pt}] (97.22,148.16) .. controls (97.22,89.18) and (145.03,41.37) .. (204.01,41.37) .. controls (262.98,41.37) and (310.79,89.18) .. (310.79,148.16) .. controls (310.79,207.13) and (262.98,254.94) .. (204.01,254.94) .. controls (145.03,254.94) and (97.22,207.13) .. (97.22,148.16) -- cycle ;
\draw  [dash pattern={on 0.84pt off 2.51pt}] (373.14,148.16) .. controls (373.14,89.14) and (420.98,41.3) .. (480,41.3) .. controls (539.02,41.3) and (586.87,89.14) .. (586.87,148.16) .. controls (586.87,207.18) and (539.02,255.02) .. (480,255.02) .. controls (420.98,255.02) and (373.14,207.18) .. (373.14,148.16) -- cycle ;
\draw (175.99,141.9) node [anchor=north west][inner sep=0.75pt]  [font=\scriptsize,rotate=-0.03]  {$-k_{0}$};
\draw (488.99,139.05) node [anchor=north west][inner sep=0.75pt]  [font=\scriptsize,rotate=-0.03]  {$k_{0}$};
\draw (44,52.4) node [anchor=north west][inner sep=0.75pt]  [font=\scriptsize]  {$U_{\epsilon }( -k_{0})$};
\draw (569,46.4) node [anchor=north west][inner sep=0.75pt]  [font=\scriptsize]  {$U_{\epsilon }( k_{0})$};
\draw (266,91.4) node [anchor=north west][inner sep=0.75pt]  [font=\scriptsize]  {$\Gamma _{3,\ \epsilon }$};
\draw (263,180.4) node [anchor=north west][inner sep=0.75pt]  [font=\scriptsize]  {$\Gamma _{3,\ \epsilon }^{*}$};
\draw (126,96.4) node [anchor=north west][inner sep=0.75pt]  [font=\scriptsize]  {$\Gamma _{4,\ \epsilon }$};
\draw (123,175.4) node [anchor=north west][inner sep=0.75pt]  [font=\scriptsize]  {$\Gamma _{4,\ \epsilon }^{*}$};
\draw (539,96.4) node [anchor=north west][inner sep=0.75pt]  [font=\scriptsize]  {$\Gamma _{1,\ \epsilon }$};
\draw (544,179.4) node [anchor=north west][inner sep=0.75pt]  [font=\scriptsize]  {$\Gamma _{1,\ \epsilon }^{*}$};
\draw (406,94.4) node [anchor=north west][inner sep=0.75pt]  [font=\scriptsize]  {$\Gamma _{2,\ \epsilon }$};
\draw (421,190.4) node [anchor=north west][inner sep=0.75pt]  [font=\scriptsize]  {$\Gamma _{2,\ \epsilon }^{*}$};
\end{tikzpicture}
\caption{\small The jump contours of the local parametrix $\breve{M}^r_{-k_0}$ and $\breve{M}^r_{k_0}$}\label{FigLocalModel}
\end{center}
\end{figure}
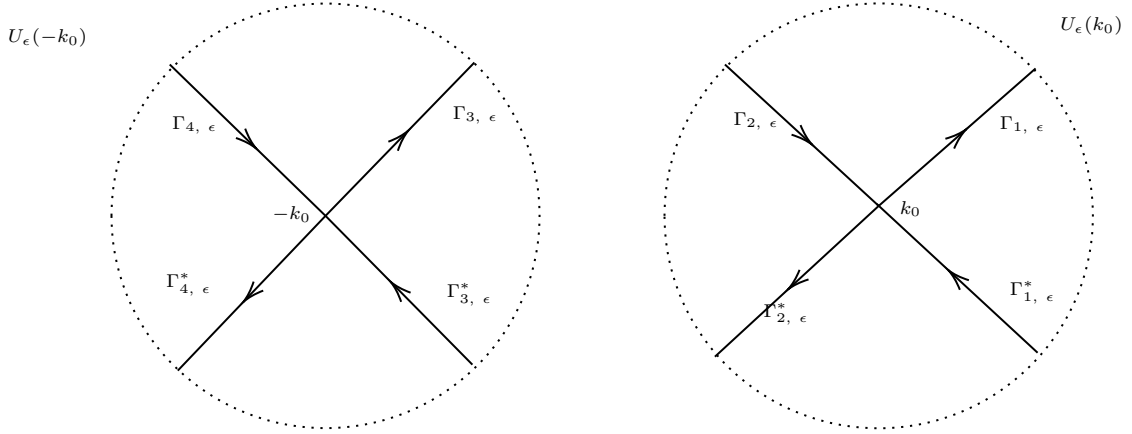

Here we consider $r_j(-k_0)\neq 0$, $j=1,2$. For the cases when one of the $r_j(-k_0)$ or the both equal to $0$ and thus $\nu(-k_0)=0$,
we suggest the interested readers to refer \cite[Sec 1.6, Ch 2]{PainleveTrans}. Additionally, we notice that $\breve{M}^{r}(x,t,k)$ is similar to
in which in the case of zero boundary condition, so we follow the idea of \cite{NMKdVzeroback} to complete the rest asymptotic analysis.

At first, we define some quantities which are convenient for our expressions in this subsection.
\begin{equation}
    \eta:=\frac{k_0}{2}, \quad \rho=\eta\sqrt{48k_0}, \quad \tau:=-t\rho^2=-12tk_0^3>0, \quad \nu:=\nu(-k_0).
\end{equation}
And
\begin{equation}
    \varphi(\xi;\zeta):=2i\theta\left(\xi,-k_0+\frac{\eta}{\rho}\right)=16ik_0^3-\frac{i}{2}\zeta^2+\frac{i\zeta^3}{12\rho}.
\end{equation}
\begin{remark}\rm
    We can see that: when $t\rightarrow-\infty$, $\tau\rightarrow+\infty$. Whereas, in \cite{NMKdVzeroback} where $\tau:=t\rho^2=12tk_0^3$ for $t>0$, $\tau\rightarrow+\infty$ as $t\rightarrow+\infty$.
\end{remark}

For $k$ near $-k_0$, we notice that
\begin{equation}
    \theta(k,\xi)=8k_0^3-12k_0(k+k_0)^2+4(k+k_0)^3.
\end{equation}
Thus, for $k\in U_{\epsilon}(-k_0)$, we define the rescaled variable $\zeta$ by
\begin{equation}\label{scaling}
    k=\frac{\zeta}{\sqrt{-48tk_0}}-k_0=\frac{\eta}{\sqrt{\tau}}\zeta-k_0
\end{equation}
(for $k\in U_{\epsilon}(k_0)$, $k=\frac{\zeta}{\sqrt{-48tk_0}}+k_0$). And the scaling operator $N_{-k_0}$ admits the mapping
\begin{align}
    f(\zeta)\longmapsto (N_{-k_0}f)(\zeta)=f\left(\frac{\zeta}{\sqrt{-48tk_0}}-k_0\right).
\end{align}

Now we match the ${q}^r_{j}(-k_0)$, $j=1,2$ appeared in \eqref{jumppc}. For $k\in\Gamma\cap U_{\epsilon}(-k_0)$, for instance $k\in\Gamma_{3,\epsilon}$,
Use formulae \eqref{deltahatchi} and \eqref{scaling} , we have
\begin{equation}
    \delta^{-2}{r}^{r}_1(k)e^{2it\theta}|_{k=-k_0}=e^{16itk_0^3-i\zeta^3(144\tau)^{-\frac{1}{2}}}\tau^{i\nu}e^{-2\chi(\xi,-k_0)}{r}^{r}_1(-k_0)4^{2i\nu}e^{\frac{i}{2}\zeta^2}\zeta^{-2i\nu},
\end{equation}
then compare to the first jump of \eqref{jumppc}, we obtain
\begin{align}
    q_1^r(-k_0)=e^{-2\chi(\xi,-k_0)}{r}^{r}_1(-k_0)e^{2i\nu\log 4}.
\end{align}
Similarly, we derive
\begin{align}
    q^r_2(-k_0)=e^{2\chi(\xi,-k_0)}{r}^{r}_2(-k_0)e^{-2i\nu\log 4}.
\end{align}

The local parametrix $\breve{M}^r_{-k_0}$ determined by
\begin{equation}\label{LCPara-k_0}
    \breve{M}^r_{-k_0}(x,t,k)=\Xi(\xi,t)m^{pc}_{-k_0}\left(\xi, \zeta(k)=\sqrt{-48tk_0}(k+k_0)\right)\Xi^{-1}(x,t)
\end{equation}
where
\begin{equation}\label{Xibiaodashi}
    \Xi(\xi,t)=e^{-\frac{t}{2}\varphi(\xi,0)\sigma_3}\cdot\tau^{-\frac{i\nu\sigma_3}{2}}.
\end{equation}

The RH problem $m^{pc}_{-k_0}(\xi,\zeta)$ appeared in \eqref{LCPara-k_0} can be solved explicitly, in terms of parabolic cylinder functions (see Appendix \ref{Appendixpcmodel}) and admits
the asymptotic behavior as $\zeta\rightarrow\infty$
\begin{equation}\label{mpc-k0biaodashi}
    m^{pc}_{-k_0}\left(\xi, \zeta(k)\right)=I+\frac{i}{\zeta}\begin{pmatrix}
         0  & -\beta^r(\xi) \\ \gamma^r(\xi) & 0
    \end{pmatrix}+ O(\frac{1}{\zeta^2}),\quad \zeta\rightarrow\infty.
\end{equation}
where
\begin{align}\label{betargammar}
    \beta^r(\xi)=\frac{\sqrt{2\pi}e^{\frac{i\pi}{4}}e^{-\frac{\pi \nu(-k_0)}{2}}}{{q}^r_{1}(-k_0)\Gamma(-i\nu(-k_0))}, \quad
    \gamma^r(\xi)=\frac{\sqrt{2\pi}e^{-\frac{i\pi}{4}}e^{-\frac{\pi \nu(-k_0)}{2}}}{{q}^r_{2}(-k_0)\Gamma(i\nu(-k_0))}
\end{align}

By the symmetry condition \eqref{breveMsym}, we extend the local parametrix to a neighborhood of $k_0$ by
\begin{equation}
    \breve{M}^r_{k_0}(x,t,k)=\overline{\breve{M}^r_{-k_0}(x,t,-\bar{k})}, \quad \vert k-k_0 \vert<\epsilon.
\end{equation}

\subsection{Error analysis}\label{subsectionAANO1erroranalysis}
Having constructed the local parametrix $\breve{M}^{r}_{\pm k_0}$, then we define so-called error matrix as
\begin{equation}\label{defineE}
    E(x,t,k)=\left\{
        \begin{aligned}
        & \breve{M}^{r}(x,t,k)\left(\breve{M}^{r}_{-k_0}\right)^{-1}(x,t,k), &k\in U_{\epsilon}(-k_0)=\left\{k:|k+k_0|<\epsilon\right\}, \\
        & \breve{M}^{r}(x,t,k)\left(\breve{M}^{r}_{k_0}\right)^{-1}(x,t,k), &k\in U_{\epsilon}(k_0)=\left\{k:|k-k_0|<\epsilon\right\}, \\
        & \breve{M}^{r}(x,t,k), &{\rm elsewhere.}
        \end{aligned}
        \right.
\end{equation}
The jump contours of $E(x,t,k)$ is depicted as Fig \ref{FigJumpE}.

\begin{figure}[htbp]
\begin{center}
\tikzset{every picture/.style={line width=0.75pt}} 
\begin{tikzpicture}[x=0.75pt,y=0.75pt,yscale=-1,xscale=1]
\draw    (435.36,140.53) -- (553.29,92.86) ;
\draw [shift={(499.89,114.45)}, rotate = 157.99] [color={rgb, 255:red, 0; green, 0; blue, 0 }  ][line width=0.75]    (10.93,-3.29) .. controls (6.95,-1.4) and (3.31,-0.3) .. (0,0) .. controls (3.31,0.3) and (6.95,1.4) .. (10.93,3.29)   ;
\draw    (435.36,140.53) -- (324.29,98.86) ;
\draw [shift={(386.38,122.16)}, rotate = 200.56] [color={rgb, 255:red, 0; green, 0; blue, 0 }  ][line width=0.75]    (10.93,-3.29) .. controls (6.95,-1.4) and (3.31,-0.3) .. (0,0) .. controls (3.31,0.3) and (6.95,1.4) .. (10.93,3.29)   ;
\draw    (324.29,98.86) -- (220.99,140.19) ;
\draw [shift={(267.07,121.75)}, rotate = 338.2] [color={rgb, 255:red, 0; green, 0; blue, 0 }  ][line width=0.75]    (10.93,-3.29) .. controls (6.95,-1.4) and (3.31,-0.3) .. (0,0) .. controls (3.31,0.3) and (6.95,1.4) .. (10.93,3.29)   ;
\draw    (435.36,140.53) -- (325.29,181.86) ;
\draw [shift={(374.71,163.31)}, rotate = 339.42] [color={rgb, 255:red, 0; green, 0; blue, 0 }  ][line width=0.75]    (10.93,-3.29) .. controls (6.95,-1.4) and (3.31,-0.3) .. (0,0) .. controls (3.31,0.3) and (6.95,1.4) .. (10.93,3.29)   ;
\draw    (549.29,190.5) -- (435.36,140.53) ;
\draw [shift={(486.83,163.11)}, rotate = 23.68] [color={rgb, 255:red, 0; green, 0; blue, 0 }  ][line width=0.75]    (10.93,-3.29) .. controls (6.95,-1.4) and (3.31,-0.3) .. (0,0) .. controls (3.31,0.3) and (6.95,1.4) .. (10.93,3.29)   ;
\draw    (325.29,181.86) -- (220.99,140.19) ;
\draw [shift={(267.57,158.8)}, rotate = 21.78] [color={rgb, 255:red, 0; green, 0; blue, 0 }  ][line width=0.75]    (10.93,-3.29) .. controls (6.95,-1.4) and (3.31,-0.3) .. (0,0) .. controls (3.31,0.3) and (6.95,1.4) .. (10.93,3.29)   ;
\draw    (113.29,99.86) -- (220.99,140.19) ;
\draw [shift={(172.76,122.13)}, rotate = 200.53] [color={rgb, 255:red, 0; green, 0; blue, 0 }  ][line width=0.75]    (10.93,-3.29) .. controls (6.95,-1.4) and (3.31,-0.3) .. (0,0) .. controls (3.31,0.3) and (6.95,1.4) .. (10.93,3.29)   ;
\draw   (410.74,140.53) .. controls (410.74,126.93) and (421.76,115.91) .. (435.36,115.91) .. controls (448.96,115.91) and (459.99,126.93) .. (459.99,140.53) .. controls (459.99,154.13) and (448.96,165.16) .. (435.36,165.16) .. controls (421.76,165.16) and (410.74,154.13) .. (410.74,140.53) -- cycle ;
\draw   (196.36,140.19) .. controls (196.36,126.59) and (207.39,115.56) .. (220.99,115.56) .. controls (234.59,115.56) and (245.62,126.59) .. (245.62,140.19) .. controls (245.62,153.79) and (234.59,164.81) .. (220.99,164.81) .. controls (207.39,164.81) and (196.36,153.79) .. (196.36,140.19) -- cycle ;
\draw   (224.9,118.96) -- (213.12,116.17) -- (222.88,109.01) ;
\draw   (438.23,119.35) -- (426.24,117.67) -- (435.3,109.63) ;
\draw  [dash pattern={on 0.84pt off 2.51pt}]  (88.8,141.4) -- (584.8,140.4) ;
\draw    (121,181.14) -- (220.99,140.19) ;
\draw [shift={(176.55,158.39)}, rotate = 157.73] [color={rgb, 255:red, 0; green, 0; blue, 0 }  ][line width=0.75]    (10.93,-3.29) .. controls (6.95,-1.4) and (3.31,-0.3) .. (0,0) .. controls (3.31,0.3) and (6.95,1.4) .. (10.93,3.29)   ;
\draw (211.99,145.9) node [anchor=north west][inner sep=0.75pt]  [font=\scriptsize,rotate=-0.03]  {$-k_{0}$};
\draw (430.99,144.05) node [anchor=north west][inner sep=0.75pt]  [font=\scriptsize,rotate=-0.03]  {$k_{0}$};
\draw (525,80.4) node [anchor=north west][inner sep=0.75pt]  [font=\scriptsize]  {$\Gamma _{1}$};
\draw (526,186.4) node [anchor=north west][inner sep=0.75pt]  [font=\scriptsize]  {$\Gamma _{1}^{*}$};
\draw (361,90.4) node [anchor=north west][inner sep=0.75pt]  [font=\scriptsize]  {$\Gamma _{2}$};
\draw (372.33,170.6) node [anchor=north west][inner sep=0.75pt]  [font=\scriptsize]  {$\Gamma _{2}^{*}$};
\draw (278,92.4) node [anchor=north west][inner sep=0.75pt]  [font=\scriptsize]  {$\Gamma _{3}$};
\draw (275,174.4) node [anchor=north west][inner sep=0.75pt]  [font=\scriptsize]  {$\Gamma _{3}^{*}$};
\draw (140,91.4) node [anchor=north west][inner sep=0.75pt]  [font=\scriptsize]  {$\Gamma _{4}$};
\draw (144.64,173.66) node [anchor=north west][inner sep=0.75pt]  [font=\scriptsize]  {$\Gamma _{4}^{*}$};
\end{tikzpicture}
\caption{\small The jump contours $\Gamma_E:=\Gamma\cup\partial U_{\epsilon}(-k_0)\cup\partial U_{\epsilon}(k_0)$ of $E(x,t,k)$}\label{FigJumpE}
\end{center}
\end{figure}
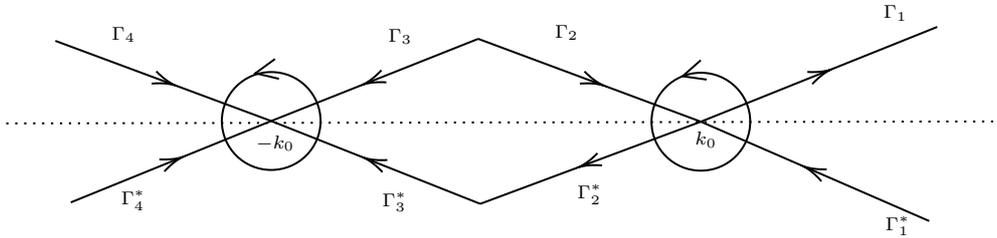

And $E(x,t,k)$ satisfies the RH problem as follows
\begin{RHP}\label{RHPE}
    Find a $2\times2$ matrix-valued function $E(x,t,k)$ such that
    \begin{itemize}
        \item[(i)] $E(x,t,k)$ is analytic for $k\in\mathbb{C}\backslash\Gamma_E$, where $\Gamma_E:=\Gamma\cup\partial U_{\epsilon}(-k_0)\cup\partial U_{\epsilon}(k_0)$.
        \item[(ii)] $E(x,t,k)$ takes continuous boundary conditions $E_{+}(x,t,k)=E_{-}(x,t,k)J_E(x,t,k)$ with
        \begin{equation}\label{jumpforE}
            J_E(x,t,k)=\left\{
                \begin{aligned}
                &\breve{M}^{r}_{-k_0}(x,t,k)\breve{J}^r(x,t,k)\left(\breve{M}^{r}_{-k_0}\right)^{-1}(x,t,k), &k\in\Gamma\cap U_{\epsilon}(-k_0)=\Gamma_{-k_0,\epsilon},, \\
                &\breve{M}^{r}_{k_0}(x,t,k)\breve{J}^r(x,t,k)\left(\breve{M}^{r}_{k_0}\right)^{-1}(x,t,k), &k\in\Gamma\cap U_{\epsilon}(k_0)=\Gamma_{k_0,\epsilon}, \\
                &\left(\breve{M}^{r}_{-k_0}\right)^{-1}(x,t,k),  &k\in\partial U_{\epsilon}(-k_0),\\
                &\left(\breve{M}^{r}_{k_0}\right)^{-1}(x,t,k),  &k\in\partial U_{\epsilon}(k_0), \\
                &\breve{J}^{r}(x,t,k), &\Gamma\backslash\Gamma_{\epsilon}.
                \end{aligned}
                \right.
        \end{equation}
        \item[(iii)] $E(x,t,k)=I+O(k^{-1})$, as $k\rightarrow\infty$.
    \end{itemize}
\end{RHP}

Define $w(x,t,k):=J_E(x,t,k)-I$. By \eqref{breveJrsym} and \eqref{breveMsym}, we know that $w(x,t,k)$ admits the symmetry
\begin{equation}\label{wsymmetry}
    w(x,t,k)=\overline{w(x,t,-\bar{k})}.
\end{equation}
Then next lemma shows that some estimates for $k\in\Gamma\backslash\Gamma_{\epsilon}$, $k\in\Gamma_{\epsilon}$ and $k\in\partial U_{\epsilon}(\pm k_0)$ respectively.
\begin{lemma}\label{lemwest}
    As $\tau\rightarrow\infty$, $w(x,t,k)$ admits some estimates as follows
   \begin{itemize}
    \item[(i)] For $k\in\Gamma\backslash\Gamma_{\epsilon}$, as $\tau\rightarrow\infty$
    \begin{subequations}\label{westlemest1}
        \begin{align}
            &\Vert w(x,t,\cdot) \Vert_{L^p(\Gamma\backslash\Gamma_{\epsilon})}=O(\epsilon^{\frac{1}{p}}\tau^{-1}), \quad p=1,2, \label{westlemest1a}\\
            &\Vert w(x,t,\cdot) \Vert_{L^\infty(\Gamma\backslash\Gamma_{\epsilon})}=O(\tau^{-1}). \label{westlemest1b}
        \end{align}
    \end{subequations}
    \item[(ii)] For $k\in\Gamma_{\epsilon}$, as $\tau\rightarrow\infty$,
        \begin{align}\label{westlemest2}
             w(x,t,\cdot)=O\left(\tau^{-\frac{\alpha}{2}+\vert \im\nu \vert}e^{(-\tau/24\eta^2)\vert k\pm k_0 \vert^2}\right).
        \end{align}
        where $\alpha\in(\lambda,1)$ with $\lambda:=\max(1/2, \ \underset{\xi<0, |\xi|=O(1)}{\rm sup}\ 2|\im\nu(k(\xi))|)$.
    \item[(iii)] For $k\in\partial U_{\epsilon}(-k_0)$, as $\tau\rightarrow\infty$
    \begin{align}\label{westlemest3}
        w(x,t,k)=(M^{r}_{-k_0})^{-1}(x,t,k)-I=\frac{\mathfrak{B}_0(\xi,t)}{\sqrt{\tau}(k+k_0)}+\hat{R}_{1}(\xi,t),
    \end{align}
    where
    \begin{subequations}
    \begin{align}\label{hatBexpression}
        &\mathfrak{B}_0(\xi,t)=-i\eta\begin{pmatrix}
            0 & \beta^r(\xi)e^{-t\varphi(\xi,0)}\tau^{-i\nu} \\
            \gamma^r(\xi)e^{t\varphi(\xi,0)}\tau^{i\nu} & 0
            \end{pmatrix}\\
        &\hat{R}_{1}(\xi,t)=\begin{pmatrix}
                O(\tau^{-1-\im \nu}) & O(\tau^{-1+\im \nu})
            \end{pmatrix}
    \end{align}
    \end{subequations}
   \end{itemize}
\end{lemma}
\begin{proof}
    For \eqref{westlemest1} and \eqref{westlemest2}, we can  follow the proof of \cite[Lemma 2]{NMKdVzeroback} and \cite[Lemma 4]{NMKdVzeroback} respectively. Now we turn to
    the proof of \eqref{westlemest3}.

    Recall the $\Xi(\xi,t)$ defined by \eqref{Xibiaodashi}, we obtain that
    \begin{equation}\label{Xiasy}
        \Xi(\xi,t)=O\begin{pmatrix} \tau^{-\frac{\im\nu}{2}} & \tau^{\frac{\im\nu}{2}}\end{pmatrix}, \quad \tau\rightarrow\infty.
    \end{equation}
For $k\in\partial U_{\epsilon}(-k_0)$, as $\tau\rightarrow\infty$
    \begin{align}
        w&=\left(\breve{M}^{r}_{-k_0}\right)^{-1}(x,t,k)-I\nonumber\\
        &\overset{\eqref{LCPara-k_0}}{=}\Xi(\xi,t)\left(\left(m^{pc}_{-k_0}\right)^{-1}\left(\xi,\frac{\sqrt{\tau}}{\eta}\left(k+k_0\right)\right)-I\right)\Xi^{-1}(\xi,t)\nonumber\\
        &\overset{\eqref{mpc-k0biaodashi}}{=}\Xi(\xi,t)\left(I-\frac{i\eta}{\sqrt{\tau}(k+k_0)}\begin{pmatrix}
            0 & \beta^r(\xi) \\
            \gamma^r(\xi) & 0
            \end{pmatrix}+O(\tau^{-1})-I\right)\Xi^{-1}(\xi,t).
    \end{align}
    Then by \eqref{Xibiaodashi} and \eqref{Xiasy}, we arrive at \eqref{westlemest3}.
\end{proof}

Then we use Lemma \ref{lemwest} and closely follow the \cite[Lemma 2.6]{lenellsmkdv} or \cite[Eq.(99)]{NMKdVzeroback}, we have the proposition as follows
\begin{proposition}As $\tau\rightarrow\infty$,
    \begin{subequations}\label{westprop1}
        \begin{align}
            &\Vert w(x,t,\cdot) \Vert_{L^2(\Gamma_E)}=O\left(\epsilon^{\frac{1}{2}}\tau^{-\frac{\alpha}{2}+\vert \im \nu\vert}\right), \label{westprop1a}\\
            &\Vert w(x,t,\cdot) \Vert_{L^\infty(\Gamma_E)}=O\left(\tau^{-\frac{\alpha}{2}+\vert \im \nu\vert}\right), \label{westprop1b}\\
            &\Vert w(x,t,\cdot) \Vert_{L^p(\Gamma_\epsilon)}=O\left(\epsilon^{\frac{1}{p}}\tau^{-\frac{1}{2p}-\frac{\alpha}{2}+\vert \im \nu\vert}\right).\label{westprop1c}
        \end{align}
    \end{subequations}
where $p\in[1,\infty)$. Moreover, recall that $w^{(j)}$ denotes the $j$-th of matrix $w$, we have
\begin{subequations}\label{westprop2}
    \begin{align}
        &\Vert w^{(j)}(x,t,\cdot) \Vert_{L^2(\Gamma_E)}=O\left(\epsilon^{\frac{1}{2}}\tau^{-\frac{\alpha}{2}+(-1)^j\im\nu}\right), \label{westprop2a}\\
        &\Vert w^{(j)}(x,t,\cdot) \Vert_{L^\infty(\Gamma_E)}=O\left(\tau^{-\frac{\alpha}{2}+(-1)^j\im\nu}\right), \label{westprop2b}\\
        &\Vert w^{(j)}(x,t,\cdot) \Vert_{L^p(\Gamma_\epsilon)}=O\left(\epsilon^{\frac{1}{p}}\tau^{-\frac{1}{2p}-\frac{\alpha}{2}+(-1)^j\im\nu}\right).\label{westprop2c}
    \end{align}
\end{subequations}
for $j=1,2$.
\end{proposition}

Define the Cauchy-type operator $C_{w}:L^2(\Gamma_E)+L^{\infty}(\Gamma_E)\rightarrow L^2(\Gamma_E)$ by $C_{w}f=C_{-}(fw)$, where $(C_{-}f)(k)$, $k\in\Gamma_E$ are the negative (according
to the orientation of $\Gamma_E$) non-tangential boundary values of
\begin{equation}
    \left(Cf\right)(k'):=\frac{1}{2\pi i}\int_{\Gamma_E}\frac{f(s)}{s-k'}ds, \quad k'\in\mathbb{C}\backslash\Gamma_E.
\end{equation}
Owe to $C_{-}$ is an operator of $L^2(\Gamma_E)\rightarrow L^2(\Gamma_E)$, we obtain that
\begin{align}
    \Vert C_{w} \Vert\leqslant Const.\Vert w\Vert_{L^{\infty}(\Gamma_E)}\overset{\eqref{westprop1b}}{=}O\left(\tau^{-\frac{\alpha}{2}+\vert \im \nu\vert}\right), \quad \tau\rightarrow\infty.
\end{align}
Then, by $\alpha\in(\lambda,1)$, $\Vert C_{w} \Vert$ decays to zero as $\tau\rightarrow\infty$. Hence we conclude that $I-C_w$ is invertible for large $\tau$.

By classical Beals-Cofiman theory \cite{BCdecom}, the $E(x,t,k)$ could be represented in terms of a singular integral equation by $w:=J_E-I$ and the standard normalization condition
$E(x,t,k)\rightarrow I$ as $k\rightarrow\infty$, i.e.,
\begin{equation}\label{Efredholmexp}
    E(x,t,k)=I+C(\mu w)=I+\frac{1}{2\pi i}\int_{\Gamma_E}\mu(x,t,s)w(x,t,s)\frac{ds}{s-k},
\end{equation}
where $\mu$ is the solution of the Fredholm-type equation $(I-C_{w})\mu=I$. Moreover, by \eqref{westprop1a} and $\mu-I=(I-C_w)^{-1}C_wI$, we have
\begin{equation}\label{mu-Iest}
    \Vert \mu-I \Vert_{L^2(\Gamma_E)}=O\left(\epsilon^{\frac{1}{2}}\tau^{-\frac{\alpha}{2}+\vert \im \nu\vert}\right), \quad {\rm} \quad \tau\rightarrow\infty.
\end{equation}

\begin{lemma}\label{lemmak(E-I)}
    As $\tau\rightarrow\infty$,
    \begin{equation}
        \lim_{k\rightarrow\infty}k(E(x,t,k)-I)=\mathfrak{B}^r(\xi,t)-\overline{\mathfrak{B}^r(\xi,t)}+R(\xi,t).
    \end{equation}
    where
    \begin{align}\label{mathfrakB^rdefine}
        \mathfrak{B}^r(\xi,t)=\begin{pmatrix}
            0 & i\eta\beta^r(\xi)e^{-t\varphi(\xi,0)}\tau^{-i\nu-\frac{1}{2}} \\
            -i\eta\gamma^r(\xi)e^{t\varphi(\xi,0)}\tau^{i\nu-\frac{1}{2}} & 0
            \end{pmatrix},
    \end{align}
    and $R(\xi,t)=\hat{R}_1(\xi,t)+\hat{R}_2(\xi,t)+\hat{R}_3(\xi,t)$,
    \begin{subequations}
        \begin{align}
            &\hat{R}_{1}(\xi,t)=\begin{pmatrix}
                O(\tau^{-1-\im \nu}) & O(\tau^{-1+\im \nu})
            \end{pmatrix},  &\tau\rightarrow\infty,\\
            &\hat{R}_{2}(\xi,t)=\begin{pmatrix}O\left(\epsilon\tau^{-\frac{1+\alpha}{2}+\vert\im\nu\vert-\im\nu}\right) & O\left(\epsilon\tau^{-\frac{1+\alpha}{2}+\vert\im\nu\vert+\im\nu}\right)\end{pmatrix}, & \tau\rightarrow\infty,\\
            &\hat{R}_{3}(\xi,t)=\begin{pmatrix}O\left(\epsilon\tau^{-\frac{1+\alpha}{2}+\vert\im\nu\vert-\im\nu}\right) & O\left(\epsilon\tau^{-\frac{1+\alpha}{2}+\vert\im\nu\vert+\im\nu}\right)\end{pmatrix}, & \tau\rightarrow\infty.
        \end{align}
    \end{subequations}
    Moreover, recall $\alpha<1$, we have
    \begin{align}
        R(\xi,t):=\begin{pmatrix}O\left(\epsilon\tau^{-\frac{1+\alpha}{2}+\vert\im\nu\vert-\im\nu}\right) & O\left(\epsilon\tau^{-\frac{1+\alpha}{2}+\vert\im\nu\vert+\im\nu}\right)\end{pmatrix}, \ \tau\rightarrow\infty.
    \end{align}
In detail, we set $R(\xi,t)=\begin{pmatrix}
    R_1(\xi,t) & R_2(\xi,t)\\
    R_1(\xi,t) & R_2(\xi,t)
\end{pmatrix}$,
and
\begin{subequations}
    \begin{align}
        &R_1(\xi,t)=\left\{
            \begin{aligned}
            &O(\epsilon\tau^{-\frac{1+\alpha}{2}}), &\im\nu\geqslant 0 \\
            &O(\epsilon\tau^{-\frac{1+\alpha}{2}+2\vert\im\nu\vert}), &\im\nu<0
            \end{aligned}
            \right. \\
        &R_2(\xi,t)=\left\{
            \begin{aligned}
            &O(\epsilon\tau^{-\frac{1+\alpha}{2}+2\vert\im\nu\vert}), &\im\nu\geqslant 0 \\
            &O(\epsilon\tau^{-\frac{1+\alpha}{2}}), &\im\nu<0
            \end{aligned}
            \right.
    \end{align}
\end{subequations}
which also hold true for $\hat{R}_2^{(j)}(\xi,t)$ and $\hat{R}_3^{(j)}(\xi,t)$, $j=1,2$.
\end{lemma}
\begin{proof} Start from \eqref{Efredholmexp}
    \begin{align}\label{k(E-I)firststep}
    &\lim_{k\rightarrow\infty}k(E(x,t,k)-I)=-\frac{1}{2\pi i}\int_{\Gamma_E}\mu(x,t,s)w(x,t,s)ds \nonumber\\
    &=-\frac{1}{2\pi i}\left(\oint_{|s+k_0|=\epsilon}+\oint_{|s-k_0|=\epsilon}\right)\mu(x,t,s)w(x,t,s)ds-\frac{1}{2\pi i}\int_{\Gamma}\mu(x,t,s)w(x,t,s)ds
    \end{align}

Secondly, for the first term of \eqref{k(E-I)firststep}, as $\tau\rightarrow\infty$,
\begin{align}
    &\oint_{|s+k_0|=\epsilon}\mu(x,t,s)w(x,t,s)ds=\oint_{|s+k_0|=\epsilon}w(x,t,s)ds+\oint_{|s+k_0|=\epsilon}(\mu(x,t,s)-I)w(x,t,s)ds \nonumber\\
    &\overset{\eqref{westlemest3}}{=}\frac{\mathfrak{B}_0(\xi,t)}{\sqrt{\tau}}\oint_{|s+k_0|=\epsilon}\frac{1}{s+k_0}ds+\hat{R}_1(\xi,t)+\hat{R}_2(\xi,t)\nonumber\\
    &=-2\pi i\mathfrak{B}^r(\xi,t)+\hat{R}_1(\xi,t)+\hat{R}_2(\xi,t).
\end{align}
and
\begin{align}
    &\oint_{|s-k_0|=\epsilon}\mu(x,t,s)w(x,t,s)ds=\oint_{|s-k_0|=\epsilon}w(x,t,s)ds+\oint_{|s-k_0|=\epsilon}(\mu(x,t,s)-I)w(x,t,s)ds \nonumber\\
    &=\oint_{|s-k_0|=\epsilon}\overline{w(x,t,-\bar{s})}ds+\oint_{|s-k_0|=\epsilon}(\mu(x,t,s)-I)w(x,t,s)ds \nonumber\\
    &=\overline{\oint_{|s+k_0|=\epsilon}w(x,t,s)ds}+\oint_{|s-k_0|=\epsilon}(\mu(x,t,s)-I)w(x,t,s)ds \nonumber\\
    &=\overline{\frac{\mathfrak{B}_0(\xi,t)}{\sqrt{\tau}}\oint_{|s+k_0|=\epsilon}\frac{1}{s+k_0}ds}+\hat{R}_1(\xi,t)+\hat{R}_2(\xi,t)\nonumber\\
    &=2\pi i\overline{\mathfrak{B}^r(\xi,t)}+\hat{R}_1(\xi,t)+\hat{R}_2(\xi,t).
\end{align}
where
\begin{align*}
    \mathfrak{B}^r(\xi,t)=-\frac{\mathfrak{B_0}}{\sqrt{\tau}}=\begin{pmatrix}
        0 & -i\eta\beta^r(\xi)e^{-t\varphi(\xi,0)}\tau^{-i\nu-\frac{1}{2}} \\
        i\eta\gamma^r(\xi)e^{t\varphi(\xi,0)}\tau^{i\nu-\frac{1}{2}} & 0
        \end{pmatrix}
\end{align*}
and
\begin{align}
    \hat{R}_2(\xi,t)&=\Vert \mu-I \Vert_{L^2(\partial U_{\epsilon}(-k_0))}O\left((\mathfrak{B}^r(\xi,t))\right)\nonumber\\
    &=O\left(\epsilon^{\frac{1}{2}}\tau^{-\frac{\alpha}{2}+\vert \im \nu\vert}\right)
    \cdot O\begin{pmatrix}\epsilon^{\frac{1}{2}}\tau^{-\frac{1}{2}-\im\nu} & \epsilon^{\frac{1}{2}}\tau^{-\frac{1}{2}+\im\nu}\end{pmatrix}\nonumber\\
    &=O\begin{pmatrix}\epsilon\tau^{-\frac{1+\alpha}{2}+\vert\im\nu\vert-\im\nu} & \epsilon\tau^{-\frac{1+\alpha}{2}+\vert\im\nu\vert+\im\nu}\end{pmatrix}
\end{align}

Thirdly, for the $j$-th column of second term of \eqref{k(E-I)firststep}, for instance, we consider the $2$-nd column
\begin{align}
    \left\vert\left(\int_{\Gamma}\mu(x,t,s)w(x,t,s)\right)^{(2)}ds\right\vert&\leqslant\left\vert\int_{\Gamma}\left(\mu(x,t,s)-I\right)w^{(2)}(x,t,s)ds\right\vert+\left\vert\int_{\Gamma}w^{(2)}(x,t,s)ds \right\vert\nonumber\\
    &\leqslant \Vert \mu-I \Vert_{L^2(\Gamma)}\Vert w^{(2)} \Vert_{L^2(\Gamma)}+\Vert w^{(2)} \Vert_{L^1(\Gamma)}.
\end{align}
Recall \eqref{westlemest1a}, \eqref{westprop2c}, we have
\begin{subequations}
    \begin{align}
        &\Vert w^{(2)} \Vert_{L^1(\Gamma)}=O\left(\epsilon\tau^{-1}+\epsilon\tau^{-\frac{1+\alpha}{2}+\im\nu}\right),\\
        &\Vert w^{(2)} \Vert_{L^2(\Gamma)}=O\left(\epsilon^{\frac{1}{2}}\tau^{-1}+\epsilon^{\frac{1}{2}}\tau^{-\frac{1}{4}-\frac{\alpha}{2}+\im\nu}\right),
    \end{align}
\end{subequations}
then take into account \eqref{mu-Iest} and obtain
\begin{subequations}
\begin{align}
\left\vert\int_{\Gamma}\left(\mu(x,t,s)w(x,t,s)\right)^{(2)}ds\right\vert=O\left(\epsilon\tau^{-\frac{1+\alpha}{2}+\vert \im\nu\vert+\im\nu}\right).
\end{align}
Similarly, we have
\begin{align}
    \left\vert\int_{\Gamma}\left(\mu(x,t,s)w(x,t,s)\right)^{(1)}ds\right\vert=O\left(\epsilon\tau^{-\frac{1+\alpha}{2}+\vert \im\nu\vert-\im\nu}\right).
\end{align}
\end{subequations}
Then we make
\begin{align}\label{hatR3inproof}
\int_{\Gamma}\mu(x,t,s)w(x,t,s)ds=\begin{pmatrix}O\left(\epsilon\tau^{-\frac{1+\alpha}{2}+\vert\im\nu\vert-\im\nu}\right) & O\left(\epsilon\tau^{-\frac{1+\alpha}{2}+\vert\im\nu\vert+\im\nu}\right)\end{pmatrix}
=:\hat{R}_3(\xi,t), \quad \tau\rightarrow\infty.
\end{align}
Summarize the \eqref{k(E-I)firststep} to \eqref{hatR3inproof}, we complete the proof of this lemma.
\end{proof}

\subsection{Long-time asymptotics}
To recover the potential $u(x,t)$, we mainly refer \cite[pp.722]{RDJDE2021}. From Lemma \ref{lemmak(E-I)}, we see that the main term in the large-negative-$t$ development of $E$ in \eqref{Efredholmexp}
is given by the the integral along the circle $|s+k_0|=\epsilon$ and $|s-k_0|=\epsilon$, which in turn gives
\begin{equation}
    E(x,t,k)=I-\frac{1}{2\pi i}\oint_{|s+k_0|=\epsilon}\frac{\mathfrak{B}^r(\xi,t)}{(s+k_0)(s-k)}ds+\frac{1}{2\pi i}\oint_{|s-k_0|=\epsilon}\frac{\overline{\mathfrak{B}^r(\xi,t)}}{(s-k_0)(s-k)}ds+R(\xi,t),
    \quad |k\pm k_0|>\epsilon,
\end{equation}
By \eqref{defineE}, $E(x,t,k)=\breve{M}^{r}(x,t,k)$ for all $k:\vert k\pm k_0\vert>\epsilon$, thus we have
\begin{equation}\label{k(breveMr-I)}
    \lim_{k\rightarrow\infty}k\left(\breve{M}^r(x,t,k)-I\right)=\lim_{k\rightarrow\infty}k(E(x,t,k)-I)=\mathfrak{B}^r(\xi,t)-\overline{\mathfrak{B}^r(\xi,t)}+R(\xi,t).
\end{equation}
By residue theorem, we additionally have
\begin{subequations}\label{breveMdeyigebiaoshi}
    \begin{align}
        &\breve{M}^r(x,t,0)=I+\frac{\mathfrak{B}^r(\xi,t)}{k_0}+\frac{\overline{\mathfrak{B}^r(\xi,t)}}{k_0}+R(\xi,t), \\
        &\breve{M}^r(x,t,i\kappa)=I+\frac{\mathfrak{B}^r(\xi,t)}{k_0+i\kappa}+\frac{\overline{\mathfrak{B}^r(\xi,t)}}{k_0-i\kappa}+R(\xi,t).
    \end{align}
\end{subequations}
Now we aim at evaluating $P_{12}(x,t)$ as well as $P_{21}(x,t)$ appeared in \eqref{uintermofregularrhp}. Firstly, we consider $g(x,t)=(g_1(x,t),g_2(x,t))^{\rm T}$ and
$h(x,t)=(h_1(x,t), h_2(x,t))^{\rm T}$ defined by \eqref{gandhdefine}, and using \eqref{breveMdeyigebiaoshi}, we obtain that
\begin{subequations}\label{evaluategandh}
    \begin{align}
    \left\{
        \begin{aligned}
        &g_1(x,t)=i\kappa+R_1(\xi,t),\\
        &g_2(x,t)=i\kappa\left(\frac{\mathfrak{B}^r_{21}(\xi,t)}{k_0+i\kappa}+\frac{\overline{\mathfrak{B}^r_{21}(\xi,t)}}{k_0-i\kappa}\right)+R_1(\xi,t),
        \end{aligned}
        \right.
    \end{align}
and
    \begin{align}
        \left\{
        \begin{aligned}
        &h_1(x,t)=c_0(\xi)+\frac{i\kappa}{k_0}\left(\mathfrak{B}_{12}^r(\xi,t)+\overline{\mathfrak{B}_{12}^r(\xi,t)}\right)+R_3(\xi,t),\\
        &h_2(x,t)=i\kappa+\frac{c_0(\xi)}{k_0}\left(\mathfrak{B}_{21}^r(\xi,t)+\overline{\mathfrak{B}_{21}^r(\xi,t)}\right)+R_3(\xi,t),
        \end{aligned}
        \right.
    \end{align}
\end{subequations}
where we have used the facts that $\mathfrak{B}_{11}^r(\xi,t)=\mathfrak{B}_{22}^r(\xi,t)=0$,  algebraic decay term taking up more dominant than exponential decay term when $\tau\rightarrow\infty$,
as well as $R_3(\xi,t)=R_1(\xi,t)+R_2(\xi,t)$ with
\begin{equation}
    R_3(\xi,t)=\left\{
            \begin{aligned}
            &O(\epsilon\tau^{-\frac{1+\alpha}{2}}), &\im\nu=0\\
            &O(\epsilon\tau^{-\frac{1+\alpha}{2}+2\vert\im\nu\vert}), &\im\nu\neq 0
            \end{aligned}
            \right.
\end{equation}
Furthermore, we obtain that
\begin{subequations}\label{gihj}
    \begin{align}
        &g_1h_1=i\kappa c_0(\xi)-\frac{\kappa^2}{k_0}\left(\mathfrak{B}_{12}^r(\xi,t)+\overline{\mathfrak{B}_{12}^r(\xi,t)}\right)+R_3(\xi,t),\\
        &g_1h_2=-\kappa^2+\frac{i\kappa c_0(\xi)}{k_0}\left(\mathfrak{B}_{12}^r(\xi,t)+\overline{\mathfrak{B}_{12}^r(\xi,t)}\right)+R_3(\xi,t),\\
        &g_2h_1=i\kappa c_0(\xi)\left(\frac{\mathfrak{B}^r_{21}(\xi,t)}{k_0+i\kappa}+\frac{\overline{\mathfrak{B}^r_{21}(\xi,t)}}{k_0-i\kappa}\right)+R_1(\xi,t),\\
        &g_2h_2=-\kappa^2\left(\frac{\mathfrak{B}^r_{21}(\xi,t)}{k_0+i\kappa}+\frac{\overline{\mathfrak{B}^r_{21}(\xi,t)}}{k_0-i\kappa}\right)+R_1(\xi,t).
    \end{align}
\end{subequations}
Substitute \eqref{gihj} into \eqref{P12andP21}, we have that
\begin{subequations}\label{P12andP21evaluate}
\begin{align}
&P_{12}(x,t)=-\frac{ic_0}{\kappa}+\frac{1}{k_0}\left(\mathfrak{B}_{12}^r(\xi,t)+\overline{\mathfrak{B}_{12}^r(\xi,t)}\right)+\frac{ic_0^2(\xi)}{k_0\kappa}\left(\frac{\mathfrak{B}^r_{21}(\xi,t)}{k_0+i\kappa}+\frac{\overline{\mathfrak{B}^r_{21}(\xi,t)}}{k_0-i\kappa}\right)+R_3(\xi,t),\\
&P_{21}(x,t)=-\left(\frac{\mathfrak{B}^r_{21}(\xi,t)}{k_0+i\kappa}+\frac{\overline{\mathfrak{B}^r_{21}(\xi,t)}}{k_0-i\kappa}\right)+R_1(\xi,t).
\end{align}
\end{subequations}

We can see that \eqref{P12andP21evaluate} involve $\kappa$, then we should introduce $\mathfrak{B}$ which will make those terms which are dependent on $\kappa$ vanish
in the main asymptotic terms. Introduce
\begin{equation}
    \mathfrak{B}(\xi,t)=\begin{pmatrix}
        0 & -i\eta\beta(\xi)e^{-t\varphi(\xi,0)}\tau^{-i\nu-\frac{1}{2}} \\
        i\eta\gamma(\xi)e^{t\varphi(\xi,0)}\tau^{i\nu-\frac{1}{2}} & 0
        \end{pmatrix},
\end{equation}
and we can see that $\mathfrak{B}(\xi,t)$ is defined similarly to $\mathfrak{B}^r(\xi,t)$, see \eqref{mathfrakB^rdefine}, with
$q^r_j(-k_0)$ is replaced by $q_j(-k_0)$ (i.e.,$r^r_j(-k_0)$ is replaced by $r_j(-k_0)$). In detail,
\begin{subequations}\label{betaandgammaevaluate}
    \begin{align}
        &\beta(\xi)=\frac{\sqrt{2\pi}e^{\frac{i\pi}{4}}e^{-\frac{\pi \nu(-k_0)}{2}}}{{q}_{1}(-k_0)\Gamma(-i\nu(-k_0))}, \quad q_1(-k_0)=e^{-2\chi(\xi,-k_0)}r_1(-k_0)e^{2i\nu\log 4}, \\
        &\gamma(\xi)=\frac{\sqrt{2\pi}e^{-\frac{i\pi}{4}}e^{-\frac{\pi \nu(-k_0)}{2}}}{{q}_{2}(-k_0)\Gamma(i\nu(-k_0))},\quad q_2(-k_0)=e^{2\chi(\xi,-k_0)}r_2(-k_0)e^{-2i\nu\log 4}.
    \end{align}
\end{subequations}
Compare to \eqref{betargammar}, and take into account \eqref{rtor^r}, we obtain that
\begin{align}
    \beta^r(\xi)=\frac{k_0}{k_0+i\kappa}\beta(\xi), \quad \gamma^r(\xi)=\frac{k_0+i\kappa}{k_0}\gamma(\xi),
\end{align}
which imply that
\begin{align}\label{mathfrakBremovekappa}
    \mathfrak{B}^r_{12}(\xi,t)=\frac{k_0}{k_0+i\kappa}\mathfrak{B}_{12}(\xi,t), \quad \mathfrak{B}^r_{21}(\xi,t)=\frac{k_0+i\kappa}{k_0}\mathfrak{B}_{21}(\xi,t).
\end{align}
Substitute \eqref{P12andP21evaluate}, \eqref{mathfrakBremovekappa} and \eqref{k(breveMr-I)} into \eqref{uintermofregularrhp}, we finally obtain the asymptotics
which described by Theorem \ref{mainthm1}.

\section{Asymptotic behavior for $\xi:=x/(12t)>0$}\label{sectionAANO2}
In this section, we investigate the long-time asymptotics under the condition $\xi>0$, and the signature
table of this case corresponds to Fig \ref{FigCasexi>0} with
\begin{equation}\label{imthetaforsectionAANO2}
    \im\theta(k,\xi)=4\im k\left(3\left(\re k\right)^2-\left(\im k\right)^2+3\xi\right).
\end{equation}
Comparing to Section \ref{sectionAANO1}, we point out that the analysis for $t\rightarrow+\infty$
is more convenient to deal with the singularities of RH problem at $k=0$ in this case of $\xi>0$.

\subsection{Asymptotic analysis for $0<\xi<\kappa^2$}\label{subsectionAANO2middle+upkappa}
In this case, $i\sqrt{\xi}<i\sqrt{3\xi}<i\kappa$ for $\xi\in(0,\frac{\kappa^2}{3})$ and $i\sqrt{\xi}<i\kappa<i\sqrt{3\xi}$ for $\xi\in(\frac{\kappa^2}{3},\kappa^2)$.
Our technique is to deform the information of jump contour to other contours. Define two contours $\Gamma_1=\left\{k: k=k_1+i\sqrt{\xi}\right\}$, $\Gamma_1^*=\left\{k: k=k_1-i\sqrt{\xi}\right\}$,
which is parallel to $\mathbb{R}$ and the following domains $U_1=\{k: 0<\im k<i\sqrt{\xi}\}$, $U_1^{*}=\{k: -i\sqrt{\xi}<\im k<0\}$ as Fig \ref{FigCase>0kappaupdeform}, Fig \ref{FigCase>0kappamiddledeform}.
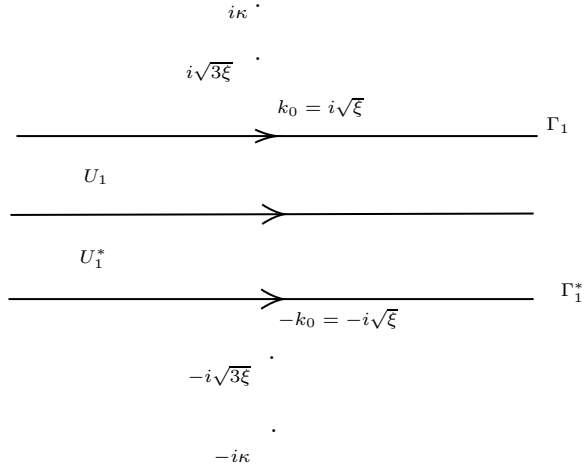
\begin{figure}[htbp]
\begin{center}
    \tikzset{every picture/.style={line width=0.75pt}} 
    \begin{tikzpicture}[x=0.75pt,y=0.75pt,yscale=-1,xscale=1]
    \draw    (199,152) -- (460.01,151.3) ;
    \draw [shift={(335.5,151.63)}, rotate = 179.85] [color={rgb, 255:red, 0; green, 0; blue, 0 }  ][line width=0.75]    (10.93,-4.9) .. controls (6.95,-2.3) and (3.31,-0.67) .. (0,0) .. controls (3.31,0.67) and (6.95,2.3) .. (10.93,4.9)   ;
    \draw    (202,112.33) -- (462,112.33) ;
    \draw [shift={(332,112.33)}, rotate = 180] [color={rgb, 255:red, 0; green, 0; blue, 0 }  ][line width=0.75]    (10.93,-3.29) .. controls (6.95,-1.4) and (3.31,-0.3) .. (0,0) .. controls (3.31,0.3) and (6.95,1.4) .. (10.93,3.29)   ;
    \draw    (198,194.33) -- (460,194.33) ;
    \draw [shift={(335,194.33)}, rotate = 180] [color={rgb, 255:red, 0; green, 0; blue, 0 }  ][line width=0.75]    (10.93,-4.9) .. controls (6.95,-2.3) and (3.31,-0.67) .. (0,0) .. controls (3.31,0.67) and (6.95,2.3) .. (10.93,4.9)   ;
    \draw (330.62,91.48) node [anchor=north west][inner sep=0.75pt]  [font=\scriptsize]  {$k_{0} =i\sqrt{\xi }$};
    \draw (331,197.73) node [anchor=north west][inner sep=0.75pt]  [font=\scriptsize]  {$-k_{0} =-i\sqrt{\xi }$};
    \draw (319,71.4) node [anchor=north west][inner sep=0.75pt]    {$.$};
    \draw (306,45.4) node [anchor=north west][inner sep=0.75pt]  [font=\scriptsize]  {$i\kappa $};
    \draw (326,221.4) node [anchor=north west][inner sep=0.75pt]    {$.$};
    \draw (299,268.4) node [anchor=north west][inner sep=0.75pt]  [font=\scriptsize]  {$-i\kappa $};
    \draw (285,73.4) node [anchor=north west][inner sep=0.75pt]  [font=\scriptsize]  {$i\sqrt{3\xi }$};
    \draw (319,44.4) node [anchor=north west][inner sep=0.75pt]    {$.$};
    \draw (286,226.4) node [anchor=north west][inner sep=0.75pt]  [font=\scriptsize]  {$-i\sqrt{3\xi }$};
    \draw (327,258.4) node [anchor=north west][inner sep=0.75pt]    {$.$};
    \draw (465,101.4) node [anchor=north west][inner sep=0.75pt]  [font=\scriptsize]  {$\Gamma _{1}$};
    \draw (472,184.4) node [anchor=north west][inner sep=0.75pt]  [font=\scriptsize]  {$\Gamma _{1}^{*}$};
    \draw (234,127.4) node [anchor=north west][inner sep=0.75pt]  [font=\scriptsize]  {$U_{1}$};
    \draw (232,167.4) node [anchor=north west][inner sep=0.75pt]  [font=\scriptsize]  {$U_{1}^{*}$};
    \end{tikzpicture}
\caption{\small Contours of RH problem $\breve{N}(x,t,k)$ for $0<\xi<\frac{\kappa^2}{3}$.}\label{FigCase>0kappaupdeform}
\end{center}
\end{figure}

\begin{figure}[htbp]
    \begin{center}
    \tikzset{every picture/.style={line width=0.75pt}} 
    \begin{tikzpicture}[x=0.75pt,y=0.75pt,yscale=-1,xscale=1]
    \draw    (199,152) -- (460.01,151.3) ;
    \draw [shift={(335.5,151.63)}, rotate = 179.85] [color={rgb, 255:red, 0; green, 0; blue, 0 }  ][line width=0.75]    (10.93,-4.9) .. controls (6.95,-2.3) and (3.31,-0.67) .. (0,0) .. controls (3.31,0.67) and (6.95,2.3) .. (10.93,4.9)   ;
    \draw    (202,112.33) -- (462,112.33) ;
    \draw [shift={(332,112.33)}, rotate = 180] [color={rgb, 255:red, 0; green, 0; blue, 0 }  ][line width=0.75]    (10.93,-3.29) .. controls (6.95,-1.4) and (3.31,-0.3) .. (0,0) .. controls (3.31,0.3) and (6.95,1.4) .. (10.93,3.29)   ;
    \draw    (198,194.33) -- (460,194.33) ;
    \draw [shift={(335,194.33)}, rotate = 180] [color={rgb, 255:red, 0; green, 0; blue, 0 }  ][line width=0.75]    (10.93,-4.9) .. controls (6.95,-2.3) and (3.31,-0.67) .. (0,0) .. controls (3.31,0.67) and (6.95,2.3) .. (10.93,4.9)   ;
    \draw (331.62,93.48) node [anchor=north west][inner sep=0.75pt]  [font=\scriptsize]  {$k_{0} =i\sqrt{\xi }$};
    \draw (331,197.73) node [anchor=north west][inner sep=0.75pt]  [font=\scriptsize]  {$-k_{0} =-i\sqrt{\xi }$};
    \draw (319,71.4) node [anchor=north west][inner sep=0.75pt]    {$.$};
    \draw (304,70.4) node [anchor=north west][inner sep=0.75pt]  [font=\scriptsize]  {$i\kappa $};
    \draw (326,221.4) node [anchor=north west][inner sep=0.75pt]    {$.$};
    \draw (301,221.4) node [anchor=north west][inner sep=0.75pt]  [font=\scriptsize]  {$-i\kappa $};
    \draw (289,33.4) node [anchor=north west][inner sep=0.75pt]  [font=\scriptsize]  {$i\sqrt{3\xi }$};
    \draw (319,44.4) node [anchor=north west][inner sep=0.75pt]    {$.$};
    \draw (286,256.4) node [anchor=north west][inner sep=0.75pt]  [font=\scriptsize]  {$-i\sqrt{3\xi }$};
    \draw (327,258.4) node [anchor=north west][inner sep=0.75pt]    {$.$};
    \draw (465,101.4) node [anchor=north west][inner sep=0.75pt]  [font=\scriptsize]  {$\Gamma _{1}$};
    \draw (472,184.4) node [anchor=north west][inner sep=0.75pt]  [font=\scriptsize]  {$\Gamma _{1}^{*}$};
    \draw (234,127.4) node [anchor=north west][inner sep=0.75pt]  [font=\scriptsize]  {$U_{1}$};
    \draw (232,167.4) node [anchor=north west][inner sep=0.75pt]  [font=\scriptsize]  {$U_{1}^{*}$};
    \end{tikzpicture}
    \caption{\small Contours of RH problem $\breve{N}(x,t,k)$ for $\frac{\kappa^2}{3}<\xi<\kappa^2$.}\label{FigCase>0kappamiddledeform}
    \end{center}
    \end{figure}
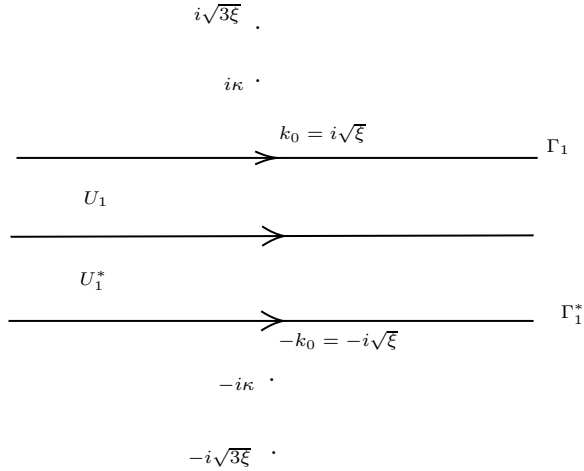
Next, we deform contours (see Fig \ref{FigCase>0kappaupdeform} and Fig \ref{FigCase>0kappamiddledeform}) via defining $\breve{N}(x,t,k)$ by
\begin{equation}\label{breveN}
    \breve{N}(x,t,k)=\left\{
        \begin{aligned}
        &M(x,t,k)
        \begin{pmatrix}
            1 & 0\\
            -r_1(k)e^{2it\theta} & 1
        \end{pmatrix}, & k\in U_1, \\
        &M(x,t,k)
        \begin{pmatrix}
            1 & r_2(k)e^{-2it\theta}\\
            0 & 1
        \end{pmatrix}, & k\in U^*_1, \\
        &M(x,t,k), & {\rm elsewhere}
        \end{aligned}
        \right.
\end{equation}
Then $\breve{N}(x,t,k)$ satisfies the following RH problem
\begin{RHP}\label{RHPbreveN}
    Find a $2\times2$ matrix-valued function $\breve{N}(x,t,k)$ such that
    \begin{itemize}
        \item[(i)] $\breve{N}(x,t,k)$ is meromorphic for $k\in\mathbb{C}\backslash\left\{\left\{0\right\}\cup\Gamma_1\cup\Gamma_1^*\right\}$ and has a simple pole located at $k=i\kappa$, $\kappa>0$.
        \item[(ii)] Jump conditions. The  The non-tangential limits $\breve{N}_{\pm}(x,t,k)=\underset{k'\rightarrow k, k'\in\mathbb{C}_{\pm}}{\lim}\breve{N}(x,t,k')$ exist for
        $k\in\Gamma_1\cup\Gamma_1^*$ and $\breve{N}_{\pm}(x,t,k)$ satisfy
        the jump condition $\breve{N}_{+}(x,t,k)=\breve{N}_{-}(x,t,k){J}_{\breve{N}}(x,t,k)$ for $k\in\Gamma_1\cup\Gamma_1^*$, where
        \begin{equation}\label{jumpforbreveN}
            {J}_{\breve{N}}(x,t,k)=\left\{
                \begin{aligned}
                &\begin{pmatrix}
                    1 & 0\\
                    r_1(k)e^{2it\theta} & 1
                \end{pmatrix}, & k\in\Gamma_1, \\
                &\begin{pmatrix}
                    1 & r_2(k)e^{-2it\theta}\\
                    0 & 1
                \end{pmatrix}, & k\in\Gamma^*_1, \\
                &I, &k\in\mathbb{R}\backslash\left\{0\right\}.
                \end{aligned}
                \right.
        \end{equation}
        \item[(iii)] Normalization condition at $k=\infty$. $\breve{N}(x,t,k)=I+O(k^{-1})$ as $k\rightarrow\infty$.
        \item[(iv)] Residue condition at $k=i\kappa$.
        \begin{equation}\label{rescondition of breveN}
            \underset{k=i\kappa}{\rm Res}\breve{N}^{(1)}(x,t,k)=c_2(x,t)\breve{N}^{(2)}(x,t,i\kappa),
        \end{equation}
        with $c_2(x,t)=\frac{\gamma_0}{a_{1}'(i\kappa)}e^{-2\kappa x+8\kappa^3t}=\frac{\gamma_0}{a_{1}'(i\kappa)}e^{t(-24\kappa\xi+8\kappa^3)}$, $\gamma_0^2=1$.
        \item[(v)] Singularities at the origin. In both cases (Case I and II), as $k\rightarrow 0$, $\breve{N}(x,t,k)$ satisfies

    \begin{subequations}\label{singularity of breveN at k=0 BothCases}
        \begin{align}
            &\breve{N}_{+}(x,t,k)=\left(
                    \begin{array}{cc}
                    -\frac{2i}{A}\bar{v}_2(-x,-t)+O(k) & -\frac{1}{k}\bar{v}_2(-x,-t)+O(1) \\
                    -\frac{2i}{A}\bar{v}_1(-x,-t)+O(k) & -\frac{1}{k}\bar{v}_1(-x,-t)+O(1)
                    \end{array}
                \right), &\mathbb{C}_{+}\ni k\rightarrow 0,\label{singularity of breveN at k=+i0 BothCases}\\
            &\breve{N}_{-}(x,t,k)=\frac{2i}{A}\left(
                \begin{array}{cc}
                    -\bar{v}_2(-x,-t)+O(k) &  -\frac{A}{2ik}\bar{v}_2(-x,-t)+O(k)\\
                    -\bar{v}_1(-x,-t)+O(k) &  -\frac{A}{2ik}\bar{v}_2(-x,-t)+O(k)
                \end{array}
                \right), &\mathbb{C}_{-}\ni k\rightarrow 0.\label{singularity of breveN at k=-i0 BothCases}
        \end{align}
    Furthermore, we can see that the singularity conditions at origin can be reduced to
    \begin{align}\label{singularity of breveN as ResCondition BothCases}
        \underset{k=0}{\rm Res}\breve{N}^{(2)}(x,t,k)=\frac{A}{2i}\breve{N}^{(1)}(x,t,0),
    \end{align}
\end{subequations}
\end{itemize}
\end{RHP}
\begin{proof}
    The proof is similar to RH problem \ref{RHPbreveM}.
\end{proof}
Still using the so-called Blaschke-Potapov factors and following the similar steps of Subsection \ref{subsectionregularRHP}, we can reduce $\breve{N}(x,t,k)$
to an regular RH problem $\breve{N}^r(x,t,k)$ as follows
\begin{RHP}\label{RHPregularN}
    Find a $2\times2$ matrix-valued function $\breve{N}^{r}(x,t,k)$ such that
\begin{itemize}
    \item[(i)] $\breve{N}^{r}(x,t,k)$ is analytic for $k\in\mathbb{C}\backslash(\Gamma_1\cup\Gamma_1^*)$.
    \item[(ii)]  $\breve{N}^{r}_{+}(x,t,k)=\breve{N}^{r}_{-}(x,t,k)J_{\breve{N}^r}(x,t,k)$ with
    \begin{equation}
        J_{\breve{N}^r}(x,t,k)=
        \begin{pmatrix}
            1 & 0\\
            0 & \frac{k-i\kappa}{k}
        \end{pmatrix}
        J_{\breve{N}}(x,t,k)
        \begin{pmatrix}
            1 & 0\\
            0 & \frac{k}{k-i\kappa}
        \end{pmatrix}.
    \end{equation}
    More specifically,
    \begin{equation}\label{jumpforbreveNr}
        J_{\breve{N}^r}(x,t,k)=\left\{
            \begin{aligned}
            &\begin{pmatrix}
                1 & 0\\
                {r}^{r}_1(k)e^{2it\theta} & 1
            \end{pmatrix}, & k\in\Gamma_1, \\
            &\begin{pmatrix}
                1 & {r}^{r}_2(k)e^{-2it\theta}\\
                0 & 1
            \end{pmatrix} , & k\in\Gamma^*_1.
            \end{aligned}
            \right.
    \end{equation}
    with
    \begin{equation*}
        {r}^{r}_1(k)=\frac{k-i\kappa}{k}r_1(k), \quad {r}^{r}_2(k)=\frac{k}{k-i\kappa}r_2(k).
    \end{equation*}
    \item[(iii)] Normalization at $k=\infty$. $\breve{N}^{r}(x,t,k)\rightarrow I$ as $k\rightarrow\infty$.
    \item[(iv)] Matrix-valued factor $P(x,t)$ are determined in terms of $\breve{N}^{r}(x,t,k)$ as \eqref{P12andP21} and \eqref{gandhdefine} via replacing $c_1(x,t)$, $c_0(\xi)$
    by $c_2(x,t)$, $\frac{A}{2i}$ respectively.
\end{itemize}
\end{RHP}

Then we have the potential recovering formulae
\begin{subequations}\label{uintermofregularNrhp}
    \begin{align}
        &u(x,t)=-2\kappa P_{12}(x,t)+2i\lim_{k\rightarrow\infty}k\breve{N}^{r}_{12}(x,t,k), \quad x>0, \ t>0, \label{uintermofregularNrhp1}\\
        &u(x,t)=-2\kappa P_{21}(-x,-t)+2i\lim_{k\rightarrow\infty}k\breve{N}^{r}_{21}(-x,-t,k), \quad x<0, \ t<0, \label{uintermofregularNrhp2}
    \end{align}
\end{subequations}
where $P_{12}(x,t)$, $P_{21}(x,t)$ defined by the item (iv) of RH problem \ref{RHPregular}. For $0<\xi<\frac{\kappa^2}{3}$, and $\frac{\kappa^2}{3}<\xi<\kappa^2$,
different asymptotics are presented.

\subsubsection{Asymptotics for $\frac{\kappa^2}{3}<\xi<\kappa^2$}\label{subsubsectionAANO2middlekappa}
Consider $\breve{N}^r(x,t,k)\approx I$ for $t\rightarrow\infty$, we have
\begin{equation}
    u(x,t)\approx-2\kappa P_{12}(x,t)\approx\frac{2i\kappa^2\cdot\frac{A}{2i}}{\kappa^2-c_2(x,t)\cdot\frac{A}{2i}}\underset{t\rightarrow+\infty}{\approx} A,
\end{equation}
where we used the fact that $c_2(x,t)=\frac{\gamma_0}{a_{1}'(i\kappa)}e^{t(-24\kappa\xi+8\kappa^3)}\rightarrow 0$ as $t\rightarrow\infty$ for $x>0$, $\xi\in(\frac{\kappa^2}{3},\kappa^2)$.
And $A$ is the leading term of $u(x,t)$ for $x>0$, $t\rightarrow+\infty$.

Similarly, we also have
\begin{equation}
    u(x,t)\approx-2\kappa P_{21}(-x,-t)\approx 2\kappa\frac{-c_2(-x,-t)(-i\kappa)}{-\kappa^2+\frac{A}{2i}c_2(-x,-t)}\approx 0.
\end{equation}
which is the leading term for $x<0$, $t\rightarrow-\infty$.

Then as $t\rightarrow+\infty$, use \eqref{imthetaforsectionAANO2}, we easily find  that
\begin{equation}
    J_{\breve{N}^r}(x,t,k)=I+O(e^{-16t\xi^{3/2}}), \quad k\in\Gamma_1\cup\Gamma_1^*, \quad t\rightarrow\infty
\end{equation}
Obey the error analysis similar to subsection \ref{subsectionAANO1erroranalysis} and use \eqref{uintermofregularNrhp}, we obtain the
\begin{subequations}
    \begin{align}
        &u(x,t)=A+O\left(t^{-\frac{1}{2}}e^{-16t\xi^{3/2}}\right), \quad x>0, \ t>0, \ \frac{\kappa^2}{3}<\xi<\kappa^2.\\
        &u(x,t)=O\left((-t)^{-\frac{1}{2}}e^{16t\xi^{3/2}}\right), \quad x<0, \ t<0, \ \frac{\kappa^2}{3}<\xi<\kappa^2,
    \end{align}
\end{subequations}
Which are the \eqref{asyRIM}, \eqref{asyRIIIM} in Theorem \ref{mainthm2} respectively.

\subsubsection{Asymptotics for $0<\xi<\frac{\kappa^2}{3}$}\label{subsubsectionAANO2upkappa}
In this case, $c_2(x,t)$ can not decay $0$ as $t\rightarrow\infty$. Owe to this reason, we make our asymptotics is well-defined in so-called solitonic region.
Still noticing $\breve{N}^r(x,t,k)\approx I$ for $t\rightarrow\infty$, we have
\begin{equation}
    u(x,t)\approx-2\kappa P_{12}(x,t)\approx\frac{2i\kappa^2\cdot\frac{A}{2i}}{\kappa^2-c_2(x,t)\cdot\frac{A}{2i}}\approx\frac{A}{1-C_1(\kappa)e^{-2\kappa x+8\kappa^3t}},
\end{equation}
where $C_1(\kappa)=\frac{A\gamma_0}{2ia_1'(i\kappa)\kappa^2}$ with $\gamma_0^2=1$.

Similarly, we also have
\begin{equation}
    u(x,t)\approx-2\kappa P_{21}(-x,-t)\approx 2\kappa\frac{-c_2(-x,-t)(-i\kappa)}{-\kappa^2+\frac{A}{2i}c_2(-x,-t)}\approx \frac{4}{C_2(\kappa)e^{-2\kappa x+8\kappa^3t}-A\kappa^{-2}},
\end{equation}
where $C_2(\kappa)=\frac{2ia_1'(i\kappa)}{\gamma_0}$ with $\gamma_0^2=1$. And the radiation term is similar to the previous subsubsection \ref{subsubsectionAANO2middlekappa}.
Final asymptotics for this case are described by \eqref{asyRIL} and \eqref{asyRIIIR} in Theorem \ref{mainthm2} respectively.

\subsection{Asymptotic analysis for $\xi>\kappa^2$}\label{subsectionAANO2downkappa}
In this case, $i\kappa<i\sqrt{\xi}<i\sqrt{3\xi}$, we import a new variable $\kappa_{\delta}\in(0,\kappa)$ and redefine contours
$\Gamma_1=\left\{k: k=k_1+i\kappa_{\delta}\right\}$, $\Gamma_1^*=\left\{k: k=k_1-i\kappa_{\delta}\right\}$ and two domains
$U_1=\{k: 0<\im k<i\kappa_{\delta}\}$, $U_1^{*}=\{k: -i\kappa_{\delta}<\im k<0\}$ as Fig \ref{FigCase>0kappadowndeform}.
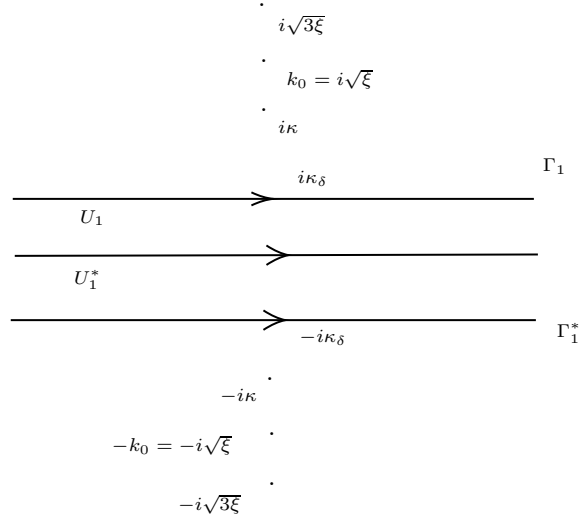
\begin{figure}[htbp]
\begin{center}
\tikzset{every picture/.style={line width=0.75pt}} 
\begin{tikzpicture}[x=0.75pt,y=0.75pt,yscale=-1,xscale=1]
\draw    (174,135) -- (435.01,134.3) ;
\draw [shift={(310.5,134.63)}, rotate = 179.85] [color={rgb, 255:red, 0; green, 0; blue, 0 }  ][line width=0.75]    (10.93,-4.9) .. controls (6.95,-2.3) and (3.31,-0.67) .. (0,0) .. controls (3.31,0.67) and (6.95,2.3) .. (10.93,4.9)   ;
\draw    (173,106.33) -- (433,106.33) ;
\draw [shift={(303,106.33)}, rotate = 180] [color={rgb, 255:red, 0; green, 0; blue, 0 }  ][line width=0.75]    (10.93,-3.29) .. controls (6.95,-1.4) and (3.31,-0.3) .. (0,0) .. controls (3.31,0.3) and (6.95,1.4) .. (10.93,3.29)   ;
\draw    (172,167.33) -- (434,167.33) ;
\draw [shift={(309,167.33)}, rotate = 180] [color={rgb, 255:red, 0; green, 0; blue, 0 }  ][line width=0.75]    (10.93,-4.9) .. controls (6.95,-2.3) and (3.31,-0.67) .. (0,0) .. controls (3.31,0.67) and (6.95,2.3) .. (10.93,4.9)   ;
\draw (313.62,90.48) node [anchor=north west][inner sep=0.75pt]  [font=\scriptsize]  {$i\kappa _{\delta }$};
\draw (221,223.73) node [anchor=north west][inner sep=0.75pt]  [font=\scriptsize]  {$-k_{0} =-i\sqrt{\xi }$};
\draw (298,194.4) node [anchor=north west][inner sep=0.75pt]    {$.$};
\draw (275,200.4) node [anchor=north west][inner sep=0.75pt]  [font=\scriptsize]  {$-i\kappa $};
\draw (254,253.4) node [anchor=north west][inner sep=0.75pt]  [font=\scriptsize]  {$-i\sqrt{3\xi }$};
\draw (299,247.4) node [anchor=north west][inner sep=0.75pt]    {$.$};
\draw (436,84.4) node [anchor=north west][inner sep=0.75pt]  [font=\scriptsize]  {$\Gamma _{1}$};
\draw (443,167.4) node [anchor=north west][inner sep=0.75pt]  [font=\scriptsize]  {$\Gamma _{1}^{*}$};
\draw (205,110.4) node [anchor=north west][inner sep=0.75pt]  [font=\scriptsize]  {$U_{1}$};
\draw (202,140.4) node [anchor=north west][inner sep=0.75pt]  [font=\scriptsize]  {$U_{1}^{*}$};
\draw (314.62,170.48) node [anchor=north west][inner sep=0.75pt]  [font=\scriptsize]  {$-i\kappa _{\delta }$};
\draw (299,222.4) node [anchor=north west][inner sep=0.75pt]    {$.$};
\draw (308,39.73) node [anchor=north west][inner sep=0.75pt]  [font=\scriptsize]  {$k_{0} =i\sqrt{\xi }$};
\draw (294,6.4) node [anchor=north west][inner sep=0.75pt]    {$.$};
\draw (304,65.4) node [anchor=north west][inner sep=0.75pt]  [font=\scriptsize]  {$i\kappa $};
\draw (304,12.4) node [anchor=north west][inner sep=0.75pt]  [font=\scriptsize]  {$i\sqrt{3\xi }$};
\draw (295,59.4) node [anchor=north west][inner sep=0.75pt]    {$.$};
\draw (295,34.4) node [anchor=north west][inner sep=0.75pt]    {$.$};
\end{tikzpicture}
\caption{\small Contours of RH problem $\breve{N}(x,t,k)$ for $\xi>\kappa^2$.}\label{FigCase>0kappadowndeform}
\end{center}
\end{figure}

As the considerations in subsection \ref{subsectionAANO2middle+upkappa}, using \eqref{imthetaforsectionAANO2}, we can find that
\begin{equation}
    J_{\breve{N}^r}(x,t,k)=I+O\left(e^{-8t\kappa_{\delta}(3\xi-\kappa_{\delta}^2)}\right), \quad k\in\Gamma_1\cup\Gamma_1^*, \quad t\rightarrow\infty,
\end{equation}
and derive
\begin{subequations}
    \begin{align}
        &u(x,t)=A+O\left(t^{-\frac{1}{2}}e^{-8t\kappa_{\delta}(3\xi-\kappa_{\delta}^2)}\right), \quad x>0, \ t>0, \ \xi>\kappa^2. \\
        &u(x,t)=O\left((-t)^{-\frac{1}{2}}e^{8t\kappa_{\delta}(3\xi-\kappa_{\delta}^2)}\right), \quad x<0, \ t<0, \ \xi>\kappa^2,
    \end{align}
\end{subequations}
which are presented as \eqref{asyRIR} and \eqref{asyRIIIL} of Theorem \ref{mainthm2} respectively.

\section{Conclusions and further discussions}\label{sectionFinalRemark}
In the present work, we mainly investigate the large-time asymptotics of solution for the Cauchy problem of the nonlinear focusing nonlocal MKdV equation with step-like initial data, i.e.,
$u_0(x)\rightarrow 0$ as $x\rightarrow-\infty$, $u_0(x)\rightarrow A$ as $x\rightarrow+\infty$,
where $A$ is an arbitrary positive real number. We firstly develop the direct scattering theory
to establish the Riemann-Hilbert (RH) problem associated with step-like initial data.
Thanks to the symmetries $x\rightarrow-x$, $t\rightarrow-t$ of nonlocal integrable systems, we investigate the asymptotics for $t\rightarrow-\infty$ and $t\rightarrow+\infty$ respectively.
Our main technique is to use the steepest descent analysis to deform the original matrix-valued RH problem to
corresponded regular RH problem, which could be explicitly solved. Finally we obtain the different long-time asymptotic behaviors
of the solution of the Cauchy problem \eqref{Cauchy Problem} for focusing nonlocal MKdV equation in different
space-time regions $\mathcal{R}_{I}$,  $\mathcal{R}_{II}$,  $\mathcal{R}_{III}$ and $\mathcal{R}_{IV}$ on the whole $(x,t)$-plane.
The asymptotics of $\mathcal{R}_{II}$ and $\mathcal{R}_{IV}$ admit typical Zakhrov-Manakov form, and the subleading term of
the asymptotics in region $\mathcal{R}_{IV}$ depends on the value of $\im\nu(-k_0)$ (see Theorem \ref{mainthm1}). $\mathcal{R}_{I}$ can be divide into three regions, which are
$\mathcal{R}_{I,\mathcal{R}}$, $\mathcal{R}_{I,\mathcal{M}}$ and $\mathcal{R}_{I,\mathcal{L}}$, the former one is the solitonic region and all three take the
exponential orders as error term. Correspondingly, $\mathcal{R}_{III}$ can be also divide into three regions, which are $\mathcal{R}_{III,\mathcal{R}}$,
$\mathcal{R}_{III,\mathcal{M}}$, $\mathcal{R}_{III,\mathcal{L}}$, which admit similar asymptotic types as three regions in $\mathcal{R}_{I}$ (see Theorem \ref{mainthm2}).

In the present work, the symmetry of integrable nonlocal MKdV is the great help to us. When we face the obstacles as $t\rightarrow+\infty$, we turn to
analyze the asymptotics as $t\rightarrow-\infty$, see Section \ref{sectionAANO1}. In the same way, analysis as $t\rightarrow+\infty$ is more convenient than $t\rightarrow-\infty$
in Section \ref{sectionAANO2}.

Finally, we list some problems which remain to be discussed in our forthcoming paper.
\begin{itemize}
    \item[(i)] How to directly analyze for $t\rightarrow+\infty$ in Section \ref{sectionAANO1} ? In this case, main difficulties have been described in Remark \ref{whywedot-infty}.
    Besides, after opening lens, we have to deal with the singularities at $k=0$ under the second jump factorization with $\frac{r_{j}}{1+r_1r_2}$ instead of $r_j$ when we directly investigate $t\rightarrow+\infty$,
    which can not be easily to convert the singularity conditions of $k=0$ into the residue condition as \eqref{singularity of breveM as ResCondition BothCases}.

    \item[(ii)]Asymptotics of transitions regions between $x>0$ and $x<0$ are still open. From the point of the RH problem formalism, the transition region between corresponds to
    the saddle points $\pm k_0$ and singularities $k=0$. I hold the view that the main difficult comes from the origin ponint. For the case of $x>0$, $\xi<0$, we see that the leading term
    $A\delta^2(\xi,0)$ increasing oscillations as $-k_0\rightarrow 0^{-}$, $k_0\rightarrow 0^+$ respectively. We notice that in \cite{RDStudies}, Rybalko and Shepelsky give a good solvable
    results along the curved wedges in the long-time asymptotics for the integrable nonlinear nonlocal Schr\"odinger equation for this problem. We look forward to extending similar results to our
    forthcoming work for nonlocal MKdV equation.

    \item[(iii)]Does Painlav\'e-type asymptotics exist or not? Different from the local MKdV with step-like initial data, there exists a singularity at $k=0$ in $\mathbb{C}$ rather than a
    cut $[-ic,ic]$ in $\mathbb{C}$ as in \cite{KotMinakovJMP}. To some extent, besides the singular point $k=0$, I think the properties of nonlocal MKdV with step-like intial values are
    more close to the local MKdV equation with decay boundary condtions such as in \cite{DZAnn}, which take Painlev\'e-type asymptotics in some transition regions. That's our motivation to
    ask this quesiton.
\end{itemize}

\section*{Acknowledgements}
This work is strongly supported by the National Science Foundation of China (Grant No. 11671095, 51879045).
\hspace*{\parindent}

\appendix
\renewcommand\thefigure{\Alph{section}\arabic{figure}}
\setcounter{figure}{0}
\renewcommand{\theequation}{\thesection.\arabic{equation}}
\section{Proof of Proposition \ref{kappaexpression}}\label{kappaCalculation}
{\bf Case I: generic case}

Define functions $\tilde{a}_1(k)$ and $\tilde{a}_2(k)$ by
\begin{equation}
\tilde{a}_1(k)=a_1(k)\frac{k^2}{(k-i\kappa)(k+i)}, \quad \tilde{a}_2(k)=a_2(k)\frac{k-i\kappa}{k-i}.
\end{equation}
Via the determinant relation (item (iv) in Proposition \ref{propspecfuncs}) and asymptotic of $a_{j}(k)$ as $k\rightarrow\infty$ (item (iii) in Proposition \ref{propspecfuncs}), we can construct
the following scalar RH problem with respect to $\tilde{a}_j(k)$, $j=1,2$.
\begin{itemize}
    \item $\tilde{a}_1(k)$ and $\tilde{a}_2(k)$ are analytic and have no zeros in $\overline{\mathbb{C}_{+}}$ and $\overline{\mathbb{C}_{-}}$ respectively.
    \item $\tilde{a}_j(k)$, $j=1,2$ admit the jump condition
    \begin{equation}
        \tilde{a}_1(k)\tilde{a}_2(k)=\frac{k^2}{1+k^2}\left(1-b^2(k)\right), \quad k\in\mathbb{R}.
    \end{equation}
    \item $\tilde{a}_j(k)\rightarrow 1$ as $k\rightarrow\infty$.
\end{itemize}
The unique solution of the scalar RH problem is given by
\begin{equation*}
    \tilde{a}_1(k)=e^{\chi(k)}, \quad \tilde{a}_2(k)=e^{-\chi(k)}
\end{equation*}
where
\begin{equation}
    \chi(k)=\frac{1}{2\pi i}\int_{-\infty}^{\infty}\frac{{\rm log}\frac{s^2}{1+s^2}\left(1-b^2(s)\right)}{s-k}ds.
\end{equation}
Moreover, we have
\begin{equation}\label{a1a2CaseI}
    a_1(k)=\frac{(k+i)(k-i\kappa)}{k^2}e^{\chi(k)}, \quad a_2(k)=\frac{k-i}{k-i\kappa}e^{-\chi(k)}.
\end{equation}
Near $k=0$, we obtain
\begin{equation}\label{CaseIchi0}
    a_1(k)=\frac{\kappa}{k^2}e^{\chi(+i0)}\left(1+o(k)\right), \quad a_2(0)=\frac{1}{k}e^{-\chi(-i0)},
\end{equation}
where $\pm i0$ represent that $k$ tends to $0$ from the positive side and negative side respectively. By Sokhotski-Plemelj formulas, we have that
\begin{equation}\label{CaseIPlemelj}
    \chi(+i0)+\chi(-i0)=\frac{1}{i\pi}{\rm p.v.}\int_{-\infty}^{\infty}\frac{{\rm log}\frac{s^2}{1+s^2}\left(1-b^2(s)\right)}{s}ds.
\end{equation}

Recall the item (v) in the Proposition \ref{propspecfuncs},
\begin{equation}\label{CaseIa1k=0}
    a_1(k)=\frac{A^2}{4k^2}a_2(0)(1+o(k)), \quad k\rightarrow 0.
\end{equation}
Compare \eqref{CaseIchi0} and \eqref{CaseIa1k=0}, also use \eqref{CaseIPlemelj}, we reach \eqref{kappageneric}.

{\bf Case II: non-generic case}

Define functions $\breve{a}_1(k)$ and $\breve{a}_2(k)$ by
\begin{equation}
\breve{a}_1(k)=a_1(k)\frac{k}{k-i\kappa}, \quad \breve{a}_2(k)=a_2(k)\frac{k-i\kappa}{k}.
\end{equation}
Similar to the claims in Case I, $\breve{a}_j(k)$, $j=1,2$ satisfy the following scalar RH problem
\begin{itemize}
    \item $\breve{a}_1(k)$ and $\breve{a}_2(k)$ are analytic and have no zeros in $\overline{\mathbb{C}_{+}}$ and $\overline{\mathbb{C}_{-}}$ respectively.
    \item $\breve{a}_j(k)$, $j=1,2$ admit the jump condition
    \begin{equation}
        \breve{a}_1(k)\breve{a}_2(k)=1-b^2(k), \quad k\in\mathbb{R}.
    \end{equation}
    \item $\breve{a}_j(k)\rightarrow 1$ as $k\rightarrow\infty$.
\end{itemize}
We can solve that
\begin{subequations}\label{a1a2CaseII}
    \begin{align}
        & a_1(k)=\frac{k-i\kappa}{k}\exp\left\{\frac{1}{2\pi i}\int_{-\infty}^{\infty}\frac{{\rm log}\left(1-b^2(s)\right)}{s-k}ds\right\}, \\
        & a_2(k)=\frac{k}{k-i\kappa}\exp\left\{-\frac{1}{2\pi i}\int_{-\infty}^{\infty}\frac{{\rm log}\left(1-b^2(s)\right)}{s-k}ds\right\}.
    \end{align}
\end{subequations}
Apply Sokhotski-Plemelj formulae to \eqref{a1a2CaseII}, we find
\begin{equation}\label{a11,a2'atk=0}
    a_{11}=-i\kappa I_1I_2, \quad a_2'(0)=i\kappa^{-1}I^{-1}_1I_2,
\end{equation}
where $I_1$, $I_2$ are defined by \eqref{I1I2}.
Additionally, we can see that $a_{11}a_2'(0)=I_2^2=1-b^2(0)\neq 0$.

On the other hand, due to the symmetry relation \eqref{symb} and the item (v) of Proposition \ref{propofpsi}, the behavior of $\psi_j(x,t,k)$, $j=1,2$ as $k\rightarrow 0$ could be
exhibited as follows
\begin{subequations}
        \begin{align}
            &\psi_{1}^{(1)}(x,t,k)=\frac{1}{k}\left(
                \begin{array}{cc}
                v_{1}(x,t)\\
                v_{2}(x,t)
                \end{array}
            \right)+
            \left(
                \begin{array}{cc}
                s_{1}(x,t)\\
                s_{2}(x,t)
                \end{array}
            \right)+O(k), & k\rightarrow 0,\\
            &\psi_{1}^{(2)}(x,t,k)=\frac{2i}{A}\left(
                \begin{array}{cc}
                v_{1}(x,t)\\
                v_{2}(x,t)
                \end{array}
            \right)+k
            \left(
                \begin{array}{cc}
                h_{1}(x,t)\\
                h_{2}(x,t)
                \end{array}
            \right)+O(k^2), & k\rightarrow 0, \\
            &\psi_{2}^{(1)}(x,t,k)=-\frac{2i}{A}\left(
                \begin{array}{cc}
                \overline{v_{2}(-x,-t)}\\
                \overline{v_{1}(-x,-t)}
                \end{array}
            \right)-k
            \left(
                \begin{array}{cc}
                \overline{h_{2}(-x,-t)}\\
                \overline{h_{1}(-x,-t)}
                \end{array}
            \right)+O(k^2), & k\rightarrow 0,\\
            &\psi_{2}^{(2)}(x,t,k)=-\frac{1}{k}\left(
                \begin{array}{cc}
                \overline{v_{2}(-x,-t)}\\
                \overline{v_{1}(-x,-t)}
                \end{array}
            \right)+
            \left(
                \begin{array}{cc}
                \overline{s_{2}(-x,-t)}\\
                \overline{s_{1}(-x,-t)}
                \end{array}
            \right)+O(k), & k\rightarrow 0
        \end{align}
\end{subequations}
with some functions $v_j$, $h_j$ and $s_j$, $j=1,2$. Use determinant representation for spectral functions (see \eqref{detrepforspecfunc}) as $k\rightarrow 0$ and take into account that
$\vert v_2(0,0) \vert^2-\vert v_1(0,0) \vert^2=0$ under the Case II (see Remark \ref{remarkv200-v100}), we have
\begin{subequations}\label{accuratespecfuncatk=0}
    \begin{align}
        &a_1(k)=\frac{1}{k}(v_1\bar{s}_1-\bar{v}_1s_1-v_2\bar{s}_2+\bar{v}_2s_2)|_{x=0,t=0}+O(1), &k\rightarrow 0,\\
        &a_2(k)=k\frac{2i}{A}(v_1\bar{h}_1+\bar{v}_1h_1-v_2\bar{h}_2-\bar{v}_2h_2)|_{x=0,t=0}+O(k^2), &k\rightarrow 0,\\
        &b(k)=v_1\bar{h}_1-v_2\bar{h}_2+\frac{2i}{A}(v_1\bar{s}_1-v_2\bar{s}_2)|_{x=0,t=0}+O(k), &k\rightarrow 0.
    \end{align}
\end{subequations}
From \eqref{accuratespecfuncatk=0}, we find that (Recall that $a_{11}=\lim_{k\rightarrow 0}ka_1(k)$)
\begin{equation*}
    b(0)+\overline{b(0)}=\frac{A}{2i}a_2'(0)-\frac{2i}{A}a_{11},
\end{equation*}
which is equivalent to (notice $b(k)$=$\bar{b}(-k)$)
\begin{equation}\label{AppendixCala_{11}CaseII}
    a_{11}=iAb(0)-\frac{A^2}{4}a_2'(0).
\end{equation}
Combine \eqref{a11,a2'atk=0}, we arrive at \eqref{kappanongeneric}.

\section{Parabolic Cylinder Model}\label{Appendixpcmodel}
    Find a matrix-valued function $m^{pc}_{-k_0}(\zeta):=m^{pc}_{-k_0}(\xi;\zeta)$ with following properties:
    \begin{itemize}
    \item[*] $m^{pc}_{-k_0}(\zeta)$ is analytical  in $\mathbb{C}\backslash \Sigma^{pc} $ with $\Sigma^{pc}$ shown in Fig \ref{FigSigma-k_0};
    \item[*] $m^{pc}_{-k_0}$ has continuous boundary values $m^{pc}_{-k_0, \pm}$ on $\Sigma^{pc}$ and
    \begin{equation}
    m^{pc}_{-k_0, +}(\zeta)=m^{pc}_{-k_0,-}(\zeta)J^{pc}(\zeta),\quad \zeta \in \Sigma^{pc},
    \end{equation}
    where
    \begin{align}\label{jumppc}
    J^{pc}(\zeta)=\left\{\begin{array}{ll}
    \zeta^{i\nu \hat{\sigma}_3}e^{-\frac{i\zeta^2}{4}\hat{\sigma}_3}\left(\begin{array}{cc}
    1 & 0\\
    q^r_1(-k_0) & 1
    \end{array}\right),  & \zeta\in\Sigma _{3}^{pc},\\[10pt]
    \zeta^{i\nu \hat{\sigma}_3}e^{-\frac{i\zeta^2}{4}\hat{\sigma}_3}\left(\begin{array}{cc}
    1 & -q^r_2(-k_0)\\
    0 & 1
    \end{array}\right),   & \zeta\in\Sigma _{3}^{pc*} ,\\[10pt]
    \zeta^{i\nu \hat{\sigma}_3}e^{-\frac{i\zeta^2}{4}\hat{\sigma}_3}\left(\begin{array}{cc}
    1& \frac{q^r_2(-k_0)}{1+q^r_1(-k_0)q^r_2(-k_0)}\\
    0 & 1
    \end{array}\right),  &\zeta\in\Sigma _{4}^{pc} \\[10pt]
    \zeta^{i\nu \hat{\sigma}_3}e^{-\frac{i\zeta^2}{4}\hat{\sigma}_3}\left(\begin{array}{cc}
    1 & 0 \\
    -\frac{q^r_1(-k_0)}{1+q^r_1(-k_0)q^r_2(-k_0)} & 1
    \end{array}\right),   & \zeta\in\Sigma _{4}^{pc*}.
    \end{array}\right.
    \end{align}
    and $\nu=\nu(-k_0)$.
    \item[*] $m^{pc}_{-k_0}(\zeta)= I+m^{pc}_{-k_0,1}\zeta^{-1}+O(\zeta^{-2})$, $\zeta \rightarrow \infty$.
    \end{itemize}

    The RH problem $m^{pc}_{-k_0}(\zeta)$ has an explicit solution, which can be expressed in terms of Webber equation
    $(\frac{\partial^2}{\partial z^2}+(\frac{1}{2}-\frac{z^2}{2}+a))D_{a}(z)=0$.

    \begin{figure}[htbp]
    \centering
    \tikzset{every picture/.style={line width=0.75pt}} 
    \begin{tikzpicture}[x=0.75pt,y=0.75pt,yscale=-1,xscale=1]
    \draw    (306.01,137.16) -- (380.29,60.02) ;
    \draw [shift={(347.31,94.27)}, rotate = 133.92] [color={rgb, 255:red, 0; green, 0; blue, 0 }  ][line width=0.75]    (10.93,-3.29) .. controls (6.95,-1.4) and (3.31,-0.3) .. (0,0) .. controls (3.31,0.3) and (6.95,1.4) .. (10.93,3.29)   ;
    \draw    (379.29,212.02) -- (306.01,137.16) ;
    \draw [shift={(338.45,170.3)}, rotate = 45.61] [color={rgb, 255:red, 0; green, 0; blue, 0 }  ][line width=0.75]    (10.93,-3.29) .. controls (6.95,-1.4) and (3.31,-0.3) .. (0,0) .. controls (3.31,0.3) and (6.95,1.4) .. (10.93,3.29)   ;
    \draw    (228.29,61.02) -- (306.01,137.16) ;
    \draw [shift={(271.43,103.29)}, rotate = 224.41] [color={rgb, 255:red, 0; green, 0; blue, 0 }  ][line width=0.75]    (10.93,-3.29) .. controls (6.95,-1.4) and (3.31,-0.3) .. (0,0) .. controls (3.31,0.3) and (6.95,1.4) .. (10.93,3.29)   ;
    \draw    (306.01,137.16) -- (232.29,215.02) ;
    \draw [shift={(265.02,180.45)}, rotate = 313.43] [color={rgb, 255:red, 0; green, 0; blue, 0 }  ][line width=0.75]    (10.93,-3.29) .. controls (6.95,-1.4) and (3.31,-0.3) .. (0,0) .. controls (3.31,0.3) and (6.95,1.4) .. (10.93,3.29)   ;
    \draw    (215.29,138.02) -- (405.29,136.05) ;
    \draw [shift={(408.29,136.02)}, rotate = 179.41] [fill={rgb, 255:red, 0; green, 0; blue, 0 }  ][line width=0.08]  [draw opacity=0] (10.72,-5.15) -- (0,0) -- (10.72,5.15) -- (7.12,0) -- cycle    ;
    \draw (300.99,145.9) node [anchor=north west][inner sep=0.75pt]  [font=\scriptsize,rotate=-0.03]  {$0$};
    \draw (386,50.4) node [anchor=north west][inner sep=0.75pt]  [font=\scriptsize]  {$\Sigma _{3}^{pc}$};
    \draw (384,206.4) node [anchor=north west][inner sep=0.75pt]  [font=\scriptsize]  {$\Sigma _{3}^{pc*}$};
    \draw (204,212.4) node [anchor=north west][inner sep=0.75pt]  [font=\scriptsize]  {$\Sigma _{4}^{pc*}$};
    \draw (201,50.4) node [anchor=north west][inner sep=0.75pt]  [font=\scriptsize]  {$\Sigma _{4}^{pc}$};
    \draw (357,99.4) node [anchor=north west][inner sep=0.75pt]  [font=\scriptsize]  {$U_{2}^{pc}$};
    \draw (363,156.4) node [anchor=north west][inner sep=0.75pt]  [font=\scriptsize]  {$U_{2}^{pc*}$};
    \draw (233,99.4) node [anchor=north west][inner sep=0.75pt]  [font=\scriptsize]  {$U_{3}^{pc}$};
    \draw (234,160.4) node [anchor=north west][inner sep=0.75pt]  [font=\scriptsize]  {$U_{3}^{pc*}$};
    \end{tikzpicture}
    \caption{\small Contours of RH problem $m^{pc}_{-k_0}$.}\label{FigSigma-k_0}
    \end{figure}

    Taking the transformation
    \begin{equation}\label{importm_0}
        m^{pc}_{-k_0}=m_0(\zeta)\mathcal{P}\zeta^{-i\nu\sigma_3}e^{\frac{i}{4}\zeta^2\sigma_3},
    \end{equation}
    where
    \begin{align}
        \mathcal{P}(\xi)=\left\{\begin{array}{ll}
        \left(\begin{array}{cc}
        1 & 0\\
        -q^r_1(-k_0) & 1
        \end{array}\right),  & \zeta\in U_{2}^{pc} ,\\[10pt]
        \left(\begin{array}{cc}
        1 & q^r_2(-k_0)\\
        0 & 1
        \end{array}\right),   & \zeta\in U_{2}^{pc*} ,\\[10pt]
        \left(\begin{array}{cc}
        1 & -\frac{q^r_2(-k_0)}{1+q^r_1(-k_0)q^r_2(-k_0)}\\
        0 & 1
        \end{array}\right),   & \zeta\in U_{3}^{pc} ,\\[10pt]
        \left(\begin{array}{cc}
        1& 0\\
        \frac{q^r_1(-k_0)}{1+q^r_1(-k_0)q^r_2(-k_0)} & 1
        \end{array}\right),   & \zeta\in U_{3}^{pc*} ,\\[10pt]
        I,   & elsewhere.
        \end{array}\right.
    \end{align}
    The  matrix-valued function $m_0(\zeta)$ satisfies the following properties
        \begin{itemize}
            \item[*] $m_0(\zeta)$ is analytical in $\mathbb{C}\backslash\mathbb{R}$;
            \item[*] $m_0(\zeta)$ takes continuous boundary values $m_{0,\pm}(\zeta)$ on $\mathbb{R}$ and
            \begin{equation}\label{m_0jump}
                m_{0,+}(\zeta)=m_{0,-}(\zeta)J_0, \quad \zeta\in\mathbb{R},
            \end{equation}
            where
            \begin{equation}
                J_0(\xi)=\left(\begin{array}{cc}
                1+q^r_1(-k_0)q^r_2(-k_0) & q^r_2(-k_0)\\
                q^r_1(-k_0) & 1
                \end{array}\right).
            \end{equation}
            \item[*] asymptotic behavior:
            \begin{equation}\label{pcpsiasy}
                m_0(\zeta)=\left(I+m^{pc}_{-k_0,1}\zeta^{-1}+\mathcal{O}(\zeta^{-2})\right)\zeta^{i\nu\sigma_3}e^{-\frac{i}{4}\zeta^2\sigma_3}, \quad {\rm as} \quad \zeta\rightarrow\infty.
            \end{equation}
        \end{itemize}

        Differentiating \eqref{m_0jump} with respect to $\zeta$,
        and combining $\frac{i\zeta}{2}\sigma_3m_{0,+}=\frac{i\zeta}{2}\sigma_3m_{0,-}J_0$, we obtain
        \begin{equation}
            \left(\frac{dm_0}{d\zeta}+\frac{i\zeta}{2}\sigma_3m_0\right)_{+}=\left(\frac{dm_0}{d\zeta}+\frac{i\zeta}{2}\sigma_3m_0\right)_{-}J_0.
        \end{equation}
        It's not difficult to verify the matrix function
        $\left(\frac{dm_0}{d\zeta}+\frac{i\zeta}{2}\sigma_3m_0\right)m_0^{-1}$ has no jump along the real axis and is an entire function
        with respect to $\zeta$. Combine \eqref{importm_0}, we can directly calculate that
        \begin{equation}\label{dpsiequeta}
            \left(\frac{dm_0}{d\zeta}+\frac{i\zeta}{2}\sigma_3m_0\right)m_0^{-1}=
            \left[\frac{dm^{pc}_{-k_0}}{d\zeta}+m^{pc}_{-k_0}\frac{i\nu}{\zeta}\sigma_3\right]\left(m^{pc}_{-k_0}\right)^{-1}+\frac{i\zeta}{2}\left[\sigma_3,m^{pc}_{-k_0}\right]\left(m^{pc}_{-k_0}\right)^{-1},
        \end{equation}
        The first term in the R.H.S of \eqref{dpsiequeta} tends to zero as $\zeta\rightarrow\infty$. We use $m^{pc}_{-k_0}=I+m^{pc}_{-k_0,1}\zeta^{-1}+O(\zeta^{-2})$
        as well as Liouville theorem to obtain that there exists a constant matrix $\beta^{mat}$ such that
        \begin{equation}
            \beta^{mat}:=\left(\begin{array}{cc}
                0 & {\beta}^r(\xi)\\
                {\gamma}^r(\xi) & 0
                \end{array}\right)
                =\frac{i}{2}\left[\sigma_3, m^{pc}_{-k_0,1}\right]=\left(\begin{array}{cc}
            0 & i[m^{pc}_{-k_0,1}]_{12}\\
            -i[m^{pc}_{-k_0,1}]_{21} & 0
            \end{array}\right).
        \end{equation}
        which implies that $[m^{pc}_{-k_0,1}]_{12}=-i{\beta}^r(\xi)$, $[m^{pc}_{-k_0,1}]_{21}=i{\gamma}^r(\xi)$.
        Use Liouville theorem again, we have
        \begin{equation}
            \left(\frac{dm_0}{d\zeta}+\frac{i\zeta}{2}\sigma_3m_0\right)=\beta^{mat}m_0.
        \end{equation}
        We rewrite the above equality to the following ODE systems
        \begin{align}
            \frac{dm_{0,11}}{d\zeta}+\frac{i\zeta}{2}m_{0,11}=\beta^r(\xi)m_{0,21},\label{pcODE1}\\
            \frac{dm_{0,21}}{d\zeta}-\frac{i\zeta}{2}m_{0,21}={\gamma}^r(\xi)m_{0,11},\label{pcODE2}
        \end{align}
        as well as
        \begin{align}
            \frac{dm_{0,12}}{d\zeta}+\frac{i\zeta}{2}m_{0,12}={\beta}^r(\xi)m_{0,22},\label{pcODE3}\\
            \frac{dm_{0,22}}{d\zeta}-\frac{i\zeta}{2}m_{0,22}={\gamma}^r(\xi)m_{0,12},\label{pcODE4}.
        \end{align}
        From \eqref{pcODE1} to \eqref{pcODE4}, we can solve that
        \begin{align}
            \frac{d^2m_{0,11}}{d\zeta^2}+\left(\frac{i}{2}+\frac{\zeta^2}{4}-{\beta}^r(\xi){\gamma}^r(\xi)\right)m_{0,11}=0,\label{pcpsi11}\quad
            \frac{d^2m_{0,21}}{d\zeta^2}+\left(-\frac{i}{2}+\frac{\zeta^2}{4}-{\beta}^r(\xi){\gamma}^r(\xi)\right)m_{0,21}=0,\\
            \frac{d^2m_{0,12}}{d\zeta^2}+\left(\frac{i}{2}+\frac{\zeta^2}{4}-{\beta}^r(\xi){\gamma}^r(\xi)\right)m_{0,12}=0,\quad
            \frac{d^2m_{0,22}}{d\zeta^2}+\left(-\frac{i}{2}+\frac{\zeta^2}{4}-{\beta}^r(\xi){\gamma}^r(\xi)\right)m_{0,22}=0.
        \end{align}
        We set $\nu={\beta}^r(\xi){\gamma}^r(\xi)$. For $m_{0,11}$, $\im\zeta>0$ we introduce the new variable $\tilde{\eta}=\zeta e^{-\frac{3i\pi}{4}}$, and
        the first equation of \eqref{pcpsi11} becomes
        \begin{equation}
            \frac{d^2m_{0,11}}{d\tilde{\eta}^2}+\left(\frac{1}{2}-\frac{\tilde{\eta}^2}{4}+i\nu\right)m_{0,11}=0.
        \end{equation}
        For $\zeta\in\mathbb{C}^{+}$, $0<{\rm Arg}\zeta<\pi$, $-\frac{3\pi}{4}<{\rm Arg}\tilde{\eta}<\frac{\pi}{4}$.
        We have $m_{0,11}=e^{-\frac{3\pi}{4}\nu(-k_0)}D_{i\nu(-k_0)}(e^{-\frac{3\pi}{4} i}\zeta)\sim\zeta^{i\nu}e^{-\frac{i}{4}\zeta^2}$
        corresponding to the $(1,1)$-entry of \eqref{pcpsiasy}. To save the space, we present the other results for $m_0$ below.

        The unique solution to  is\\
        when $\zeta\in\mathbb{C}^+$,
        \begin{align}
        m_0(\zeta)=\left(\begin{array}{cc}
        e^{-\frac{3\pi}{4}\nu(-k_0)}D_{i\nu(-k_0)}(e^{-\frac{3\pi}{4} i}\zeta) & -\frac{i\nu(-k_0)}{{\gamma^r}(\xi)}e^{\frac{\pi}{4}(\nu(-k_0)-i)}D_{-i\nu(-k_0)-1}(e^{-\frac{\pi i}{4} }\zeta)\\
        \frac{i\nu(-k_0)}{{\beta^r}(\xi)}e^{-\frac{3\pi}{4}(\nu(-k_0)+i)}D_{i\nu(-k_0)-1}(e^{-\frac{3\pi i}{4}}\zeta) & e^{\frac{\pi}{4}\nu(-k_0)}D_{-i\nu(-k_0)}(e^{-\frac{\pi}{4} i}\zeta)
        \end{array}\right).
        \end{align}
        when $\zeta\in\mathbb{C}^-$,
        \begin{align}
        m_0(\zeta)=\left(\begin{array}{cc}
        e^{\frac{\pi\nu(-k_0)}{4}}D_{i\nu(-k_0)}(e^{\frac{\pi}{4} i}\zeta) & -\frac{i\nu(-k_0)}{{\gamma^r}(\xi)}e^{-\frac{3\pi}{4}(\nu(-k_0)-i)}D_{-i\nu(-k_0)-1}(e^{\frac{3\pi i}{4} }\zeta)\\
        \frac{i\nu(-k_0)}{{\beta^r}(\xi)}e^{\frac{\pi}{4}(\nu(-k_0)+i)}D_{i\nu(-k_0)-1}(e^{\frac{\pi i}{4} }\zeta) & e^{-\frac{3\pi}{4}\nu(-k_0)}D_{-i\nu(-k_0)}(e^{\frac{3\pi}{4} i}\zeta)
        \end{array}\right),
        \end{align}
        Which is derived in \cite{lenellsmkdv}.

        From \eqref{m_0jump}, we know that $(m_{0,-})^{-1}m_{0,+}=J_0$ and
        \begin{align}
        q^r_1(-k_0)&=m_{0,-,11}m_{0,+,21}-m_{0,-,21}m_{0,+,11}\nonumber\\
        &=e^{\frac{\pi}{4}\nu(-k_0)}D_{i\nu(-k_0)}(e^{\frac{\pi}{4} i}\zeta)\cdot \frac{e^{-\frac{3\pi\nu(-k_0)}{4}}}{\beta^r(\xi)}\left[\partial_{\zeta}(D_{i\nu(-k_0)}(e^{-\frac{3\pi i}{4}}\zeta))+\frac{i\zeta}{2}D_{i\nu(-k_0)}(e^{-\frac{3\pi i}{4}}\zeta)\right]\nonumber\\
        &\hspace{1em}-e^{-\frac{3\pi}{4}\nu(-k_0)}D_{i\nu(-k_0)}(e^{-\frac{3\pi}{4} i}\zeta)\cdot \frac{e^{\frac{\pi\nu(-k_0)}{4}}}{{\beta^r}(\xi)}\left[\partial_{\zeta}(D_{i\nu(-k_0)}(e^{\frac{\pi i}{4}}\zeta))+\frac{i\zeta}{2}D_{i\nu(-k_0)}(e^{\frac{\pi i}{4}}\zeta)\right]\nonumber\\
        &=\frac{e^{-\frac{\pi}{2}\nu(-k_0)}}{{\beta^r}(\xi)}{\rm Wr}\left(D_{i\nu(-k_0)}(e^{\frac{\pi}{4} i}\zeta), D_{i\nu(-k_0)}(e^{-\frac{3\pi}{4} i}\zeta)\right)\nonumber\\
        &=\frac{e^{-\frac{\pi}{2}\nu(-k_0)}}{{\beta^r}(\xi)}\cdot\frac{\sqrt{2\pi}e^{\frac{\pi}{4}i}}{\Gamma(-i\nu(-k_0))}.
        \end{align}

        And
        \begin{align}
            {\beta^r}(\xi)=\frac{\sqrt{2\pi}e^{\frac{i\pi}{4}}e^{-\frac{\pi \nu(-k_0)}{2}}}{{q}^r_{1}(-k_0)\Gamma(-i\nu(-k_0))}, \quad
            {\gamma^r}(\xi)=\frac{\sqrt{2\pi}e^{-\frac{i\pi}{4}}e^{-\frac{\pi \nu(-k_0)}{2}}}{{q}^r_{2}(-k_0)\Gamma(i\nu(-k_0))}
        \end{align}
        Finally we have
        \begin{equation}\label{mpcxieta0}
            m^{pc}_{-k_0}=I+\frac{i}{\zeta}
                \left(\begin{array}{cc}
                0 & -{\beta^r}(\xi)\\
                {\gamma^r}(\xi) & 0
                \end{array}\right)+\mathcal{O}(\zeta^{-2}).
        \end{equation}

\end{document}